\documentclass{article}[11pt]

\usepackage[dvipsnames]{xcolor}
\usepackage{cite,color,xspace}
\usepackage{latexsym}
\usepackage{array}
\usepackage{mathtools}
\usepackage{todonotes}
\usepackage{amsmath}
\usepackage{amsfonts}
\usepackage{amsthm}
\usepackage{epsfig}
\usepackage{cite}
\usepackage{url}
\usepackage{graphicx}
\usepackage{fancyheadings}
\usepackage{caption}
\usepackage{subcaption}
\usepackage{algorithm}
\usepackage[noend]{algpseudocode}
\usepackage{fullpage}
\usepackage[matha]{mathabx}

\newcommand{\set}[1]{\left\{ #1 \right\}}

\newcommand{\cd}{\text{ :- }}

\newcommand{\setof}[2]{\left\{ #1 \mid #2 \right\}}

\newcommand{\Dom}{\text{\sf Domain}}

\newcommand{\AGM}{\textsf{AGM}}
\newcommand{\GLVV}{\textsf{GLVV}}
\newcommand{\FD}{\textsf{FD}}
\newcommand{\Light}{\textsf{Lite}}
\newcommand{\Heavy}{\textsf{Heavy}}
\newcommand{\coloring}{{\cal L}}

\newcommand{\onehat}{\hat{1}}
\newcommand{\zerohat}{\hat{0}}

\newcommand{\querylabels}{\textsf{Labels}}

\newcommand{\co}{\text{co}}

\newcommand{\vars}{\text{vars}}

\newcommand{\calA}{\mathcal A}
\newcommand{\calB}{\mathcal B}

\newcommand{\calP}{\mathcal P}

\newcommand{\calS}{\mathcal S}

\newcommand{\bR}{{\mathbf R}}
\newcommand{\bJ}{{\mathbf J}}

\newcommand{\bX}{{\mathbf X}}
\newcommand{\bC}{{\mathbf C}}
\newcommand{\bL}{{\mathbf L}}

\newcommand{\incomp}{\not\lessgtr}

\newcommand{\eat}[1]{}

\newcommand{\flow}{{\sf netflow}}

\newcommand{\lelp}{{\sf LLP}}
\newcommand{\cllp}{{\sf CLLP}}

\newcommand{\opt}{{\textsf{OPT}}}
\newcommand{\obj}{{\textsf{OBJ}}}
\newcommand{\Deg}{{\textsf{degree}}}

\newcommand{\R}{\mathbb{R}}

\newcommand{\argmin}{\mathop{\text{argmin}}}

\newcommand{\be}{\begin{enumerate}}
\newcommand{\ee}{\end{enumerate}}
\newcommand{\bi}{\begin{itemize}}
\newcommand{\ei}{\end{itemize}}
\newcommand{\beq}{\begin{equation}}
\newcommand{\eeq}{\end{equation}}

\newcommand{\bp}{\begin{proof}}
\newcommand{\ep}{\end{proof}}
\newcommand{\bcor}{\begin{cor}}
\newcommand{\ecor}{\end{cor}}
\newcommand{\bthm}{\begin{thm}}
\newcommand{\ethm}{\end{thm}}
\newcommand{\blmm}{\begin{lmm}}
\newcommand{\elmm}{\end{lmm}}
\newcommand{\bdefn}{\begin{defn}}
\newcommand{\edefn}{\end{defn}}
\newcommand{\bprop}{\begin{prop}}
\newcommand{\eprop}{\end{prop}}
\newcommand{\bconj}{\begin{conj}}
\newcommand{\econj}{\end{conj}}
\newcommand{\bopm}{\begin{opm}}
\newcommand{\eopm}{\end{opm}}
\newcommand{\brmk}{\begin{rmk}}
\newcommand{\ermk}{\end{rmk}}

\newcommand{\suchthat}{\ | \ }

\theoremstyle{plain}                   % default
\newtheorem{thm}{Theorem}[section]
\newtheorem{lmm}[thm]{Lemma}
\newtheorem{prop}[thm]{Proposition}
\newtheorem{cor}[thm]{Corollary}

\theoremstyle{definition}              % Examples and all

\newtheorem{opm}[thm]{Open Problem}
\newtheorem{conj}[thm]{Conjecture}
\newtheorem{example}[thm]{Example}

\newtheorem{defn}[thm]{Definition}

\newtheorem{rmk}[thm]{Remark}

\newtheorem*{example*}{Example}
\newcommand{\continued}[1]{{\bf \ref{#1} Continued}}

\newcommand{\bbox}{
\begin{center}
\begin{tabular}{|c|}
\hline
}
\newcommand{\ebox}{
\\
\hline
\end{tabular}
\end{center}
}

\newlength{\toppush}
\def\subjnum{CSE 713}
\def\subjname{Probabilistically Checkable Proofs and Inapproximability}
\def\doheading#1#2#3{\vfill\eject\vspace*{-\toppush}%
  \vbox{\hbox to\textwidth{{\bf}
  \subjnum: \subjname
  \hfil Lecturer: Hung Q. Ngo}%
  \hbox to\textwidth{{\bf} SUNY at Buffalo, Fall 2004\hfil#3\strut}%
  \hrule}
}

\newcommand{\defeq}{\stackrel{\mathrm{def}}{=}}

\algrenewcommand\algorithmicrequire{\textbf{Input:}}
\algrenewcommand\algorithmicensure{\textbf{Output:}}
\algrenewcommand\algorithmicwhile{\textbf{While}}
\algrenewcommand\algorithmicfor{\textbf{For}}
\algrenewcommand\algorithmicreturn{\textbf{Return}}
\algrenewcommand\algorithmicif{\textbf{If}}
\newcommand{\Lm}{{{D}}}
\newcommand{\Ln}{{{E}}}
\newcommand{\Lp}{{{F}}}
\newcommand{\Lmn}{{{G}}}
\newcommand{\Lmp}{{{I}}}
\newcommand{\Lnp}{{{J}}}
\newcommand{\Lx}{{{M}}}
\newcommand{\Ly}{{{N}}}
\newcommand{\Lz}{{{O}}}
\newcommand{\Lmnp}{{{Z}}}
\newcommand{\La}{{{P}}}
\newcommand{\Lb}{{{S}}}
\newcommand{\Lc}{{{T}}}
\newcommand{\Lab}{{{U}}}
\newcommand{\Lac}{{{V}}}
\newcommand{\Lbc}{{{W}}}

\allowdisplaybreaks[1]

\usepackage{hyperref}

\begin{document}

\title{Computing Join Queries with Functional Dependencies}

\author{
   {\sf Mahmoud Abo Khamis}\\
   LogicBlox, Inc.\\
   \& SUNY Buffalo
   \and
   {\sf Hung Q. Ngo}\\
   LogicBlox, Inc.\\
   \& SUNY Buffalo
   \and
   {\sf Dan Suciu}\\
   LogicBlox, Inc.\\
   \& University of Washington
}

\date{}

\maketitle

\thispagestyle{empty}

\begin{abstract}
%!TEX root = main.tex

Recently, Gottlob, Lee, Valiant, and Valiant (GLVV) presented an output size bound for join queries with functional dependencies (FD), based on a linear program on polymatroids. GLVV bound strictly generalizes the bound of Atserias, Grohe and Marx (AGM) for queries with no FD, in which case there are known algorithms running within AGM bound and thus are worst-case optimal.

A main result of this paper is an algorithm for computing join queries with FDs, running within GLVV bound up to a poly-log factor. In particular, our algorithm is worst-case optimal for any query where the GLVV bound is tight. As an unexpected by-product, our algorithm manages to solve a harder problem, where (some) input relations may have prescribed maximum degree bounds, of which {\em both} functional dependencies and cardinality bounds are special cases.

We extend Gottlob et al. framework by replacing all variable subsets with the lattice of closed sets (under the given FDs). This gives us new insights into the structure of the worst-case bound and worst-case instances. While it is still open whether GLVV bound is tight in general, we show that it is tight on distributive lattices and some other simple lattices. Distributive lattices capture a strict superset of queries with no FD and with simple FDs. We also present two simpler algorithms which are also worst-case optimal on distributive lattices within a single-$\log$ factor, but they do not match GLVV bound on a general lattice. Our algorithms are designed based on a novel principle: we turn a proof of a polymatroid-based output size bound into an algorithm.

\end{abstract}

\newpage

%!TEX root = main.tex

\section{Introduction}

\label{sec:intro}

Several results published in the last ten years or so have lead to tight
worst-case output size bounds ({\em AGM-bound} \cite{AGM,GM06}) 
and the development of a new class of 
query processing algorithms running within the bound's time 
budget~\cite{DBLP:conf/pods/NgoPRR12,LFTJ,skew}.  These new {\em
worst-case optimal} algorithms are quite different from traditional
query plans, in the sense that they no longer compute one pairwise join at a
time, but instead process the query globally.  The runtime is bounded
by $O(N^{\rho^*})$, where $N$ is the size of the database, while
$\rho^*$ is the value of the optimal fractional edge cover of the
query.  For example, they compute the triangle query $Q(x,y,z) \cd
R(x,y), S(y,z), T(z,x)$, in worst-case time $O(N^{3/2})$, while
any traditional query plan requires worst-case time
$\Omega(N^2)$~\cite{skew}.

While the vast majority of database engines today still rely on
traditional query plans, new, complex data analytics engines
increasingly switch to worst-case optimal algorithms: LogicBlox'
engine~\cite{logicblox} is built on a worst-case optimal algorithm
called {\sf LeapFrog Triejoin}~\cite{LFTJ} (LFTJ), the {\em Myria}
data analytics platform supports a variant of
LFTJ~\cite{DBLP:conf/sigmod/ChuBS15}, and the increased role of SIMD
instructions in modern processors also favors the new class of
worst-case optimal algorithms~\cite{DBLP:journals/corr/AbergerNOR15}.

When there are functional dependencies (FDs), however, the AGM-bound
is no-longer tight.  Gottlob, Lee, Valiant, and Valiant~\cite{GLVV}
initiated the study of the worst-case output size of a query in the
presence of FDs and described an upper bound (reviewed in
Section~\ref{sec:background}), called the {\em GLVV-bound} henceforth.  It
remains open whether the GLVV bound is tight; at the present, it is
the best known upper bound for the query size in the presence of
FDs.  A recent result by Gogacz and
Toru{\'{n}}czyk~\cite{DBLP:journals/corr/GogaczT15} proved a weaker
statement, namely that if we replace
the polymatroidal constraints in the GLVV-bound 
with entropic constraints (of which there are infinitely many),
then the bound is tight.

% While the GLVV-bound is always tighter, and can be unboundedly 
% tighter than the AGM-bound for join queries with FDs, it is not known whether 
% the GLVV-bound is tight in general.
% (There was a recent attempt~\cite{DBLP:journals/corr/GogaczT15} at showing
% tightness of the GLVV-bound.)

Our paper proposes a novel approach to studying queries in the
presence of FDs, by using lattice theory.  We present several
theoretical results that clarify precisely where queries with FDs
become more difficult to study than those without.  We describe a
an algorithm which runs in time proportional to the
GLVV-bound, within a polylogarithmic factor.  We also
describe two special cases where the polylogarithmic factor is
reduced to a single $\log$.  Before discussing our results, we explain their
significance.

\subsection{Motivation}

\label{subsec:motivation}

\paragraph*{User-Defined Functions.} UDF's, or interpreted predicates,
significantly affect the runtime of a query. Consider:
\begin{equation}
  Q(x,y,z,u) \text{ :- } R(x,y), S(y,z), T(z,u), u = f(x,z), x = g(y,u).
   \label{eqn:running}
\end{equation}
The predicates $u=f(x,z)$ and $x=g(y,u)$ represent two user-defined
functions, $f$ and $g$; for example $f(x,z)$ could be $x+z$, or the
concatenation of two strings $x$ and $z$.  Any UDF can be modeled as a
relation with a primary key, for example the function $f$ can be
viewed as a relation $F(x,z,u)$ of cardinality $N^2$ (with one entry
for every $x,z$ pair) satisfying the FD $xz \rightarrow u$; similarly
$g$ introduces the FD $yu \rightarrow x$.  The addition of these two
FD's significantly affects the output size and the query evaluation
complexity:
% Two na\"ive ways to deal with UDFs is to think of them as very large
% input relations, or to compute the query without them and filter the
% output against them afterward.  In the former strategy, for
% instance, we can think of $u=f(x,z)$ as a relation $F(u,x,z)$ of
% size $N^2$, where $N$ is the maximum size of the effective domains
% of $x$ and $z$.
if we first computed the intermediate query $R(x,y),S(y,z),T(z,u)$
then applied the two predicates $u = f(x,z)$ and $x = g(y,u)$ then the
runtime can be as high as $N^2$ because the size of the intermediate query
is $N^2$ when $R=\setof{(i,1)}{i\in [N]}$, $S=\set{(1,1)}$,
$T=\setof{(1,i)}{i\in [N]}$.  However, as we will show, the GLVV bound
for the size of $Q$ is $\leq N^{3/2}$. Our new algorithms run in
$\tilde O(N^{3/2})$, and thus reduce the running time asymptotically.
As shall be seen, this runtime is also worst-case optimal.

A related problem is querying relations with restricted access
patterns~\cite{DBLP:journals/pvldb/BenediktLT15}.  In that setting,
some of the relations in the database can only be read by providing
the values of one or more attributes.  As shown above, a user defined
function $f(x,z)$ can be modeled as an infinite relation $F(u,x,z)
\equiv (u = f(x,z))$ with the restriction that $F$ can be accessed
only by providing inputs for the variables $x$ and $z$.  The work on
querying under restricted access patterns has focused on the
answerability question (whether a query can or cannot be answered).
Our work extends that, by finding algorithms for
answering the query within GLVV-bound.

\paragraph*{Known Frequencies.} Systems often know an upper bound
on the frequencies (or degrees) in the database.  For example,
consider the triangle query above, and assume that the binary graph
defined by the relation $R(x,y)$ has a bounded degree: all outdegrees
are $\text{deg}(x) \leq d_1$ and all indegrees are $\text{deg}(y) \leq
d_2$.  One can model this scenario by introducing an artificial color
$c_1$ on outgoing edges, and a color $c_2$ for incoming edges of $R$:
\begin{align}
  Q(x,y,z,c_1,c_2) \cd & R(x,c_1,c_2,y),S(y,z),T(z,x),C_1(c_1),C_2(c_2),\nonumber\\
   & xc_1\rightarrow y, yc_2 \rightarrow x, xy \rightarrow c_1c_2
   \label{eqn:degree-bound}
\end{align}
where $|R|=|S|=|T|=N$ and $|C_1|= d_1, |C_2|=d_2$.  Thus, each
outgoing edge from $x$ is colored with some distinct color
$c_1$, and similarly each incoming edge to $y$ is colored with some
distinct color $c_2$.  The new predicates $C_1,C_2$ limit the number
of colors to $d_1, d_2$ respectively.  We will show that the worst
case query output decreases from $N^{3/2}$ to
$\min(N^{3/2},Nd_1,Nd_2)$
Alternatively, we can use the linear program and algorithm in
Sec~\ref{sec:csma} to capture queries with known maximum degree
bounds.  We note that another approach to evaluate a limited class of
queries over databases with known degrees has been recently described
in~\cite{DBLP:journals/corr/JoglekarR15}.

\subsection{Overview of the Results}
\label{subsec:overview}

Grohe and Marx~\cite{GM06} and later Atserias,
Grohe and Marx~\cite{AGM} derived an elegant tight upper bound on the 
output size of a join query: $\prod_{j=1}^m |R_j|^{w_j}$,
where $R_1, \dots, R_m$ are the input relations to the query, and
$(w_j)_{j=1}^m$ is any fractional edge cover of the query's
hypergraph.  This bound is known today as the AGM bound.  A simple
example that gives a great intuition for this formula, due to
Grohe~\cite{grohe2013bounds}, is the following.  Suppose we choose a subset of
relations $R_{j_1}, R_{j_2}, \ldots$ that together contain all variables of the
query, in other words they form an integral cover of the query's
hypergraph. Then, obviously $|Q|$ is upper bounded by the product
$|R_{j_1}| |R_{j_2}| \cdots$, because the output to $Q$ is contained
in the cross product of these relations.  The AGM bound simply
generalizes this property from integral to fractional edge covers.
They proved that the bound is tight 
by describing a simple database instance for any query,
such that the query's output matches the bound.  
In that instance, every relation is a cross
product of sets, one set per variable; we call it a {\em product
instance}.

An obvious open question was whether a query $Q$ can be evaluated on any 
database instance $D$ in time that is no larger than the AGM bound of $Q$ 
on databases with the same cardinalities as $D$; such an algorithm is 
called {\em worst-case optimal}.  Ngo, Porat, R{\'e}, and 
Rudra~\cite{DBLP:conf/pods/NgoPRR12} described the
first worst-case optimal algorithm;
later Veldhuizen~\cite{LFTJ} proved that LFTJ, an algorithm already implemented
at LogicBlox earlier, is also worst-case optimal.
A survey and unification of these two algorithms can be found in~\cite{skew}.

Neither the AGM bound nor the associated algorithms analytically exploit FDs in
the database.\footnote{Algorithmically, LFTJ handles FDs by binding variables 
at the earliest trie level at which they are functionally determined.  
For example, in $R(x,y),S(y,z),T(z,u),u=f(x,z),x=g(y,u)$ with a key order 
$[x,y,z,u]$, whenever $z$ was bound the value of $u$ would be immediately 
computed by $u=f(x,z)$.}  
Such FDs can provably reduce the worst-case output of a
query, but the upper bound and algorithms mentioned above cannot use
this information, and instead treat the query by ignoring the FDs.
Gottlob et al.~\cite{GLVV} studied the upper bound of the query size
in the presence of FDs, and established two classes of results.  The
first was a characterization of this bound in the case when the FD's
are restricted to simple keys; as we will show, this case
can be solved entirely using the AGM bound by simply
replacing each relation with its closure.  Next, they described a
novel approach to reasoning about the worst-case output of a query,
using information theory.  They viewed the query output as a
multivariate probability space, and introduced two constraints on the
marginal entropies: a cardinality constraint for each input
relation $R$, stating that the entropy of its variables cannot exceed
the uniform entropy $H(\vars(R))\leq \log_2 |R|$, and one constraint
for each FD $X\rightarrow Y$, stating $H(XY) = H(X)$. (Note that $X$ and $Y$
are sets of variables.) The largest
answer to the query $Q$ is then given by the largest possible value of
$2^{H(\vars(Q))}$, over all choices of entropic functions $H$ that
satisfy these constraints.  But characterizing the space of all
entropic functions $H$ is a long standing open problem in information
theory \cite{Yeung:2008:ITN:1457455}; 
to circumvent that, they relax the function $H$ by
allowing it to be any function that satisfies Shannon
inequalities. Such a function is called a {\em polymatroid} in the
literature, and we denote it with lower case $h$ to distinguish it from
entropic functions $H$.  Thus, the problem in~\cite{GLVV} can be
stated equivalently as: find the maximum value $h(\vars(Q))$
where $h$ ranges over all polymatroids satisfying the given
constraints. 

In this paper, we continue the study of query evaluation under general FDs.
We establish both bounds and algorithms.  Our novelty is to model FDs
as a lattice $\bL$ of the closed sets of attributes, and to study
polymatroids on lattices.  The function $h$ is now any non-negative,
monotone, sub-modular function 
(i.e. $h(X)+h(Y) \geq h(X \vee Y) + h(X \wedge Y)$)
that satisfies all cardinality constraints. FD constraints are
enforced {\em automatically} by the lattice structure and the upper bound on
the query size is $2^{\max h(\vars(Q))}$. When there are no FDs, the
lattice is a Boolean algebra.

Our first question is whether the elegant AGM bound and worst-case
product instance carries over to arbitrary FDs.  We answer this
question completely, by proving that both sides of the AGM bound hold
iff the lattice has a special structure, which we call a {\em normal
  lattice}.  Both upper and lower bounds require minor extensions to
be applicable to normal lattices.  The standard AGM upper bound is
given in terms of fractional edge covers of the query's hypergraph,
but in a normal lattice one needs to consider a dual hypergraph, whose
nodes are $L$'s co-atoms. In a Boolean algebra, these two hypergraphs
are isomorphic, because of the bijection $X \mapsto (\vars(Q) -
\set{X})$ between variables and co-atoms, but in a general lattice
they can be significantly different.  
The notion of normal lattice seems novel, and
strictly includes all distributive lattices, which in turn include all
lattices corresponding to simple FDs (each FD is of the form $a
\rightarrow b$, where $a,b$ are attributes).  The second minor change
is that one needs to allow for a slight generalization of product
instances, to what we call quasi-product instances.  Importantly, both
these properties fail on non-normal lattices; in particular, worst-case 
instances cannot be quasi-free.
% , which explains the need for the
% complex construction in~\cite{DBLP:journals/corr/GogaczT15}.

The canonical example of a non-normal lattice is $M_3$ (one of the two
canonical non-distributive lattices, see the right part of
Fig.~\ref{fig:non:normal}).  
Every lattice $L$ that contains
$M_3$ as a sublattice such that $\max L = \max M_3$ is non-normal; we
conjecture that the converse also holds.   Interestingly, the other
canonical non-distributive lattice $N_5$ is normal.

Next, we examine algorithms whose runtime is bounded by the GLVV bound.  
We propose a
novel methodology for designing such algorithms, starting from the
observation that such an algorithm must provide a proof of the query's
upper bound, equivalently, a proof of an inequality of the form
$\sum_j w_j h(\vars(R_j)) \geq h(\vars(Q))$, where $R_1,
R_2, \ldots$ are the input relations.  In the case of a Boolean
algebra, Shearer's lemma~\cite{MR859293} is of this form; in a normal lattice this 
corresponds to a fractional edge cover of the co-atomic hypergraph; and, for a
general lattice it is a general inequality. 
A key motivation behind 
NPRR~\cite{DBLP:conf/pods/NgoPRR12} was to prove inequalities algorithmically.
This paper completes the cycle by proceeding in the opposite direction: 
given a proof method for such inequalities, we
design algorithms whose steps correspond to the proof steps.  We
design three such algorithms, corresponding to three methods for
proving the above type of inequalities.

The first algorithm called the {\em chain algorithm} runs within the {\em
chain bound}; the proof technique is adapted from 
Radhakrishnan's proof~\cite{radhakrishnan} of Shearer's lemma to general
lattices. Both the chain bound and algorithm strictly generalize AGM-bound and
worst-case optimal algorithms for join queries {\em without} FDs.
The second algorithm, called the {\em sub-modular algorithm}, runs within
the {\em sub-modularity bound}; the proof technique is that of
Balister and
Bollob{\'{a}}s's~\cite{DBLP:journals/combinatorica/BalisterB12}.
The third algorithm, called the {\em conditional sub-modularity algorithm}
(CSMA) runs within the general {\em GLVV bound},
up to a polylogarithmic factor; the proof technique is our own, based on linear
programming duality.
In addition to being able to achieve the most general bound, CSMA can be
used straightforwardly to handle input relations with known maximum degree
bounds.
We remark that GLVV bound is stronger than both the other two bounds. We
show that they are tight, and thus our algorithms are worst-case optimal,
for distributive lattices.

%\subsection{Overview of Techniques}
%\label{subsec:techniques}
%
%\yell{TBD: present here chain bound, chain algo, SMA algo, for the running 
%example; in effect move some examples here. This was requested by the reviews, 
%and it should make the paper more ``believable''. We won't need extra space,
%because we will be moving stuff from the body here. }

\paragraph*{Outline.} The paper is organized as follows.  Background material
is reviewed in Sec.~\ref{sec:background}, and basic definitions for
our lattice-based approach are given in Sec.~\ref{sec:lattice}.  We
describe the main result on normal lattices in Sec.~\ref{sec:normal},
then present our three algorithms and bounds in Sec.~\ref{sec:proof:algorithms}.

%!TEX root = main.tex

\section{Notations and Prior Results}
\label{sec:background}

For any positive integer $n$, $[n]$ denotes the set $\{1,\dots,n\}$.
All $\log$ in the paper are of base $2$.
We fix a relational schema $\bR = \{R_1,\ldots,R_m\}$ whose attributes
belong to a set of attributes $\bX = \{x_1, \ldots, x_k\}$.  We refer
to $x_i \in \bX$ interchangeably as an {\em attribute} or a {\em
  variable}; similarly we refer to $R_j$ as a {\em relation}, or an
{\em input}.  We use lower case letters $x \in \bX$ to denote single
variables, and upper case letters $X \subseteq \bX$ to denote sets of
variables. The {\em domain} of variable $x$ is denoted by $\Dom(x)$. 
For each relation $R_j$ we denote 
$\vars(R_j) = X_j \subseteq \bX$ its
set of attributes, and sometimes blur the distinction between $R_j$
and $X_j$, writing, with some abuse, $R_j \subseteq \bX$.  We consider
full conjunctive queries without self-joins:
\begin{align}
  Q(x_1, \ldots, x_k) \cd R_1(X_1), \ldots, R_m(X_m) \label{eq:q}
\end{align}
%
% Self-joins are omitted because they can be handled using the Chase as
% in~\cite{GLVV}. 
We will drop variables from the head, since it is understood that all
variables need to be listed.

A database instance $D$ consists of one relational instance $R_j^D$
for each relation symbol; we denote $N_j = |R_j^D|$,
$N=|D| = \sum_j N_j$, and use lower case for logs, $n_j = \log_2 N_j$.
We denote $Q^D$ the answer to $Q$ on the database instance $D$.  A
{\em product database instance} is a database instance such that
$R_j^D = \prod_{x_i \in R_j} \Dom(x_i)$ for $j\in [m]$;
the query answer on a product database is the cross product of all domains,
$Q^D = \prod_{i=1}^{k} \Dom(x_i)$.

\paragraph*{The AGM Bound.} A series of results over the last ten
years~\cite{GM06,AGM,DBLP:conf/pods/NgoPRR12,skew,LFTJ} have
established tight connections between the maximum output size of a
query and the runtime of a query evaluation algorithm.  The {\em query
  hypergraph} of a query $Q$ is $H_Q = (\bX, \bR)$; its nodes are the
variables and its hyperedges are the input relations (where each $R_j
\in \bR$ is viewed as a set of variables).  Consider the following two
linear programs (LP's), called weighted fractional edge cover LP and
vertex packing LP, respectively:
%
%\begin{align*}
%& \text{(Weighted) Edge Cover} && \text{(Weighted) Ver. Packing} \\
%& \text{minimize} \sum_j w_jn_j && \text{maximize} \sum_i v_i \\
%\forall i: & \sum_{j: x_i \in R_j} w_j \geq 1 & \forall j: & \sum_{i:  x_i \in
%R_j} v_i \leq n_j\\
%\forall j: & \ w_j \geq 0 & \forall i: & \ v_i \geq 0
%\end{align*}
\begin{equation*}
   \begin{array}{lrlll@{\hskip 0.5in}lrlll}
      \multicolumn{5}{l}{\text{(Weighted) Fractional Edge Cover}} &\multicolumn{5}{l}{\text{(Weighted) Fractional Vertex Packing}} \\
      \text{minimize} &\multicolumn{1}{l}{\sum_j w_jn_j} &    &    &&\text{maximize} & \multicolumn{1}{l}{\sum_i v_i} \\
                      & \sum_{j: x_i \in R_j} w_j        &\geq& 1, & \forall i\in [k] &&   \sum_{i:  x_i \in R_j} v_i &\leq& n_j,&\forall j \in [m]\\
                      & w_j                              &\geq& 0, & \forall j
   \in [m]&& v_i &\geq& 0, &\forall i \in [k] \end{array} \end{equation*}
We call a feasible solution to the first LP a {\em fractional edge
cover}, and to the second a {\em fractional vertex packing}.  
The traditional (unweighted) notions correspond to $n_j=1, \forall j$.

\begin{thm}[AGM
  bound]\cite{GM06,AGM} \label{th:agm}
  (1) Let $(w_j)_{j=1}^m$ be a fractional edge cover.  Then, for any
  input database $D$ s.t. $|R^D_j| \leq N_j$ for all $j\in[m]$, the
  output size of $Q$ is bounded by $2^{\sum_j w_jn_j}$. In other
  words, $|Q^D| \leq \prod_j N_j^{w_j}$.  (2) Let $v_i, i \in [k]$ be
  a fractional vertex packing, and let $D$ be the product database
  instance where $D_i = [2^{v_i}]$.  Then, $|Q^D| = \prod_i 2^{v_i}$.
\end{thm}

(In the statement above, we ignore the issue of integrality of the $v_i$
for the sake of clarity. There is a bit of loss when $v_i$ are not integers,
but this fact does not affect the asymptotics of the lowerbound \cite{AGM}.)
By strong duality, these two LPs have the same optimal objective
value, denoted by $\rho^*(Q,(N_j)_{j=1}^m)$~\cite{grohe2013bounds}.
% \footnote{It is standard to
% use $\rho^*$ for the edge covering number: $\rho^* = \min \sum_j w_j$;
% following Grohe~\cite{grohe2013bounds} we
% write $\rho^*(Q,(N_j)_{j=1}^m) = \min \sum_j w_jn_j$ for the weighted
% version.}  
Let $(w_j^*)_{j=1}^m$ be the optimal edge cover.  The {\em AGM bound}
of the query $Q$ is 
\[ \AGM(Q,(N_j)_{j=1}^m) = 2^{\rho^*(Q,(N_j)_{j=1}^m)} = \prod_j N_j^{w_j^*},\] or
just $\AGM(Q)$ when the cardinalities $(N_j)_{j=1}^m$ are clear from
the context.  The query's output size is always $\leq \AGM(Q)$, and
this bound is tight, because on the product database described above,
the output is $\AGM(Q)$.  It is easy to check that $\AGM(Q) = \min_w
\prod_j N_j^{w_j}$, where $w$ ranges over the vertices of the edge
cover polytope.  For example, for $Q=R(x,y),S(y,z),T(z,x)$ the edge
cover polytope has vertices $\left\{ (\frac 1 2,\frac 1 2,\frac 1 2),
  (1,1,0),(1,0,1),(0,1,1) \right\}$, thus,
\begin{align}
\AGM(Q) =
% \min (\sqrt{|R|\cdot |S| \cdot |T|},|R|\cdot |S|,|R|\cdot|T|,|S|\cdot|T|).
\min (\sqrt{N_R N_S N_T},N_R N_S, N_R N_T, N_S N_T) \label{eq:shearer}
\end{align}

Several query evaluation algorithms have been described in the
literature with runtime\footnote{$\tilde O$ means up to a logarithmic
factor.} $\tilde O(N+\AGM(Q))$: NPRR~\cite{DBLP:conf/pods/NgoPRR12},
LFTJ~\cite{LFTJ},
Generic-join~\cite{skew}.

\paragraph*{Functional Dependencies.} 
A query with functional dependencies is a pair $(Q,\FD)$, where $Q$ is
a query and $\FD$ is a set of {\em functional dependencies} (fd),
which are expressions of the form $U \rightarrow V$ where $U,V
\subseteq \bX$.  An fd can be either defined by some relation
$R_j(X_j)$, in which case we call it {\em guarded} (in particular $U,
V \subseteq X_j$), or can be defined by a UDF (as we saw in
Sec.~\ref{subsec:motivation}), and then we call it {\em unguarded}.  A
{\em simple} fd is of the form $u \rightarrow v$ where both $u, v$ are
variables, and a {\em simple key} is a simple fd guarded in some $R_j$
s.t. $u$ is a key for $R_j$.

Output size bounds in the presence of a set of fd's was studied
in~\cite{GLVV}. Their bound is defined only in terms of the maximum
cardinality, $\max_j N_j$, but in this paper we generalize the
discussion to all cardinalities $(N_j)_j$.  The key technique
introduced in~\cite{GLVV} consists of using information theory to
upper bound the size of the query, as reviewed next.

Let $Q^D$ be the query answer over some instance $D$.  Define a
probability distribution over $\prod_{i=1}^k D_i$  by randomly drawing one tuple
from $Q^D$ with probability $1/|Q^D|$ each.  Under this distribution, the joint
entropy of the $k$ random variables $\bX$ is $H(\bX)= \log_2 |Q^D|$.
Each subset of variables $X \subseteq \bX$ defines a (marginal)
distribution, with entropy $H(X)$.  If $X = X_j$ ($= \vars(R_j)$),
then $H$ must satisfy the following {\em cardinality
constraint} $H(X_j) \leq \log_2 N_j = n_j$, because $\Pi_{X_j}(Q^D)
\subseteq R_j^D$ and the marginal entropy is bounded above by the
uniform marginal entropy. With some abuse we write $H(R_j) \leq \log_2
N_j$, blurring the distinction between $R_j$ and $X_j$.  In addition,
for any fd $U \rightarrow V$ the entropy must satisfy the {\em
  fd-constraint} $H(U) = H(UV)$.

We give here a very simple illustration of how the approach
in~\cite{GLVV} models the query output using entropy, by illustrating
on the query $Q(x,y,z) = R(x,y),S(y,z),T(z,x)$ (without fd's) and
output with five outcomes:

\begin{center}
\begin{tabular}[t]{|l|l|l|l} \cline{1-3} $x$ & $y$ & $z$ \\
\cline{1-3} $a$ & $3$ & $r$ & $1/5$\\ $a$ & $2$ & $q$ & $1/5$\\ $b$ &
$2$ & $q$ & $1/5$\\ $d$ & $3$ & $r$ & $1/5$\\ $a$ & $3$ & $q$ &
$1/5$\\ \cline{1-3}
\end{tabular}
\begin{tabular}[t]{|l|l|l} \cline{1-2} $x$ & $y$ \\ \cline{1-2} $a$ &
$3$ & $2/5$ \\ $a$ & $2$ & $1/5$ \\ $b$ & $2$ & $1/5$ \\ $d$ & $3$ &
$1/5$ \\ \cline{1-2}
\end{tabular}
\begin{tabular}[t]{|l|l|l} \cline{1-2} $y$ & $z$ \\ \cline{1-2} $3$ &
$r$ & $2/5$ \\ $2$ & $q$ & $2/5$ \\ $3$ & $q$ & $1/5$ \\ $4$&$q$&0\\ \cline{1-2}
\end{tabular}
\begin{tabular}[t]{|l|l|l} \cline{1-2} $x$ & $z$ \\ \cline{1-2} $a$ &
$r$ & $1/5$ \\ $a$ & $q$ & $2/5$ \\ $b$ & $q$ & $1/5$ \\ $d$ & $r$ &
$1/5$ \\ \cline{1-2}
\end{tabular}
\end{center}

Here $H(xyz) = \log 5$, $H(xy) \leq \log |R| = \log 4$, $H(yz) \leq
\log |S| = \log 4$, and $H(xz) \leq \log |T| = \log 4$.

\paragraph*{GLVV Bound.} 
Gottlob et al.~\cite{GLVV} observe that $\log |Q^D| \leq \max_H
H(\bX)$, where $H$ ranges over all entropic functions that satisfy the
cardinality constraints and the fd-constraints, and this bound has
recently been shown to be tight~\cite{DBLP:journals/corr/GogaczT15}.
However, computing this upper bound is extremely difficult, because of
a long standing open problem in information theory: the
characterization of the cone of the closure of the set of entropic
functions.  To circumvent this difficulty, in~\cite{GLVV} the entropic
function $H$ is replaced with a polymatroid function $h$.  A {\em
  polymatroid} over a set of variables $\bX$ is a function $h : 2^\bX
\rightarrow \R^+$ that satisfies the following inequalities, called
{\em Shannon inequalities}:
\begin{align*}
  h(X) + h(Y) \geq & \ h(X \cap Y) + h(X \cup Y) &&  \text{Sub-modularity} \\
  h(X\cup Y) \geq & \ h(X) && \text{Monotonicity} \\
  h(\emptyset) = & \ 0 && \text{Zero}
\end{align*}
Every entropic function $H$ is a polymatroid, and the converse fails
when $\bX$ has four or more
variables~\cite{DBLP:journals/tit/ZhangY98}.  The bound introduced
in~\cite{GLVV} is $\GLVV(Q, \FD, (N_j)_j) \defeq \max_h h(\bX)$, where
$h$ ranges over all polymatroids that satisfy all cardinality
constraints $h(X_j) \leq n_j$, and all fd-constraints $h(UV)=h(U)$ for
$U \rightarrow V \in \FD$.  We abbreviate it $\GLVV(Q)$ when $\FD$ and
$(N_j)_j$ are clear from the context.  Clearly $\GLVV(Q)$ is an upper
bound on $|Q^D|$, and it is open whether the bound is always tight.

\paragraph*{Closure.} Fix a set $\FD$.  The {\em closure} of a set $X
\subseteq \bX$ is the smallest set $X^+$ such that $X\subseteq X^+$,
and that $U\to V \in \FD$ and $U\subseteq X^+$ imply $V \subseteq X^+$.
%that (1) contains $X$ and (2)
%contains $V$ whenever there exists $U \subseteq X^+$ and $U
%\rightarrow V \in \FD$.  
The intersection $X^+ \cap Y^+$ of closed
sets is closed, hence the family $\setof{X^+}{X \subseteq \bX}$ is a
{\em closure}.

Denote by $Q^+$ the query obtained by replacing each relation
$R_j(X_j)$ with $R_j(X_j^+)$, then forgetting all functional
dependencies; it is easy to check that $\max_D |Q^D| \leq \AGM(Q^+)$.
This bound is tight when all fd's are simple keys, because any
product database over the schema of $Q^+$ can be converted into a
database over the schema for $Q$ that satisfies all simple keys.
Thus, for simple keys, $\AGM(Q^+)$ is a tight upper bound on $|Q^D|$.
Theorem 4.4 in~\cite{GLVV} uses query coloring to prove essentially
the same result.  For a simple illustration, consider $Q \cd
R(x,y),S(y,z),T(z,u),K(u,x)$, where $\AGM(Q) =
\min(|R|\cdot|T|,|S|\cdot|K|)$.  If we define $y$ to be a key in $S$,
in other words $\FD = \set{y\rightarrow z}$, then $Q^+ =
R(x,y,z),S(y,z),T(z,u),K(u,x)$ and $AGM(Q^+) = \min
\set{|R|\cdot|T|,|S|\cdot|K|,|R|\cdot|K|}$.  However, for fd's other
than simple keys this technique fails, as illustrated by $Q(x,y,z) =
R(x),S(y),T(x,y,z)$ where $xy$ is a key in $T$ (in other words $xy
\rightarrow z$): when $|R|=|S|=N$, $|T| = M \gg N^2$ then $Q^+=Q$ and
$\AGM(Q^+) = M$, yet one can easily verify that $|Q^D| \leq N^2$.

% For example, consider
% $Q(x,y,z)=R(x),S(y),z=x+y$; the query satisfies the FD
% $xy\rightarrow z$, yet $Q^+=Q$ and $\AGM(Q^+) = \infty$ because $z$ is
% isolated in the hypergraph $H_Q$.

% \paragraph*{Extensions over GLVV~\cite{GLVV}.} Our model extends that
% of~\cite{GLVV} in two ways.  First, our bounds depend on all
% cardinalities $N_j$; in contrast the bounds in~\cite{GLVV} are given
% in terms of $\max_j N_j$.  Second, we do not require the FDs to be
% guarded: an FD $X \rightarrow Y$ is {\em guarded} if there exists a
% relation $R_j$ where this FD holds.  We allow FD's to be unguarded
% because in practice they may be introduced by a UDF.  Consequently,
% when we describe an instance $D$, we need to describe it through the
% query output $Q^D$, and assume implicitly that $R_j^D \defeq
% \Pi_{X_j}(Q^D)$.  For example, given the query
% $$Q \cd R(x),S(y),T(z),xy\rightarrow z,xz\rightarrow
% y,yz\rightarrow x,$$ we may describe an instance as
% $$\setof{(i,j,k) \in [N]^3}{i+j+k \mod N=0}.$$
% Just the relations $R^D=S^D=T^D=[N]$ are insufficient to capture the
% FDs.

\paragraph*{The Expansion Procedure.} All our algorithms use the
following simple subroutine, called an expansion.  Fix a relation
$R(X)$ and a database instance $D$.  $R$ may be an input relation, or
some intermediate relation generated during query evaluation.  An {\em
  expansion} of $R^D$ is some relation $(R^+)^D$ over attributes $X^+$
such that $\Pi_{X^+}(Q^D)\subseteq(R^+)^D$ and $\Pi_X((R^+)^D)
\subseteq R^D$.  If $X^+ = X$, then the expansion could be $R^D$
itself, or any partial semi-join reduction that removes dangling
tuples from $R^D$ (which do not join with tuples in other relations).
If $X \neq X^+$, then the expansion fills in the extra attributes, by
repeatedly applying functional dependencies $X\rightarrow y$: if the
fd is guarded in $R_j$, then it joins $R$ with $\Pi_{Xy}(R_j)$;
otherwise, if the fd corresponds to a UDF $y = f(X)$ then it simply
computes $f$ for each tuple in $R$.  The expansion of
$R$ can be done in time $\tilde O(N)$ using standard techniques.
% 
% , either
% by looking up the value $y$ corresponding to $X$ in a relation that
% defines this FD, or by computing the UDF $y = f(X)$, which we assume
% can be done in time $O(\log N)$.

% We will assume that an expanded relation can be computed
% in time $|D|\log |D|$, by iterating over every tuple $t \in R^D$ and
% filling in the missing attributes (or deleting the tuple if no
% corresponding values exist), for example by computing a UDF, or by
% looking up the missing value in some other relation that guards the
% FD.
% 
% : exact details depend on what has
% generated the FD's in the first place, for example in the case of user-defined
% functions the missing attributes are computed by applying
% these functions; or, if they are derived from FD's on other tables,
% then we look up those values. For example, given three relations
% $R(x,y),S(x,u,v),T(y,u,z)$ and FD's $x \rightarrow u$ (in $S$) and $yu
% \rightarrow z$ (in $T$), the expansion $R^+(x,y,z)$ is computed by
% first using $x$ to look up the $u$-value in $S$, then using $x,u$ to
% lookup $z$ in $T$.

%!TEX root = main.tex

\section{A Lattice-based Approach}
\label{sec:lattice}

\subsection{Lattice representation of queries with FDs}
\label{subsec:lattice:rep}

Fix a query with functional dependencies $(Q,\FD)$, over variables
$\bX$.  It is well
known~\cite{DBLP:journals/dam/DemetrovicsLM92,DBLP:journals/actaC/Levene95,Harremoes2011g}
that the set of closed sets forms a lattice:

\begin{defn}[Lattice representation]
  The lattice associated to $\FD$ is $\bL_{\FD} = (L,\preceq)$, where $L$
  consists of all closed sets and the partial order
  $\preceq$ is $\subseteq$.
  The {\em lattice representation} of a query $(Q,\FD)$ is the pair
  $(\bL_{\FD}, \bR^+)$, where $\bR^+ = \set{R_1^+, \ldots, R_m^+}
  \subseteq L$ is the set of closures of the attributes of the input
  relations.
We drop the subscript $\FD$ when it is clear from the context and write simply
$\bL$.  
With the expansion procedure in place, w.l.o.g. we 
assume that all the input $R_j$ are closed sets; thus, the lattice
representation of the query can be denoted simply by
$(\bL, \bR)$, where $\bR = \set{R_1, \ldots, R_m} \subseteq L$, and
$\bigvee \bR = \bigvee_{R_j\in \bR} R_j = \hat 1$.
\end{defn}

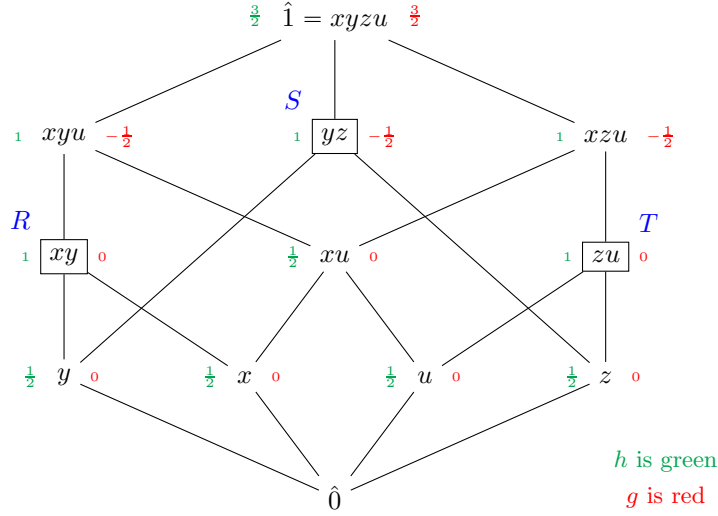
\begin{figure}[t]
   \centering
\begin{tikzpicture}[domain=0:20, scale=0.8]
  \node[Green] at (10,0.6) (0) {\small $h$ is green};
  \node[red] at (10,0) (0) {\small $g$ is red};
  \node[] at (4.5,0) (0) {$\hat 0$};
  \node[] at (0,2) (y) {$y$};
  \node[] at (3,2) (x) {$x$};
  \node[] at (6,2) (u) {$u$};
  \node[] at (9,2) (z) {$z$};
  \node[rectangle,draw] at (0,4) (xy) {$xy$};
  \node[] at (4.5,4) (xu) {$xu$};
  \node[rectangle,draw] at (9,4) (zu) {$zu$};
  \node[] at (0,6) (xyu) {$xyu$};
  \node[rectangle,draw] at (4.5,6) (yz) {$yz$};
  \node[] at (9,6) (xzu) {$xzu$};
  \node[] at (4.5,8) (1) {$\hat 1 = xyzu$};

  \node[thick, Green, left=0 of y] {\tiny $\frac 1 2$};
  \node[thick, red, right=0 of y] {\tiny $0$};
  \node[thick, Green, left=0 of x] {\tiny $\frac 1 2$};
  \node[thick, red, right=0 of x] {\tiny $0$};
  \node[thick, Green, left=0 of u] {\tiny $\frac 1 2$};
  \node[thick, red, right=0 of u] {\tiny $0$};
  \node[thick, Green, left=0 of z] {\tiny $\frac 1 2$};
  \node[thick, red, right=0 of z] {\tiny $0$};
  \node[blue,thick, above left=0 of xy] {$R$};
  \node[blue,thick, above left=0 of yz] {$S$};
  \node[blue,thick, above right=0 of zu] {$T$};
  \node[thick, Green, left=0 of xy] {\tiny $1$};
  \node[thick, red, right=0 of xy] {\tiny $0$};
  \node[thick, Green, left=0 of zu] {\tiny $1$};
  \node[thick, red, right=0 of zu] {\tiny $0$};
  \node[thick, Green, left=0 of xu] {\tiny $\frac 1 2$};
  \node[thick, red, right=0 of xu] {\tiny $0$};
  \node[thick, Green, left=0 of xyu] {\tiny $1$};
  \node[thick, red, right=0 of xyu] {\tiny $-\frac 1 2$};
  \node[thick, Green, left=0 of yz] {\tiny $1$};
  \node[thick, red, right=0 of yz] {\tiny $-\frac 1 2$};
  \node[thick, Green, left=0 of xzu] {\tiny $1$};
  \node[thick, red, right=0 of xzu] {\tiny $-\frac 1 2$};
  \node[thick, Green, left=0 of 1] {\tiny $\frac 3 2$};
  \node[thick, red, right=0 of 1] {\tiny $\frac 3 2$};

  \path[] (0) edge (x);
  \path[] (0) edge (y);
  \path[] (0) edge (z);
  \path[] (0) edge (u);
  \path[] (x) edge (xy);
  \path[] (x) edge (xu);
  \path[] (u) edge (xu);
  \path[] (u) edge (zu);
  \path[] (z) edge (yz);
  \path[] (z) edge (zu);
  \path[] (y) edge (yz);
  \path[] (y) edge (xy);
  \path[] (xy) edge (xyu);
  \path[] (xu) edge (xyu);
  \path[] (xu) edge (xzu);
  \path[] (zu) edge (xzu);
  \path[] (xyu) edge (1);
  \path[] (yz) edge (1);
  \path[] (xzu) edge (1);

  \end{tikzpicture}
   \caption{$Q \cd R(x,y),S(y,z),T(z,u),xz\rightarrow u,yu\rightarrow x$}
   \label{fig:q1}
\end{figure}

Note that the
size of the lattice may be exponential in that of the query, but this
does not affect our complexity results, because they are given in
terms of data complexity only.  If $Q$ has no functional dependencies,
then $\bL$ is the Boolean algebra $2^\bX$.
We will use Fig.\ref{fig:q1} as a running example, which
illustrates the lattice for query~\eqref{eqn:running}:
the lattice elements corresponding to input relations are framed.

We briefly review some notions of lattice theory needed later and
refer the reader to~\cite{MR2868112} for an extensive treatment.
Let $\wedge, \vee, \hat 0, \hat 1$ denote the greatest lower bound ({\em meet}),
least upper bound ({\em join}), minimum and maximum elements of the lattice.
The following hold in $\bL_{\FD}$: $X \wedge Y = X \cap Y$, $X \vee Y =
(X \cup Y)^+$, $\hat 0 = \emptyset$, and $\hat 1 = \bX$. 
We write $X \incomp Y$ to mean $X$ and $Y$ are incomparable.
An element $U$ is said to {\em cover} an element $V$ of $L$ if $U \succ V$ and 
$U \succeq W \succeq V$ 
implies $W=U$ or $W = V$; an {\em atom} is an element $X$ that covers $\hat 0$; 
a {\em co-atom} is an element covered by $\hat 1$; $X$ is called a {\em
join-irreducible} if $Y \vee Z = X$ implies $Y=X$ or $Z=X$; $X$ is
called {\em meet-irreducible} if $Y \wedge Z = X$ implies $Y=X$ or
$Z=X$.  For any $X\in L$, let $\Lambda_X \defeq \setof{Z}{Z \preceq X,
  Z \text{ is a join-irreducible}}$.  The mapping $X \mapsto \Lambda_X$
defines a $1$-$1$ mapping from $L$ to the order ideals of the poset of
join-irreducibles, with inverse $\Lambda_X \mapsto \bigvee \Lambda_X =
X$. 

Lattice presentations and queries with functional dependencies are in
1-to-1 correspondence, up to the addition/removal of variables
functionally equivalent to some other set of variables.  To see this
in one direction, consider any pair $(\bL,\bR)$ where $\bR =
\set{R_1,\ldots, R_m} \subseteq L$ and $\bigvee \bR = \hat 1$. Let $\bX$
be the join-irreducibles of $\bL$, define a query $(Q,\FD)$ 
as follows. Let $R_j$ be a relation with attributes $\Lambda_{R_j}$, and 
define $\FD$ such that the closed sets are 
$\setof{\Lambda_U}{U \in L}$ (in other
words $\FD \defeq \setof{X \rightarrow \Lambda_{\bigvee X}}{X
\subseteq \bX}$). One can check that the lattice presentation of
this query is isomorphic to $(\bL,\bR)$.  In the other direction,
consider a query with functional dependencies, $(Q,\FD)$.  Call a
variable $x$ redundant if $Y \leftrightarrow x$ for some set $Y$ that
does not contain $x$.  W.l.o.g. we can remove all redundant variables
from $Q$ (and decrease accordingly the arity of the relations
containing $x$), because the values of a redundant variable can be
recovered through expansion (Sec.\ref{sec:background}).  We claim $x
\mapsto x^+$ is a 1-to-1 mapping between the variables of $Q$ and the
join-irreducibles of $\bL_\FD$.  We first check that $x^+$ is
join-irreducible: if $x^+ = Y \vee Z = (Y \cup Z)^+$ with $x \not\in
Y, x \not\in Z$ then $x \leftrightarrow Y \cup Z$ contradicting the
fact that $x$ is not redundant.  For injectivity, if $x \neq y$, then
$x^+ = y^+$ implies $x \leftrightarrow y^+ - \set{x}$, again
contradicting non-redundancy.  Finally, surjectivity follows from the
fact that $Y = \bigvee_{x \in Y} x^+$ for any closed set $Y \in L$: if
$Y$ is join-irreducible then $Y = x^+$ for some $x \in Y$, proving
surjectivity.
% 
%  and it's lattice presentation $(\bL,\bR)$: we show
% that, w.l.o.g., its variables are in 1-to-1 correspondence with the
% join-irreducibles of $\bL$.  To justify this, let $x$ be a variable,
% and $x^{\pm} \defeq x^+ - \{x\}$.  If $x^{\pm} \rightarrow x$ holds,
% then $x^{\pm} \leftrightarrow x$ and we can remove $x$ from the query,
% since it can be recovered through expansion.  Thus, we can assume
% w.l.o.g. that, for any subset $A \subseteq x^{\pm}$, $A
% \not\rightarrow x$.  Then, $x^+$ is a join irreducible; otherwise,
% $x^+ = Y \vee Z = (Y \cup Z)^+$, which means $Y \cup Z \rightarrow x$,
% which is a contradiction because $Y \cup Z \subseteq x^{\pm}$.  Hence,
% $x \mapsto x^+$ maps variables to join irreducibles, and is 1-to-1
% because $y \neq x$ and $x^+=y^+$ imply $y \in x^{\pm}$ and $y
% \rightarrow x$, which is a contradiction.  
%

Thus, in the rest of the paper we shall freely switch between queries
and lattices, using the following correspondence:

\begin{center}
\begin{tabular}{|l|l|} \hline
variable $x \in \bX$                    & join-irreducible $X \in L$ \\
input $R_j(X_j)$                        & input $R_j \in L$ with variables $\Lambda_{R_j}$ \\
set of variables $X \subseteq \bX$      & set of join-irreducibles $X$ \\
its closure $X^+$                       &   $X^+\defeq \Lambda_Z$, where $Z=\bigvee X$\\ 
FD$=\emptyset$                          & Boolean Algebra \\ \hline
\end{tabular}
\end{center}

The following simple observation illustrates the power of the lattice formalism.

\bprop If all functional dependencies are simple then $\bL$ is a
distributive lattice.
\label{prop:simple:FD:implies:distributive:lattice}
\eprop
\bp Construct a directed graph $D=(V,E)$ where $V$ is the set of all
variables and there is an edge $(a,b) \in D$ if $a \to b$ is a
functional dependency. Collapse each strongly connected component of
this graph into a single vertex. We are left with a directed acyclic
graph (DAG), which defines a poset $P$ where $y \prec_P x$ iff there
is a directed path from $x$ to $y$ in the DAG.  Then, $L$ is precisely
the order ideal lattice of the direct sum of these posets. The claim
follows because any order ideal lattice is distributive.  
\ep 

It is easy to construct examples when non-simple FDs still produce a
distributive lattice; for example, the lattice for
$Q \cd R(x),S(y),xy\to z$ is isomorphic to the Boolean algebra lattice
$2^{x,y}$ which is distributive.
Hence, the class of distributive lattices strictly covers the simple FD
case.

Consider the non-distributive
lattice $M_3$ (Fig.~\ref{fig:non:normal}), then the procedure we
discussed above associates it with the following query:
$R(x),S(y),T(z), xy \to z, xz \to y, yz \to x$.

\subsection{Database Instances for a Lattice}
\label{subsec:db-instances}

Let $\bX$ denote the set of join-irreducibles of a lattice $\bL$, and
let $\FD = \setof{X \rightarrow \Lambda_{\bigvee X}}{X \subseteq \bX}$
be all fd's implicitly defined by $\bL$.  A {\em database instance $D$
for $\bL$} is a relation with attributes $\bX$ that satisfies $\FD$.
We denote $h_D : L \rightarrow \R^+$ its entropy function, as
defined in Sec.~\ref{sec:background}. 
If $(\bL, \bR)$ is the lattice
presentation of a query $(Q,\FD)$, then any database instance for
$\bL$ defines a standard database instance for $Q$ by
$R^D_j \defeq \Pi_{R_j}(D)$.  For example, given the query
$$Q \cd R(x),S(y),T(z),xy\rightarrow z,xz\rightarrow
y,yz\rightarrow x,$$
An example of an instance for its lattice (which is $M_3$ in
Fig.~\ref{fig:non:normal}) is:
$$D = \setof{(i,j,k) \in [N]^3}{i+j+k \mod N=0}.$$
It defines the following standard instance: $R^D=S^D=T^D=[N]$.  Notice
that all fd's are unguarded and hence lost in the standard instance.
From Sec.~\ref{sec:ca} and up we will assume that during the query
evaluation we have access to the UDF's that defined the unguarded
fd's.

\bdefn[Materializability]
A function $h : \bL \rightarrow \R^+$ is called {\em $\bL$-entropic} if
$h = h_D$ for some instance $D$ for $\bL$ and it is called {\em
materializable} w.r.t. an input query $(\bL,\bR)$
if there exists an instance $D$ such that
$\log |R_j^{D}| \leq h(R_j)$ forall $j=1,m$ and
$\log |Q^D| \geq h(\hat 1)$.  Obviously, any entropic function is
materializable w.r.t. some input instance.
\label{defn:Materializability}
\edefn

\subsection{The Lattice Linear Program} 
\label{subsec:llp}

This section shows that a very simple linear program defined on the FD lattice 
$\bL$ is equivalent to GLVV bound.
Given a query $(\bL, \bR)$ and log cardinalities $(n_j)_{j=1}^m$, we
define the following {\em Lattice Linear Program}, or LLP, over the
{\em non-negative} variables $(h(X))_{X \in L}$:
%
%\begin{align}
%  & \text{maximize } h(\hat 1) \label{eq:llp} \\
%\forall j: & h(R_j) \leq n_j \nonumber \\
%\forall x \incomp y: & h(x\wedge y) + h(x \vee y) - h(x) - h(y) \leq 0 \nonumber
%\end{align}
\begin{equation}\label{eq:llp}
\begin{matrix*}[r]
\max & \multicolumn{1}{l}{h(\hat 1)}\\
     & h(X\wedge Y) + h(X \vee Y) - h(X) - h(Y) &\leq& 0, & \forall X,Y\in L, X \incomp Y.\\
     & h(R_j) &\leq& n_j, & \multicolumn{1}{l}{\forall j \in [m].}
%    & h(X) & \geq & 0, & \multicolumn{1}{l}{\forall X \in L.}
\end{matrix*}
\end{equation}
%
%where all variables in both LPs are assumed to be $\geq 0$.  
% where we write $u\incomp v$ to denote the fact that $u, v$ are incomparable. 
% The primal has $|L|$ variables, $h(x)$ for $x \in L$, and
% two types of constraints: cardinality and sub-modularity constraints.
% The dual LLP has variables $w_j$ (corresponding to input
% relations) and one variable $s_{X,Y}$ for every unordered pair
% of incomparable elements $X,Y$ (i.e. $s_{X,Y}$ and $s_{y,x}$ denote
% the same variable).  

A feasible solution $h$ to LLP is called a (non-negative) {\em
$\bL$-sub-modular} function. 
If $h$ is also {\em $\bL$-monotone}, i.e. $X \preceq Y$ implies $h(X) \leq h(Y)$,
then $h$ is called an {\em $\bL$-polymatroid}. 
When $\bL$ is clear from the context we drop the $\bL$-prefix from $\bL$-monotone,
$\bL$-submodular, and $\bL$-polymatroid, respectively.
We did not require $h$ to be monotone because, at optimality,
$h^*$ can always be taken to be a polymatroid thanks to {\em Lov\'asz's
monotonization}: if $h$ is non-negative $\bL$-submodular, then the function
$\bar h(\hat 0) \defeq 0$ and
$\bar h(X) \defeq \min_{Y: X \preceq Y}h(Y)$, $X \neq \zerohat$, is an
$\bL$-polymatroid (see e.g.~\cite{MR1956925}, pp. 774, and
Appendix~\ref{sec:appendix:lattice}) and satisfies:
$\bar h(\hat 1) = h(\hat 1)$ and $\forall X$, $\bar h(X) \leq h(X)$.

%Our discussion thus far implies that 
%In particular, $\GLVV(Q,\FD,(N_j)_j) = h^*(\hat 1)$.
%For any database instance $D$ that satisfies the cardinality
%constraints (meaning: $|\Pi_{X_j}(D)| \leq N_j$ forall $j$), its
%entropy $h_D$ is a feasible solution to the LLP.  This implies:

\begin{prop}
  \label{prop:upper:bound} Let $h^*$ be an optimal solution
   to the LLP of a query $(\bL,\bR)$ where $N_j=|R_j|$, $j \in [m]$.  Then,
  $2^{h^*(\hat 1)} = \GLVV(Q,\FD,(N_j)_j).$
\end{prop}
\bp
  From the discussion above, we can assume $h^*$ is an $\bL$-polymatroid.
  Define a function $h : 2^{\bX} \to \R$ by $h(X) = h^*(X^+)$; then, clearly $h$
  is non-negative and monotone (on the original variable set, not on the lattice
  $\bL$), and satisfies the cardinality and fd constraints. 
  We verify that it is also sub-modular:
  for any (not necessarily closed) set $X$ of variables, 
  noting that $(X\cap Y)^+ \preceq X^+ \wedge Y^+$, we have
  \begin{eqnarray*}
     h(X\cup Y)+h(X\cap Y) 
     &=& h^*((X\cup Y)^+) + h^*((X \cap Y)^+)\\
     (\text{monotonicity}) &\leq& h^*((X\cup Y)^+) + h^*(X^+ \wedge Y^+)\\
     &=& h^*(X^+ \vee Y^+) + h^*(X^+ \wedge Y^+)\\
     (\text{$\bL$-submodularity}) &\leq & h^*(X^+) + h^*(Y^+)\\
     &= & h(X) + h(Y).
  \end{eqnarray*}
  Conversely, given any polymatroid $h$ that satisfies the cardinality and
  fd-constraints, the restriction of $h$ on $L$ is a feasible solution to LLP.
\ep

The main goal in this paper is to design algorithms that compute a
query $Q$ in time $\tilde O(2^{h^*(\hat 1)})$.  When the bound in the
proposition is tight (which is currently an open problem) then such an
algorithm is optimal, hence our second goal is to study cases when the
bound is tight, and this happens if and only if the LLP has some optimal
solution $h^*$ that is materializable.  A secondary goal in this paper
is to find sufficient conditions for $h$ to be materializable.

When $\FD=\emptyset$, $\bL$ is a Boolean algebra. In this case
it is easy to see that $\AGM(Q) = 2^{h^*(\hat 1)}$. To see this, in one
direction start from an optimal weighted fractional vertex packing
$(v^*_i)_{i=1}^k$
of the query hypergraph (Sec.~\ref{sec:background}), then the following is
a feasible solution to LLP with the same objective value:
\begin{align}
h(X) \defeq & \sum_{i: x_i \in X} v_i^* \label{eq:bottom:up}
\end{align}
Conversely, given a polymatroid $h^*$ that is optimal to LLP, define 
$v_i \defeq h^*(\{x_1,\dots,x_i\})-h^*(\{x_1,\dots,x_{i-1}\})$, then it can
be verified (using sub-modularity of $h^*$) that $(v_i)_{i=1}^k$ is a 
feasible weighted fractional vertex packing with the same objective value.

% It follows from our discussion in Sec.\ref{sec:background} that, if
% $D$ is any database instance satisfying the FD's, and $H$ is its entropy,
% then the restriction of $H$ to closed sets is a feasible solution to
% the LLP.  Therefore:
% 
% 
% The LLP precisely captures the framework defined by Gottlob et
% al.\cite{GLVV}, but with the important difference that it uses
% lattice-theoretic concepts. This gives us much deeper insight into
% both the optimal solution and worst-case instances
% and allows us to design matching algorithms.

\subsection{Embeddings}

Embeddings allow us to construct an instance for one lattice $\bL$
from an instance for a different lattice $\bL'$.

\begin{defn} \label{def:embedding} An {\em embedding} between two
  lattices $\bL$ and $\bL'$ is a function $f : L \rightarrow L'$ such that
   $f(\bigvee X) = \bigvee f(X)$ forall $X \subseteq L$,\footnote{Equivalently:
   $f$ is the left adjoint of a Galois connection.} and
   $f(\hat 1_{\bL}) = \hat 1_{\bL'}$.  An embedding between two queries
  $f : (\bL, \bR) \rightarrow (\bL', \bR')$ is an embedding from
  $\bL$ to $\bL'$ that is a bijection from $\bR$ to $\bR'$.
\end{defn}

If $f$ is an embedding from $\bL$ to $\bL'$ and $h'$ is a
non-negative, sub-modular function on $\bL'$, then one can check that
$h \defeq h' \circ f$ is also sub-modular.  If $f$ is an embedding
between queries, then their relation symbols are in 1-1
correspondence: $R_1, \ldots, R_m$ and $R_1',\ldots, R_m'$
respectively.  Fix two queries $(Q,\FD), (Q',\FD')$ with variables
$\bX, \bX'$ respectively.  Call a function
$\coloring: 2^{\bX} \rightarrow 2^{\bX'}$ a {\em variable renaming} if
$\coloring(X)=\coloring(X^+)$,
$\coloring(X) = (\bigcup_{x \in X} \coloring(x))^+$ ($\coloring$ is
uniquely defined by its values on single variables), and
$\coloring(\bX) = \bX'$ (all variables in $Q'$ are used).  One can
check that every embedding $f$ defines the variable renaming
$\coloring(X) \defeq f(X^+)$ forall $X \subseteq \bX$, and vice versa.

Fix an embedding $f : \bL \rightarrow \bL'$, and an instance $D'$ for
$\bL'$.  Define the database instance $D \defeq f^{-1}(D')$ as
$D \defeq \setof{t\circ\coloring}{t \in D'}$.  In other words, every
tuple $t' \in D'$ has attributes $\bX'$, hence it can be viewed as a
function $t' : \bX' \rightarrow \text{Dom}$: for each such $t'$ we
include in $D$ a tuple
$t' \circ \coloring: \bX \rightarrow \text{Dom}$, obtained from $t'$
by renaming its attributes.  One can check that $|D| = |D'|$ (because
all variables in $\bX'$ are used), and the same holds for all
projections $|\Pi_X(D)| = |\Pi_{\coloring(X)}(D')|$ forall
$X \subseteq \bX$.  For example, if
$\coloring(x_1) \defeq \set{y_1,y_2}$,
$\coloring(x_2) \defeq \set{y_2,y_3}$ then
$D(x_1,x_2) \defeq \setof{((ab), (bc))}{(a,b,c) \in D'(y_1,y_2,y_3)}$.
This proves:

\begin{prop}
  If $f : \bL \rightarrow \bL'$ is a lattice embedding, then for any
  instance $D'$ for $\bL'$, $h_{D} = h_{D'} \circ f$, where
  $D \defeq f^{-1}(D')$.  If $f: (\bL,\bR) \rightarrow (\bL',\bR')$
  is a query embedding , and $h' : L'\rightarrow \R^+$ is
  materializable, then $h \defeq h' \circ f$ is also materializable.
  \label{prop:entropy-preserving-embedding}
\end{prop}

\begin{defn}
A {\em quasi-product database instance} for $\bL$ is an instance of
the form $f^{-1}(D')$ where $f$ embeds $\bL$ into a Boolean algebra
$A$, and $D'$ is a product database instance for $A$.
\label{defn:quasi:product:instance}
\end{defn}

% We define a {\em quasi-product database instance} to be $D =
% f^{-1}(D')$, where $f : L \rightarrow A$ is an embedding into a
% Boolean Algebra $A$, and $D'$ is a product instance on a Boolean
% algebra (each relation is the cross product of its variable domains).
% Quasi-product instances are, thus, product instances, up to the
% renaming of attributes.  

\begin{example} \label{ex:running:1} Continuing our running example,
  consider the following two queries:
  \begin{eqnarray*}
    Q&= & R(x,y),S(y,z),T(z,u), xz\rightarrow u, yu\rightarrow x\\
    Q'&= & R'(a,b),S'(b,c),T'(c,a)
  \end{eqnarray*}
  $Q$ is the query in Fig.~\ref{fig:q1}, while $Q'$ is the triangle
  query (without fd's).  Consider the following renaming
  $\coloring(x)=\coloring(u)=a$, $\coloring(y)=b$, $\coloring(z)=c$,
  which defines an embedding\footnote{We need to check
    $\coloring(X)=\coloring(X^+)$ for all $X$: e.g.
    $\coloring(xz)=\coloring(xz^+)$ holds because $\coloring(xz)=ac$,
    $\coloring(xz^+)=\coloring(xzu)=ac$, and similarly
    $\coloring(yu)=\coloring(yu^+)$.} from the lattice in
  Fig.~\ref{fig:q1} to the Boolean algebra $2^{\set{a,b,c}}$.  From
   the product instance $$D' \defeq \setof{(a,b,c)=(i,j,k)}{i,j,k \in [\sqrt
    N]}$$ we construct, through renaming, the quasi-product instance $$D
   = \setof{(x,y,z,u)=(i,j,k,i)}{i,j,k \in [\sqrt N]}.$$  Notice that the FD's
  $xz\rightarrow u, yu\rightarrow x$ hold in $D$ (as they should).
\end{example}

\subsection{Inequalities} 

In information theory, an {\em information inequality} is defined by a
vector of real numbers $(w_X)_{X \subseteq \bX}$ such that the
inequality $\sum_{X \subseteq \bX} w_X H(X) \geq 0$ holds for any
entropic function $H$.  We introduce here a related notion, which we
use to describe upper bounds on the query size.

Fix a query $(\bL, \bR)$.  An {\em output-inequality} is given by a
vector of non-negative real numbers $(w_j)_{j\in[m]}$ such that the
inequality
\begin{align}
  \sum_{j=1}^m w_j h(R_j) \geq h(\hat 1) \label{eq:ineq1}
\end{align}
holds for all polymatroids $h$ on the lattice $\bL$. 
For example, in a Boolean algebra output inequalities correspond to Shearer's 
lemma:
$\sum_j w_j h(R_j) \geq h(\bX)$ iff $(w_j)_{j=1}^m$ is a fractional edge
cover of the hypergraph with nodes $\setof{x}{x \in \bX}$ and
hyperedges $\setof{R_j}{j\in [m]}$.  {\em Any} output inequality gives us
immediately an upper bound on the output size of a query,
because \eqref{eq:ineq1} implies
$h(\hat 1) \leq \sum_j w_j n_j$ for any
feasible solution $h$ of the LLP with log-cardinalities $(n_j)_{j=1}^m$.

An ouput size upperbound is best if it is minimized, which is precisely
the objective of the dual-LLP, which is defined over
{\em non-negative} variables $(w_j)_{j=1}^m$ 
and $(s_{X,Y})_{X\incomp Y}$, corresponding to the input $R_j$
and incomparable pairs of lattice elements (thus $s_{X,Y}$ is the
same variable as $s_{Y,X}$):
\begin{equation}\label{eq:dual:llp}
\begin{matrix*}[r]
   \min & \multicolumn{3}{l}{\displaystyle \sum_j w_jn_j}\\[2ex]
        & {\displaystyle \sum_{X\incomp Y: X \vee Y=\hat 1} s_{X,Y}} &\geq& 1,& \\[2ex]
        & {\displaystyle w_j + \sum_{X\incomp Y: X \vee Y = R_j} s_{X,Y} +
        \sum_{X\incomp Y: X \wedge Y = R_j} s_{X,Y}
- \sum_{X \incomp R_j} s_{X,R_j}} &\geq& 0, & \multicolumn{1}{l}{\forall j \in
[m]} 
% \\[2ex]
%   & s_{X,Y} & \geq & 0,&\forall x, y \in L, x\incomp y.
\end{matrix*}
\end{equation}

% %
% We study output inequalities for two key reasons.  First, we can turn
% them into worst-case output size bounds; second, in
% Sec.~\ref{sec:proof:algorithms}, we will turn a proof of an output
% inequality into an algorithm matching the size bound.

% To motivate output inequalities, assume that $\bL$ is a lattice for
% which any optimal solution of the LLP is realizable.  Suppose we can
% design an optimal algorithm for $Q$, which, for any cardinalities
% $(N_j)_j$, computes $Q$ in time $O(2^{h^*(\hat 1)})$.  If the output
% inequality \eqref{eq:ineq1} holds for any polymatroid $h$, then the
% algorithm itself represents a proof of its validity.  Indeed, given
% any polymatroid $h$, define $n_j \defeq h(R_j)$ and denote $h^*$ the
% optimal solution of this LLP ????

%\hqn{All our algorithms use all properties of polymatroids: non-negativity,
%monotonicity, and sub-modularity. For example, the chain bound proves an output
%inequality for polymatroids -- the proof does not work without monotonicity.
%Similarly, the CSMA algorithm uses monotonicity because all conditional
%entropies have to be non-negative. I think we should work directly with
%polymatroid, and mention the fact that the LLP bound can be simplified to not
%include monotonicity.}

A simple way to prove an output inequality is to write it as a non-negative
linear combination of sub-modularity inequalities. (See example below.)
The following Lemma shows that every output inequality can be proven this way,
as follows from the theory of generalized
inequalities~\cite[Ch.2]{boyd:vandenberghe:2004}.  

% In a Boolean algebra, Shearer's lemma has the form~\eqref{eq:ineq1}.
% While there are several proofs of Shearer's lemma, it is not clear how
% to extend them to prove an output inequality for general lattices.
% From the theory of generalized
% inequalities~\cite[Ch.2]{boyd:vandenberghe:2004}, however, we can
% easily derive a necessary and sufficient condition for the
% coefficients $(w_j)_{j=1}^m$ to satisfy the inequality.  Let $M$
% denote the matrix of sub-modular inequalities (last inequality of the
% LLP \eqref{eq:llp}); in other words, the column vector $h \geq 0$
% defines a sub-modular function iff $Mh \leq 0$.  Let $K \defeq
% \setof{h}{Mh \leq 0, h\geq 0}$.  The polar cone $K^* \defeq
% \setof{c}{c^T h \leq 0, \forall h \in K}$ represents all inequalities
% that hold for all non-negative sub-modular functions.  Note that $c
% \in K^*$ iff the LP $\max \{c^Th \suchthat Mh \leq 0, h \geq 0\}$ has
% objective value $0$, which holds iff the dual LP $\{\min 0 \suchthat
% c^T \leq s^T M, s\geq 0\}$ is feasible.  This implies:
% 

\blmm 
Let $M$ denote the matrix of sub-modular inequalities (first inequality of 
the LLP \eqref{eq:llp}).  
Given a non-negative vector $(w_j)_{j\in[m]}$, define the vector
$(c_X)_{X\in L}$ by setting $c_{\hat 1} \defeq 1$, $c_{R_j} \defeq
-w_j$ and $c_X \defeq 0$ otherwise.  Then, the following
 statements are equivalent
\bi
 \item[(i)] \eqref{eq:ineq1} holds for all polymatroids $h$
 \item[(ii)] \eqref{eq:ineq1} holds for
all non-negative sub-modular functions $h$ 
 \item[(iii)] there exists $s \geq 0$ such that $c^T \leq s^T M$.  
    Equivalently, $w$ is part of a feasible solution $(s, w)$ of the
    dual-LLP~\eqref{eq:dual:llp}.
\ei
\label{lmm:generalized:inequality}
Furthermore, if $(s^*,w^*)$ and $h^*$ are dual- and primal-optimal solutions,
then $h^*(\onehat) = \sum_{j=1}^m w^*_jn_j$.
\elmm
\bp
Using the Lov\'asz monotonization map $h \to \bar h$ described in
Sec~\ref{subsec:llp}, it is
straightforward to show that $(i)$ and $(ii)$ are equivalent.
To see the equivalence between $(ii)$ and $(iii)$, note that
$K \defeq \setof{h}{Mh \leq 0, h\geq 0}$ is the
set of all non-negative, sub-modular functions, and the polar cone
$K^* \defeq \setof{c}{c^T h \leq 0, \forall h \in K}$ is the set of
all inequalities that hold for all non-negative sub-modular functions.
Now, $c \in K^*$ iff the LP $\max_h \{c^Th \suchthat Mh \leq 0, h
\geq 0\}$ has objective value $0$, which holds iff the dual LP $\min_s
\setof{0}{c^T \leq s^T M, s\geq 0}$ is feasible. 
The last statement is strong duality.
\ep

% This observation implies the simple fact that studying output
% inequalities is essentially equivalent to studying LLP.  If
% \eqref{eq:ineq1} holds, then the optimal value of the LLP satisfies
% $h^*(\hat 1) \leq \sum_j w_j n_j$, and, conversely, given an optimal
% solution $s^*,w^*$ to the dual LLP, \eqref{eq:ineq1} holds for $w^*$
% by Lemma~\ref{lmm:generalized:inequality}.

%  \begin{cor}
%    \label{cor:inequality-bound}
%    Let $(s^*,w^*)$ be an optimal solution to the dual-LLP.  Then
%    \eqref{eq:ineq1} with $w_j \defeq w^*_j$ holds for all non-negative
%    sub-modular functions $h$.
%  \end{cor}

%\begin{proof}
%Let $s^*, w^*$ denote an optimal solution to the dual-LLP.
%Define the vector $(c_x)_{x \in L}$ as follows: $c_{\hat 1} \defeq 1$, 
%$c_{R_j} \defeq w^*_j$, and $c_x=0$ otherwise.  
%Then, $c^T \leq (s^*)^TM$, and $s^* \geq 0$. Thus, \eqref{eq:ineq1} holds for the
%vector $(w^*_j)_{j=1}^m$.
%\end{proof}

\begin{example} \label{ex:running:2}
  The lattice presentation of the triangle query
  $Q=R(x,y),S(y,z),T(z,x)$ (without FD's) is the Boolean algebra
  $2^{\set{x,y,z}}$.  The following output inequality:
\begin{eqnarray}
   h(xy)+h(yz)+h(zx) &\geq & 2h(\hat 1) \label{eq:triangle}
\end{eqnarray}
follows by adding two sub-modularity inequalities $h(xy)+h(yz)\geq
h(\hat 1)+h(y)$ and $h(y)+h(zx)\geq h(\hat 1)$.  It corresponds to the
dual solution $s_{xy,yz}=s_{y,zx}=1$ with the rest 0.  Some other
output inequalities are $h(xy)+h(xz) \geq h(\hat 1)$, $h(xy)+h(yz)
\geq h(\hat 1)$, and $h(xz)+h(yz) \geq h(\hat 1)$.  Together, these
four output inequalities prove Eq.(\refeq{eq:shearer}).
% Consider now the LLP with
% cardinality constraints given by $n_R,n_S,n_T$.  For $n_R=n_S=n_T=6$,
% the optimal value is $h^*(\hat 1) = 9$, and this can be derived from
% the first inequality, $2h(\hat 1) \leq n_R+n_S+n_T$.  For $n_R=2$,
% $n_S=n_T=5$, $h^*(\hat 1) = 7$ which can be derived from the second
% inequality $h(\hat 1) \leq n_R+n_S$.
\end{example}

%!TEX root = main.tex

\section{Normal lattices}
\label{sec:normal}

% We show here that the elegant features of the AGM bound (upper bound
% given in terms of a fractional edge cover and lower bound given by a
% product database instance) carry over to a lattice $L$ iff $L$ has a
% special structure; we call $L$ a {\em normal lattice}.  Normality
% seems to be a new concept which strictly includes all distributive
% lattices and, hence, all Boolean algebras.  Fix a query $Q$ with FDs,
% let $L$ be its lattice, and let $h$ be a submodular function on $L$.
% A database $D$ is a {\em materialization} for $h$ if $|Q^D| \geq
% 2^{h(\hat 1)}$ and for any relation $R_j$, $|R_j^D| \leq 2^{h(R_j)}$.

We show here that an $\bL$-polymatroid $h$ can be materialized as a 
quasi-product database
instance iff it satisfies all output inequalities given by fractional
edge covers of a certain hypergraph; in that case we call $h$ {\em
  normal}.  The normal polymatroids are the largest class of
polymatroids that preserve the elegant properties of the AGM bound:
upper bound given in terms of a fractional edge cover, and lower bound
given a (quasi-) product database.  We then extend normality to a lattice,
which we call normal if its optimal polymatroid $h^*$ is normal.

Recall the M\"obius inversion formula in a lattice $\bL$:
\begin{equation}
  h(X) = \sum_{Y: X \preceq Y} g(Y) 
   \ \ \ \text{ iff } \ \ \ 
   g(X) = \sum_{Y: X \preceq Y} \mu(X,Y) h(Y) \label{eq:mobius}
\end{equation}
where $\mu(X,Y)$ is the M\"obius function on $\bL$ \cite{MR2868112}.  In
information theory, when $h$ is an entropy and $\bL$ a Boolean algebra,
the quantity $-g(X)$ is the (multivariate) {\em conditional mutual
information} $I(\hat 1-X \suchthat X)$, which we abbreviate CMI.
For example, in the Boolean Algebra $2^{\set{x,y,z}}$:
\begin{eqnarray*}
   g(xyz)&=& h(xyz) \\
   g(xy) &=& h(xy) - h(xyz)\\
   g(x) &=&  h(x) - h(xy) - h(xz) + h(xyz)\\
   g(\hat 0) &=& h(\hat 0) - h(x) - h(y) - h(z) 
             + h(xy) + h(xz) + h(yz) - h(xyz)
\end{eqnarray*}
%   \mak{Are you sure about the signs above? It seems to me that the second and fourth are inverted:
%   \[g(xy)=-I(z|xy)=h(xyz)-h(xy)\]
%   \[g(\zerohat)=-I(xyz|\zerohat)=h(xyz)-h(xy)-h(xz)-h(yz)+h(x)+h(y)+h(z)-h(\zerohat).\]}
%
% YES the signs are correct, just reading off from the mobius inversion formula

We give below a simple sufficient condition on $g$ which implies that $h$,
defined by Eq.\eqref{eq:mobius} is a polymatroid.  We need:

\begin{lmm} \label{lemma:compaction} For $\calS\subseteq L$, let
  $c(Z,\calS) \defeq |\setof{U}{U \in \calS, U \preceq Z}|$.  Then, for all
  $X,Y,Z \in L$,
  $c(Z,\set{X,Y}) \leq c(Z,\set{X \wedge Y, X \vee Y})$; moreover, if
  $\bL$ is a distributive lattice and $Z$ is a meet-irreducible, then 
   equality holds.
\end{lmm}
\begin{proof}
  If $c(Z,\set{X,Y})=1$, then w.l.o.g. we assume $X \preceq Z$, hence
  $X \wedge Y\preceq Z$, and
  $c(Z, \set{X \wedge Y, X \vee Y}) \geq 1$.  If $c(Z,\set{X,Y})=2$,
  then both $X \preceq Z, Y \preceq Z$ hold, hence both
  $X \wedge Y \preceq Z$, $X \vee Y \preceq Z$ hold, hence
  $c(Z, \set{X \wedge Y, X \vee Y}) = 2$.  Next, suppose $\bL$ is distributive 
   and $Z$ is a meet-irreducible. If
  $c(Z, \set{X \wedge Y, X \vee Y}) = 2$ then $X \vee Y \preceq Z$
  hence both $X \preceq Z, Y \preceq Z$ hold.  If
  $c(Z, \set{X \wedge Y, X \vee Y}) = 1$, then $X \wedge Y \preceq Z$.
  By distributivity
  $Z \succeq (X \wedge Y) \vee Z = (X \vee Z) \wedge (Y \vee Z)$.
  On the other hand both $Z \preceq X \vee Z$ and $Z \preceq Y \vee Z$
  hence $Z = (X \vee Z) \wedge (Y \vee Z)$.  Since $Z$ is a
  meet-irreducible, it must be equal to one of the two terms. Assume
  w.l.o.g. $Z = X \vee Z$ then $X \preceq Z$ proving $c(Z, \set{X,Y})
  \geq 1$. 
\end{proof}

% \begin{lmm} \label{lemma:compaction} For $S\subseteq L$,
%   let\footnote{``$S$ intersected with the principal filter defined by
%     $z$''.} $c_f(z,S) \defeq |\setof{u}{u \in S, z \preceq u}|$.
%   Then, for all $x,y,z \in L$, $c_f(z,\set{x,y}) \leq c_f(z,\set{x
%     \wedge y, x \vee y})$; moreover, if $L$ is a distributive lattice
%   and $z$ is a join-irreducible, then we have equality.
% 
% The same holds for $c_i(z,S) \defeq |\setof{u}{u \in S, u\preceq z}$,
% with equality holding if $L$ is distributive and $z$ a meet-irreducible.
% \end{lmm}
% \begin{proof}
%   If $c_f(z,\set{x,y})=1$, then w.l.o.g.  $z \preceq x$, hence
%   $z\preceq x \wedge y$, and $c_f(z, \set{x \wedge y, x \vee y}) \geq
%   1$.  If $c_f(z,\set{x,y})=2$, then both $z \preceq x, z \preceq y$
%   hold, hence both $z \preceq x \wedge y$, $z \preceq x \vee y$ hold,
%   hence $c_f(z, \set{x \wedge y, x \vee y}) = 2$.  If $L$ is
%   distributive, and $z$ join-irreducible, then we prove the converse.
%   If $c_f(z, \set{x \wedge y, x \vee y}) = 2$ then $z \preceq x \wedge
%   y$ hence both $z \preceq x, z \preceq y$ hold.  If $c_f(z, \set{x
%     \wedge y, x \vee y}) = 1$, then $z \preceq x \vee y$.  By
%   distributivity $z \preceq (x \vee y) \wedge z = (x \wedge z) \vee (y
%   \wedge z)$.  On the other hand both $x \wedge z \preceq z$ and $y
%   \wedge z \preceq z$, hence $z = (x \wedge z) \vee (y \wedge z)$.
%   Since $z$ is join irreducible it must be equal to one of the two
%   terms, assume wlog $z = x \wedge z$, implying $z \preceq x$, hence
%   $c_f(z, \set{x,y}) \geq 1$.
% \end{proof}

%
Lemma~\ref{lemma:compaction} immediately implies:

\begin{lmm} \label{lemma:normal:submodular} Let $g$ be any function
  s.t.  $g(Z) \leq 0$ for $Z \prec \hat 1$, and $g(\hat 1) = - \sum_{Z
    \prec \hat 1} g(Z)$. 
  Then the function $h$ defined by Eq.(\ref{eq:mobius}) is a polymatroid.  
  Furthermore,
   if $\bL$ is distributive, then $h$ is a modular polymatroid, i.e.
  $h(X)+h(Y)=h(X \vee Y)+h(X \wedge Y)$ for all $X,Y\in L$.
\end{lmm}

\begin{proof}
  Non-negativity and monotonicity are easy to verify, submodularity follows
  from:
   \begin{eqnarray*}
      h(X) + h(Y) 
      &=& \sum_{Z: X \preceq Z} g(Z) + \sum_{Z: Y \preceq Z} g(Z) \\
      &=& \sum_Z c(Z,\set{X,Y}) \cdot g(Z)\\
      &\geq& \sum_Z c(Z, \set{X \vee Y, X \wedge Y}) \cdot g(Z)\\
      &=& h(X \vee Y) + h(X \wedge Y),
   \end{eqnarray*}
  where the inequality holds due to 
  Lemma~\ref{lemma:compaction} and the fact that $g(Z) \leq 0$ for $Z
  \neq \hat 1$. When $\bL$ is distributive, the inequality becomes an equality.
\end{proof}

Any function $h$ satisfying the property stated in
Lemma~\ref{lemma:normal:submodular} is called a 
{\em normal submodular function}.  If, furthermore, $g(Z)=0$ for
all $Z \prec \hat 1$ other than the co-atoms, then we say that $h$ is
{\em strictly normal}.  

For a simple example, consider the function $h$ in our running example
Fig.\ref{fig:q1}: it is strictly normal because it is defined by the
CMI $g$ shown in the figure.  A negative example is given by function
$h$ on the left of Fig. \ref{fig:non:normal}, which is not normal,
because its CMI satisfies $g(\hat 0) > 0$.  If the lattice $\bL$ is a
Boolean algebra, then the optimal polymatroid $h^*$ given 
given by Eq.\eqref{eq:bottom:up} is strictly normal:
its CMI is $g(\bX) = \sum_i v_i^*$, $g(\bX-\set{x_i}) = - v_i^*$, and $g=0$
everywhere else.

Normal polymatroids are precisely non-negative linear combinations of
``step functions''. For every $Z \in L$, {\em the
step function $h_Z$ at $Z$} is defined by $h_Z(X) = 1$ if $X \not\preceq Z$,
and $0$ otherwise.
Every step function is normal, 
because its M\"obius inverse is
$g_Z(\hat 1) = 1$, $g_Z(Z) = -1$, and $g_Z(X)=0$
otherwise. Any non-negative linear combination of step functions
is normal.  Conversely, if $h$ is a normal polymatroid, then denoting
$a_Z = -g(Z)$ for all $Z \neq \hat 1$ we have $a_Z \geq 0$ (since $h$
is normal) and $g = \sum_Z a_Z g_Z$, implying $h = \sum_Z a_Z h_Z$.

% 
% 
% \begin{lmm} \label{lemma:normal:step} $h$ is normal iff it is a
%   positive linear combination of step functions.  $h$ is strictly
%   normal iff it is a linear combination of step functions at co-atoms.
% \end{lmm}
% \begin{proof}
%   Let $h$ be a normal function and $g$ its mutual information
%   function.  For each $z \in L, z \neq \hat 1$ define $a_z = -g(z)$:
%   by definition, $g(\hat 1) = \sum_z a_z$.  Denoting $h_z$ the step
%   function at $z$ and $g_z$ its the mutual information, we have $g =
%   \sum_z a_z g_z$, and therefore $h = \sum_z a_z h_z$ by the linearity
%   of the M\"obius inversion formula.
% \end{proof}

\subsection{Connection to Quasi-Product Instance}

We have seen that the worst-case instance of the AGM-bound is a
product database instance (Theorem~\ref{th:agm}), and 
its entropy function given in Eq.\eqref{eq:bottom:up} is
strictly normal.  We generalize this observation by proving that
normal polymatroids are precisely entropy functions of quasi-product
instances (Definition~\ref{defn:quasi:product:instance}). 
For one direction we need:

\begin{lmm} \label{lemma:normal:inverse} Let
  $f : L \rightarrow L'$ be an embedding between two lattices, and
  let $h'$ be a normal polymatroid on $\bL'$.  Then
  $h \defeq h' \circ f$ is a normal polymatroid on $\bL$.
\end{lmm}

\begin{proof} Recall that an embedding is the left adjoint of a Galois
  connection.  Let $r: L' \rightarrow L$ be its right
  adjoint\footnote{The standard notation is $g$, but we use $g$ for
  the CMI.}, in other words $f(X) \preceq Y$ iff $X \preceq r(Y)$.  Let
  $g'$ be the CMI for $h'$.  The function
  $g(X) \defeq \sum_{Y: r(Y)=X} g'(Y)$ is the CMI of $h$, because,
  forall $X\in L$:
   \[
    \sum_{Z: X \preceq Z} g(Z) = 
    \sum_{Z: X \preceq Z} \sum_{Z': r(Z') = Z} g'(Z') =
    \sum_{Z': X \preceq r(Z')} g'(Z') = \sum_{Z': f(X) \preceq Z'}
      g'(Z') = h'(f(X)) = h(X)
   \]
  In any Galois connection, $r(\hat 1) = \hat 1$ (since
  $f(\hat 1) \preceq \hat 1$ iff $\hat 1 \preceq r(\hat 1)$).  By normality
  of $h'$ we have $Y \neq \hat 1$ implies $g'(Y) \leq 0$.  Therefore,
  $X \neq \hat 1$ implies $g(X)=\sum_{Y: r(Y)=X} g'(Y) \leq 0$, proving
  that $h$ is normal.
\end{proof}

In the opposite direction, we need:

\begin{defn} \label{def:canonical:embedding} Let $A = 2^\bX$ be a
  Boolean algebra.  The canonical instance is the product database
  $D = [2]^\bX$; note that $|\Pi_X(D)| = 2^{|X|}$, and thus
  $h_D(X) = |X|$, forall $X \subseteq \bX$.

  Let $\bL$ be a lattice and $h$ an integer-valued, normal
  polymatroid, and $g$ its CMI. For all $X \in L$, $X\neq \hat 1$,
  let $C(X)$ be a set of $-g(X)$ arbitrary elements, such that the sets
  $(C(X))_{X\neq \onehat}$ are disjoint.  Define $C(\onehat)=\emptyset$,
  and $\bC \defeq \bigcup_{X} C(X)$.  The
  {\em canonical Boolean Algebra} is $A = (2^\bC, \supseteq)$, and the
  {\em canonical embedding} of $\bL$ is $f : L \rightarrow A$,
  $f(X) \defeq \bigcup_{Z: X \preceq Z} C(Z)$.
\end{defn}

$A$ is an ``upside-down'' Boolean algebra, where $Y \preceq Y'$ iff
$Y \supseteq Y'$, and $\vee$ is set intersection; one can check that
$f$ commutes with $\vee$.  We can now prove our result.
(As in the AGM bound, there is a bit of loss when $h$ are not integral,
but this does affect asymptotics)

\begin{lmm} \label{lemma:normal:materialize} Let $h$ be an integral,
  non-negative, submodular function $h$ on a lattice $\bL$.  Then $h$
  is normal iff it is the entropy function of a quasi-product
  instance.
\end{lmm}
\begin{proof}
  In one direction, let $D$ is a quasi-product instance and $h_D$ be
  its entropy function.  By definition, $D = f^{-1}(D')$, where
  $f : L \rightarrow A$ is an embedding into a Boolean algebra and
  $D'$ is a product instance for $A$.  Since $h_{D'}$ is normal (even
  strictly normal), $h_D$ is also normal by
  Lemma~\ref{lemma:normal:inverse}.  In the opposite direction, assume
  $h$ is normal.  Let $f : L \rightarrow A$ be the canonical
  embedding, and $D'$ be the canonical instance for $A$.  Then $D
  \defeq f^{-1}(D')$ is the quasi-product instance, and one can check
  that $h = h_D = h_{D'} \circ f$.
%  \mak{Notationally, is $h_{D'}$ same as $h'$?}
%  \hqn{Answer: for any database D, h_D is the associated entropy function}
\end{proof}

\begin{example}
  We illustrate the construction above on the lattice $\bL$ in
  Fig.~\ref{fig:q1}, showing that it leads to the quasi-product instance
  in Example~\ref{ex:running:1}.  The canonical embedding is into the
  Boolean algebra $A= 2^{\set{a,b,c}}$ because there are three
  co-atoms in $\bL$: if we call the {\em atoms} of this Boolean
  algebra $a,b,c$, then the embedding $f : L \rightarrow
  2^{\set{a,b,c}}$ is precisely the renaming in
  Example~\ref{ex:running:1}.
  % The canonical embedding is not injective in general: in our
  % example $f(xy) = f(xyu)$.
  The polymatroid $h$ is equal to $h' \circ f$, where $h'$ is defined
  on $A$ by $h'(a)=h'(b)=h'(c)=1/2$, $h'(ab)=h'(ac)=h'(bc)=1$,
  $h'(abc)=3/2$, whose materialization is that in
  Example~\ref{ex:running:1} for $N=2$.
\end{example}

\subsection{Connection to Fractional Edge Covering}

\begin{defn} \label{def:coatomic:hypergraph} Let $(\bL, \bR)$ be a
  query given in lattice presentation.  The co-atomic hypergraph 
  $H_\co = (V_\co, E_\co)$ is defined as follows. The nodes $V_\co$ are the 
  co-atoms of $\bL$, and $E_\co = \set{e_1, \ldots, e_m}$ with $e_j =
  \setof{Z}{Z \in V_\co, R_j \not\preceq Z, R_j \in \bR}$.
  In other words, each relation $R_j$ defines a hyperedge $e_j$ consisting of 
  those nodes that do {\em not} contain the variables of $R_j$.  
  A simple illustration of a co-atomic hypergraph is in Fig.\ref{fig:hypergraph1}.
\end{defn}
%
% We will use the co-atomic hypergraph
% throughout the paper, and will refer to its weighted fractional
% edge cover or vertex packing, meaning exactly
% the same definitions as in Sec.~\ref{sec:background}, with the same
% log cardinalities $n_j$, for $j \in [m]$.  
%

\begin{lmm} \label{lemma:ineq:step} Inequality (\ref{eq:ineq1}) holds
  for all normal polymatroids $h$ iff it holds for all strictly normal
  polymatroids iff $(w_j)_{j=1}^m$ is a fractional edge cover of
  $H_\co$.
\end{lmm}
\begin{proof}
  Since each (stricly) normal polymatroid is a non-negative linear
  combination of (co-atomic) step functions, it suffices to assume
  that $h$ is a step function.  First consider a co-atomic step
  functions $h_Z$, i.e.  $Z$ is a co-atom: the left hand size of
  Eq. (\ref{eq:ineq1}) is $\sum_{j: Z \in e_j} w_j$, and the right
  hand side is $1$, meaning that the inequality holds iff
  $(w_j)_{j=1}^m$ covers node $Z \in V_\co$.  It remains to show that
  if \eqref{eq:ineq1} holds for all co-atomic step functions then it
  holds for all step functions $h_X$: let $Z$ be any a co-atom s.t. $Z
  \succeq X$, then $h_Z \leq h_X$ and $h_Z(\hat 1)=h_X(\hat 1)=1$
  implying that (\ref{eq:ineq1}) holds for $h_X$.
\end{proof}

The co-atomic hypergraph is the natural concept to capture Shearer's
lemma in a general lattice.  In fact, every vector $(w_j)_{j=1}^m$
for which the output inequality (\ref{eq:ineq1}) holds 
is a fractional edge cover of the co-atomic
hypergraph.  One may wonder whether the output inequalities could also
be described by the {\em atomic hypergraph}, defined in a similar way.
In a Boolean algebra $2^\bX$ the atomic and co-atomic hypergraphs are
isomorphic via $x \mapsto \bX - \set{x}$, since $x \in R_j$ iff $R_j
\not\subseteq \bX - \set{x}$, but in a general lattice the atomic
hypegraph does not seem to lead to any interesting properties.

\subsection{Normal Lattices}

A normal lattice is a lattice where, at optimality, the polymatroid is
normal.  This is captured by the following, which is the main result
of this section.

\begin{thm} \label{th:normal} Let $(\bL, \bR)$ be a query in lattice
  presentation. The following are equivalent:
  \begin{enumerate}
  \item \label{item:normal:1} For every non-negative, submodular
    function $h$ there exists a normal polymatroid $h'$ s.t. $h'(\hat
    1) = h(\hat 1)$ and $h'(R_j) \leq h(R_j)$ for all $R_j \in \bR$.
  \item \label{item:normal:2} The previous property holds and $h'$ is
    strictly normal.
  \item \label{item:normal:3} Inequality (\ref{eq:ineq1}) holds for
    all non-negative, submodular functions $h$ iff $w_1, \ldots, w_m$
    is a fractional edge cover of the co-atomic hypergraph.
  \item \label{item:normal:4} Every non-negative submodular function
    $h$ on $\bL$ has a materialization that is a quasi-product
    database instance.
  \end{enumerate}
  If $L$ satisfies any of these conditions, then we call it a {\em
    normal lattice w.r.t. $\bR$}.  If $L$ is normal w.r.t. any inputs
  $\bR$, then we call it shortly a {\em normal lattice}.
\end{thm}
\begin{proof}
We first prove that items \ref{item:normal:1}, \ref{item:normal:2},
and \ref{item:normal:3} are equivalent, then show that item
\ref{item:normal:2} is equivalent to \ref{item:normal:4}.

Item \ref{item:normal:2}$\Rightarrow$ item \ref{item:normal:1} is
obvious.

For Item \ref{item:normal:1}$\Rightarrow$ item \ref{item:normal:3},
let $h$ be any submodular function on $L$.  Let $h'$ be a normal
function as defined in item \ref{item:normal:1}.  Since $h'$ is
normal, it satisfies the inequality (\ref{eq:ineq1}); it suffices to
note that $\sum_j w_j h(R_j) \geq \sum_j w_j h'(R_j) \geq h'(\hat 1)
\geq h(\hat 1)$.

For Item \ref{item:normal:3}$\Rightarrow$ item \ref{item:normal:2},
let $h$ be any submodular function on $L$, and let $Z_1, \ldots, Z_k$
be all co-atoms in $L$. To define a strictly normal function $h'$, we
need to find $k$ numbers $a_i \geq 0$, $i\in[k]$ and define:
\begin{eqnarray*}
   h'(\hat 1) &= & \sum_i a_i \\
   h'(X) &= & h'(\hat 1) - \sum \setof{a_i}{i\in [k]: X \preceq Z_i} =  \sum \setof{a_i}{i\in [k]: X \not\preceq Z_i}
\end{eqnarray*}
We need to find these numbers such that $h'(\hat 1) \geq h(\hat 1)$
and $h'(R_j) \leq h(R_j)$ for all $R_j \in \bR$.  To do that, consider the
following linear program:

\begin{eqnarray*}
   \text{maximize } & a_1 + \ldots + a_k \\
\forall j \in [m]: & \sum \setof{a_i}{i \in [k], R_j \not\preceq  Z_i} \leq h(R_j)
\end{eqnarray*}

We claim that its optimal value is $\geq h(\hat 1)$: this implies that
$h'$ defined above is a strictly normal function satisfying the
requirement in the theorem.  To prove the claim, consider the dual LP.
Its variables are $b_j$, for $j\in[m]$:
\begin{eqnarray*}
   \text{minimize } & \sum_j b_j h(R_j) \\
\forall i\in [k]: & \sum \setof{b_j}{j \in [m], R_j \not\preceq Z_i} \geq 1
\end{eqnarray*}
Thus, the feasible solutions of the dual LP are precisely the
fractional edge covers of the dual hypergraph.  Hence, by assumption
in Item~\ref{item:normal:3}, the output inequality (\ref{eq:ineq1})
holds:
\begin{align*}
  \sum_j b_j h(R_j) \geq h(\hat 1)
\end{align*}
which completes the proof.

%\ds{To discuss in the paragraph below what to do when $h$ is not integral.}  
Finally, we show that item \ref{item:normal:2} holds iff
\ref{item:normal:4} holds.  For the only if direction, consider a
non-negative, submodular function $h$; by \ref{item:normal:1} there
exists a normal polymatroid $h'$ s.t. $h'(\hat 1)=h(\hat 1)$ and
$h'(R_j) \leq h(R_j)$ forall $j$.  By
Lemma~\ref{lemma:normal:materialize}, $h'$ is the entropy of some
quasi-product instance $D$.  Then $D$ is a materialization of $h$
because $\log |R^D_j| = h'(R_j) \leq h(R_j)$ and
$\log |D| = h'(\hat 1) = h(\hat 1)$.  Conversely, assume item
\ref{item:normal:4} holds: $h$ has some materialization $D$ that is a
quasi-product.  Then by Lemma~\ref{lemma:normal:materialize}, $h_D$ is
normal, and $h_D(R_j) = \log |R^D_j| \leq h(R_j)$,
$h_D(\hat 1) = \log |D| = h(\hat 1)$, proving that item
\ref{item:normal:1} holds, with $h' \defeq h_D$.
\end{proof}

% It may come as a surprise that the right notion to characterize output
% inequalities (\ref{eq:ineq1}) is the co-atomic hypergraph.  One simple
% question is why restrict to co-atoms and not include all
% meet-irreducibles: the reason is because every non-co-atom $z$ is
% below some co-atom $z'$, and therefore every hyperedge $e_{R_j}$ that
% contains $z'$ ($R_j \not\preceq z'$) also contains $z$.  Thus, with or
% without $z$, the hypergraph has the same set of fractional edge
% covers.  A deeper question is why not use the atomic hypergraph (whose
% nodes are the atoms, and hyperedges are $e_j = \setof{z}{z \preceq
%   R_j}$), which is closer in spirit to the query hypergraph.  However,
% by repeating the arguments above in the dual (replacing $\preceq$ with
% $\succeq$ and allowing decreasing submodular step functions), then we
% obtain inequalities of the form $\sum_j w_j h(R_j) \geq h(\hat 0)$,
% which is not what we want.  Thus, the atomic hypergraph, and the
% related query hypergraph, are {\em not} the right notions to capture
% output inequalities (\ref{eq:ineq1}).

Normal lattices appear to be a new concept.  It is decidable whether a
lattice $\bL$ is normal w.r.t. $\bR$, using the following naive
procedure.  Enumerate all vertices $(w_j)_j$ of the fractional edge
packing polytope of the co-atomic hypergraph, and check that the
output inequality \eqref{eq:ineq1} holds, by using the criterion in
Lemma~\ref{lmm:generalized:inequality}.

We prove in Sec.~\ref{sec:sma} that every distributive lattice is
normal; furthermore,
Proposition~\ref{prop:simple:FD:implies:distributive:lattice} says
that any set of unary FDs generate a distributive lattice, which is
therefore normal.  The lattice on the left of Fig.\ref{fig:non:normal}
is a Boolean algebra, hence it is normal: note that $h$ in the figure
is not normal\footnote{The function XOR on three variables, $R(x,y,z)
  = \setof{(a,b,c)}{a,b,c \in \set{0,1}, a \text{ xor } b \text{ xor }
    c = 0}$, is the canonical example of a distribution whose entropy
  has a negative mutual information.}, but we can simply increase
$h(\hat 1)$ to 3 and now it is normal.

The lattice $\bL$ in Fig.\ref{fig:q1} is normal w.r.t. inputs
$xy,yz,zu$, which follows by exhaustively proving all inequalities
defined by the fractional edge covers of the co-atomic hypergraph,
shown in Fig.~\ref{fig:hypergraph1}.  For example, the edge cover
$(1/2,1/2,1/2)$ corresponds to the inequality $h(xy)+h(yz)+h(zu) \geq
2h(\hat 1)$, which can be proven by $h(xy)+h(yz)\geq h(\hat 1) + h(y)$
and $h(y) + h(zu) \geq h(\hat 1) + h(0)$.  In fact, $\bL$ is normal
w.r.t. any inputs.  Notice that $\bL$ is not distributive.
% ,
% because it contains the sub-lattice $\set{\hat 0, xy, z, zu, \hat 1}$
% isomorphic to $N_5$ (the standard non-modular, hence non-distributive
% lattice).

The lattice $M_3$ on the right of Fig.\ref{fig:non:normal} is not normal.
Its co-atomic hypergraph has edges $e_x = \set{y,z}$, $e_y =
\set{x,z}$, $e_z = \set{x,y}$ and the inequality corresponding to the
fractional edge cover $(1/2,1/2,1/2)$, $h(x) + h(y)+h(z) \geq 2h(\hat
1)$, fails for the submodular function shown in the figure.  In
particular, this polymatroid $h$ is not materializable as a
quasi-product database instance, but can be materialized as
$\setof{(i,j,k)}{i,j,k \in \set{0,1}, i+j+k=0 \mod 2}$.

We give a necessary condition for normality (and we conjecture it is also
sufficient):

\begin{prop}
  Let $L$ be a lattice that contains a sublattice $\set{U,X,Y,Z,\hat
    1}$ isomorphic to $M_3$, and let $\bR = \set{X,Y,Z}$.  Then $L$
  is not normal.
\end{prop}

\begin{proof}
  The co-atomic hypergraph has three hyper-edges, $X,Y,Z$.  Then
  $(1/2,1/2,1/2)$ is a fractional edge cover of the co-atomic
  hypergraph, because every co-atom is above at most one of $X,Y,Z$,
  because $X \vee Y = X \vee Z = Y \vee Z = \hat 1$, hence it belongs
  to at least two hyper-edges of the hypergraph.  This defines the
  inequality:
  \[
    h(X) + h(Y) + h(Z) \geq 2 h(\hat 1).
 \]
  We construct a polymatroid $h$ that violates this inequality:
   \[
      h(W) = \begin{cases}
         0 & \mbox{ when } W \preceq U \\
		1 & \mbox{ when } W \preceq X \mbox{ or } W \preceq Y \mbox{ or } W
         \preceq z \\
		2 & \mbox{ when } W \mbox{ is not below $X$ or $Y$ or $Z$}
      \end{cases}
   \]
  Then $h(X)=h(Y)=h(Z)=1$, $h(\hat 1)=2$ and therefore $h$ violates
  the inequality above.  We prove that $h$ is a polymatroid.
  Monotonicity is easy to check.  We prove that it satsifies the
  submodularity laws:

  \begin{align*}
	h(A) + h(B) \geq h(A \vee B) + h(A \wedge B)
  \end{align*}

Case 1: $h(A \vee B) = 0$.  Then all elements are below $U$, and all $h$ are 0.

Case 2: $h(A \vee B) = 1$.  If $h(A) = 1$ then the inequality follows
from $h(B) \geq h(A \wedge B)$, and similarly for $h(B) = 1$.  If both
$h(A)=h(B)=0$ then both $A,B$ are below $U$, hence $A \vee B$ is below
$U$, contradicting $h(A \vee B) = 1$

Case 3: $h(A \vee B) = 2$.  We can rule out the cases when $h(A),h(B)$
are 0,0 or 0,1 because that implies both $A,B$ are below $X$ (or $Y$
or $Z$), implying $h(A \vee B) = 1$.  Also, the case when $h(A)=2$ or
$h(B)=2$ follows immediately from monotonicity.  Thus, assume
$h(A)=h(B)=h(A \wedge B)=1$.  Here we use the structure of $M_3$.  If
$A$ and $B$ are below two distinct elements from $X,Y,Z$, e.g. $A \leq
X$ and $B \leq Y$, then $A\wedge B \leq U$ contradicting $h(A\wedge
B)=1$.  Hence both $A,B$ are below $X$.  But that implies $A \vee B$
is also below $X$, contradicting $h(A \vee B)=2$.
\end{proof}

% We prove in the appendix:
% \begin{prop}
%   We will show in Sec.~\ref{sec:sma} that every distributive lattice
%   is normal.  We also show in the Appendix that, if all FD's are
%   simple (meaning: of the form $x \rightarrow y$ where $x, y$ are
%   single variables) then $L$ is distributive, hence normal.
% \end{prop}

% The last item of the theorem implies that, for a non-normal lattice,
% the worst-case instance $D$ (whose bound is $|Q^D| = 2^{h^*(\hat 1)}$)
% cannot be quasi-product.  
% This explains and justifies the complexity
% required by the construction by Gogacz and
% Toru\'nczyk~\cite{szymon-2015}.

Gottlob et al.~\cite{GLVV} define a {\em coloring} of a query $Q$ with
variables $\bX$ to be a function
$\coloring : \bX \rightarrow 2^{\bX'}$ such that
$\coloring(\bX)\neq \emptyset$, and for any FD $X\rightarrow Y$,
$\coloring(Y) \subseteq \coloring(X)$, where, for each set $X$,
$\coloring(X) \defeq \bigcup_{x \in X} \coloring(X)$.  Then the {\em
  color number of $\coloring$} is defined as
$C(\coloring) = |\coloring(\bX)|/\max_j |\coloring(R_j)|$ and the {\em
  color number of $Q$} is $\max_\coloring C(\coloring)$.  They prove
two results: if the functional dependencies are restricted to simple
keys then $|Q^D| \leq \left(\max_j |R_j^D|\right)^{C(Q)}$, and
moreover this bound is essentially tight, even for general functional
dependencies.
%\mak{I thought in the presence of general functional dependencies, the 
%coloring number does NOT even provide an upper bound on $|Q^D|$; instead, 
%it provides a \emph{lower bound} on the worst-case $|Q^D|$. (This was 
%essentially what Marx pointed out to the authors of \cite{GLVV}.)}
%\hqn{Note the keyword : {\em simple key}}

Colorings correspond one-to-one to integral, normal polymatroids, via
$h(x) = |\bigcup_{x \in X} \coloring(x)|$.  To see this, in one
direction let $\coloring : \bX \rightarrow 2^{\bX'}$ be a coloring and
define $h(X) \defeq |\bigcup_{x \in X} \coloring(x)|$.  The function
$f(X) \defeq \bigcup_{x \in X} \coloring(x)$ is an embedding
$f : \bL \rightarrow 2^{\bX'}$ into a Boolean algebra (we assume
w.l.o.g. that $f(\bX) = \bX'$, otherwise we redefine $\bX'$), and we
have $h = h' \circ f$ where $h'(Y) \defeq |Y|$: since $h'$ is a
strictly normal polymatroid Lemma~\ref{lemma:normal:inverse} says that $h$
is also normal.  In the other direction, if $h$ is an integral, normal
polymatroid, then its canonical embedding
(Definition~\ref{def:canonical:embedding}) defines a coloring
$\coloring$ s.t.  $h(X) \defeq |\bigcup_{x \in X} \coloring(x)|$.
Thus, colorings are essentially normal polymatroids.  This implies the
two results in~\cite{GLVV} as follows.  If $\FD$ consists only of
simple keys, then the lattice $\bL$ is distributive
(Prop.~\ref{prop:simple:FD:implies:distributive:lattice}), hence it is
normal (Corollary~\ref{cor:every:dist:lattice:is:normal}), hence at optimality 
$h^*$ can be assumed to be a normal polymatroid (Theorem~\ref{th:normal}
item~\ref{item:normal:1}), equivalently a coloring.  The second
result, tightness, follows from our
Lemma~\ref{lemma:normal:materialize}, since any normal polymatroid
(coloring) is the entropy of a quasi-product database instance.

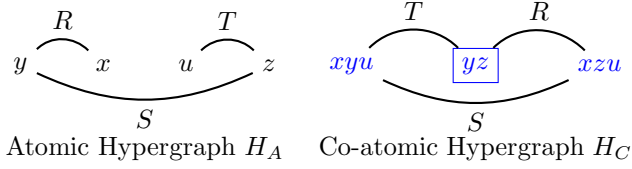
\begin{figure}
   \centering
\begin{tikzpicture}[domain=0:20, scale=0.55]
  \node[] at (2,2) (y) {$y$};
  \node[] at (4,2) (x) {$x$};
  \node[] at (6,2) (u) {$u$};
  \node[] at (8,2) (z) {$z$};
  \path[thick] (y) edge [out=45, in=135] node [above] {$R$} (x);
  \path[thick] (u) edge [out=45, in=135] node [above] {$T$} (z);
  \path[thick] (y) edge [out=-25, in=-155] node [below] {$S$} (z);

  \node[blue] at (10,2) (xyu) {$xyu$};
  \node[draw,rectangle,blue] at (13,2) (yz) {$yz$};
  \node[blue] at (16,2) (xzu) {$xzu$};
  \path[thick] (xyu) edge [out=45, in=135] node [above] {$T$} (yz);
  \path[thick] (yz) edge [out=45, in=135] node [above] {$R$} (xzu);
  \path[thick] (xyu) edge [out=-25, in=-155] node [below] {$S$} (xzu);

  \node[] at (5,0) {Atomic Hypergraph $H_A$};
  \node[] at (13,0) {Co-atomic Hypergraph $H_C$};
  \end{tikzpicture}
\caption{Atomic and Co-atomic hypergraphs for Fig.\ref{fig:q1}} 
\label{fig:hypergraph1}
\end{figure}

\begin{figure}
   \centering
\begin{tikzpicture}[domain=0:20, scale=0.5]
  \node[draw,rectangle,blue] at (7,3) (x1) {$x$};
  \node[draw,rectangle,blue] at (10,3) (y1) {$y$};
  \node[draw,rectangle,blue] at (13,3) (z1) {$z$};
  \node[] at (10,5) (11) {$\hat 1$};
  \node[] at (10,1) (01) {$\hat 0$};
  \path[] (01) edge (x1);
  \path[] (01) edge (y1);
  \path[] (01) edge (z1);
  \path[] (x1) edge (11);
  \path[] (y1) edge (11);
  \path[] (z1) edge (11);
  \node[thick, Green, left=-0.1 of x1] {\tiny $1$};
  \node[thick, red, right=-0.1 of x1] {\tiny $-1$};
  \node[thick, Green, left=-0.1 of y1] {\tiny $1$};
  \node[thick, red, right=-0.1 of y1] {\tiny $-1$};
  \node[thick, Green, left=-0.1 of z1] {\tiny $1$};
  \node[thick, red, right=-0.1 of z1] {\tiny $-1$};
  \node[thick, Green, left=-0.1 of 01] {\tiny $0$};
  \node[thick, red, right=-0.1 of 01] {\tiny $+1$};
  \node[thick, Green, left=-0.1 of 11] {\tiny $2$};

  \node[] at (1,0) (0) {$\hat 0$};
  \node[] at (-2,2) (x) {$x$};
  \node[] at (1,2) (y) {$y$};
  \node[] at (4,2) (z) {$z$};
  \node[blue] at (-2,4) (xy) {$xy$};
  \node[blue] at (1,4) (xz) {$xz$};
  \node[blue] at (4,4) (yz) {$yz$};
  \node[] at (1,6) (1) {$\hat 1$};

  \node[thick, Green, left=-0.1 of x] {\tiny $1$};
  \node[thick, red, right=-0.1 of x] {\tiny $-1$};
  \node[thick, Green, left=-0.1 of y] {\tiny $1$};
  \node[thick, red, right=-0.1 of y] {\tiny $-1$};
  \node[thick, Green, left=-0.1 of z] {\tiny $1$};
  \node[thick, red, right=-0.1 of z] {\tiny $-1$};
  \node[thick, Green, left=-0.1 of xy] {\tiny $2$};
  \node[thick, red, right=-0.1 of xy] {\tiny $0$};
  \node[thick, Green, left=-0.1 of xz] {\tiny $2$};
  \node[thick, red, right=-0.1 of xz] {\tiny $0$};
  \node[thick, Green, left=-0.1 of yz] {\tiny $2$};
  \node[thick, red, right=-0.1 of yz] {\tiny $0$};
  \node[thick, Green, left=-0.1 of 0] {\tiny $0$};
  \node[thick, red, right=-0.1 of 0] {\tiny $+1$};
  \node[thick, Green, left=-0.1 of 1] {\tiny $2$};

  \path[] (0) edge (x);
  \path[] (0) edge (y);
  \path[] (0) edge (z);
  \path[] (x) edge (xy);
  \path[] (x) edge (xz);
  \path[] (y) edge (xy);
  \path[] (y) edge (yz);
  \path[] (z) edge (xz);
  \path[] (z) edge (yz);
  \path[] (xy) edge (1);
  \path[] (yz) edge (1);
  \path[] (xz) edge (1);

  \end{tikzpicture}
\caption{A non-normal function $h$ (left) and the non-normal lattice $M_3$ (right)}
\label{fig:non:normal}
\end{figure}
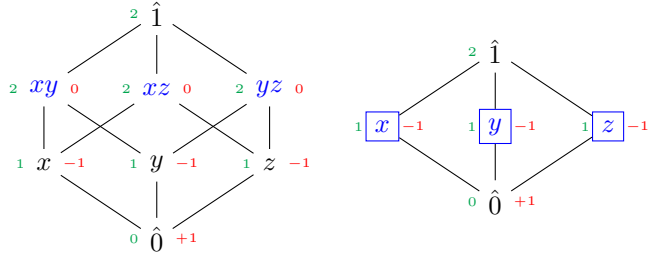

To summarize, we have introduced a framework for studying the query
upper bound under FDs that uses the lattice of its closed sets of
attributes, which extends and generalizes the framework introduced by
Gottlob et al.~\cite{GLVV}.

%!TEX root = main.tex

\section{Proof Sequence to Algorithm}

\label{sec:proof:algorithms}

We aim to design an algorithm that, given an input query $Q$ and
database instance $D$, computes $Q^D$ in time $\tilde O(2^{h^*(\hat
1)})$, where $h^*$ is the optimal solution to the LLP~\eqref{eq:llp}.  
From Lemma~\ref{lmm:generalized:inequality}, 
if $(s^*,w^*)$ is dual-optimal then $h^*(\hat 1) = \sum_{j=1}^mw^*_jn_j$
and the output inequality \eqref{eq:ineq1} holds with $w_j=w^*_j$, for 
{\em all} polymatroids.
Our main algorithmic theme is to turn this process ``inside-out'':
find a ``proof'' of inequality \eqref{eq:ineq1} consisting of a sequence of 
steps transforming the left-hand-side into the right-hand-side, then
interpret these steps as algorithmic steps whose runtime does not
exceed what the symbolic step allows.

We show that three different proof techniques for (\ref{eq:ineq1}) lead to 
three different bounds and algorithms. 
The first two are known techniques for proving Shearer's lemma
(i.e. the Boolean algebra case) that we adapt to lattices.
The bounds are tight and corresponding algorithms are worst-case optimal for some 
classes of lattices (such as distributive lattices, which include the Boolean
algebra and the lattice of queries with simpld fd's, 
subsuming results from~\cite{DBLP:conf/pods/NgoPRR12,LFTJ,GLVV}).
However, for a generic lattice the bounds these techniques can prove are 
not $h(\hat 1) \leq \sum_j w^*_jh(R_j)$, and thus the corresponding algorithms' 
runtimes in general can be worse than $\tilde O(2^{h^*(\hat 1)})$. 
An advantage of these two algorithms is that there is only one $\log$-factor 
hidden in the $\tilde O$. 

The last proof technique is our own, designed specifically to achieve the
optimal LLP bound in an arbitrary lattice; this leads to an algorithm called 
CSMA runing in the stated
time of $\tilde O(2^{h^*(\hat 1)})$, where $\tilde O$ hides a polylogarithmic 
factor. This algorithms
needs to regularize the data, somewhat similar in spirit to the
uniformization step in~\cite{MR3144912}. 

%!TEX root = main.tex

\subsection{Chain Bound and Chain Algorithm (CA)}
\label{sec:ca}

Our first proof sequence adapts Jaikumar Radhakrishnan's proof of
Shearer's lemma \cite{radhakrishnan} to lattices,
which is based on an arbitrary but fixed order
of the random variables.  We observe that a variable ordering
corresponds to a maximal chain in the Boolean algebra. This idea
allows us to generalize Radhakrishnan's proof to general lattices.
Fix a chain $\bC: \hat 0 = C_0 \prec C_1 \prec \cdots \prec C_k = \hat
1$ in $L$.  For $X \in L$, we say that {\em $X$ covers $i$} if $X
\wedge C_i \neq X \wedge C_{i-1}$.  Intuitively, $X$ covers $i$ if it
contains some variable in $C_i$ that does not appear in $C_{i-1}$.
Fix a query represented by $(\bL, \bR)$, where $\bR = (R_1, \ldots, R_m)$.

\bdefn The {\em chain hypergraph} associated with a
chain $\bC$ is $H_{\bC} = ([k], \set{e_j \suchthat j \in [m]})$, where the
hyperedge $e_j$ contains all nodes $i$ such that $R_j$ covers $i$.
\edefn

In a Boolean algebra all maximal chains have the same hypergraph,
which is the same as the query hypergraph, and the co-atomic
hypergraph.  But in a general lattice all these hypergraphs may be different, and
some output inequalities (\ref{eq:ineq1}) can be derived only from
non-maximal chains.  For that reason, we relax the maximality
requirement on the chain, as follows.  We say that the chain $\bC$ is
{\em good} for $R_j$ if:
\begin{align}
   \text{ for all } i \in [k]: \quad &
   i \in e_j  \Rightarrow C_{i-1} \vee (R_j \wedge C_i) = C_i.
   \label{eq:good}
\end{align}
The key property of ``goodness'' is that submodularity applies in
the following way:

\begin{prop}
  If $\bC$ is a maximal chain, then it is good for any $R_j$.
  Furthermore, if $\bC$ is good for $R_j$, then 
  \begin{align}
    h(R_j \wedge C_i) - h(R_j \wedge C_{i-1}) \geq h(C_i) -  h(C_{i-1}),
    \ \ \ \forall i \in e_j \label{eq:chain:cond1}
  \end{align}
  for every $\bL$-submodular function $h$.
\end{prop}
\bp
Notice that $Z \defeq C_{i-1} \vee (R_j \wedge C_i)$ is always in the interval
$[C_{i-1}, C_i]$.  In any chain, $i\not\in e_j$ iff $R_j \wedge C_i =
R_j \wedge C_{i-1}$ iff $Z = C_{i-1}$.  
Therefore, any maximal chain is good for $R_j$.
To show~\eqref{eq:chain:cond1}, replace 
$C_i$ with $(R_j \wedge C_i) \vee C_{i-1}$ and the inequality 
becomes the (lattice) submodularity constraint for the elements $(R_j \wedge C_i)$ 
and $C_{i-1}$.  
\ep
We say a chain $\bC$ is {\em good} for $\bR$ if it is good
for all $R_j \in \bR$.
Radhakrishnan's proof is adapted to a lattice as follows.

\bthm[The Chain Bound] 
Let $\bC$ be any chain that is good for $\bR$.
If $(w_j)_{j=1}^m$ is any fractional edge cover of the chain hypergraph 
$H_{\bC}$,
then inequality (\ref{eq:ineq1}) holds for
any polymatroid $h$
\ethm
\bp By writing $h(R_j)$ as a telescoping sum, we obtain
\begin{eqnarray*}
   \sum_{j=1}^m w_j h(R_j) 
   & = & \sum_{j=1}^m w_j \cdot \left(\sum_{i \in e_j} (h(R_j\wedge C_i) - h(R_j\wedge
C_{i-1})) \right)\\
(\text{from Eq.\eqref{eq:chain:cond1}}) & \geq & \sum_{j=1}^m w_j \cdot \left(
    \sum_{i \in e_j} (h(C_i) - h(C_{i-1})) \right)\\
& = & \sum_{i=1}^k  \left(\sum_{j:  i\in e_j} w_j\right) \cdot
    (h(C_i) - h(C_{i-1}))  %\geq h(\hat 1)
 \\
 & \geq & \sum_{i=1}^k (h(C_i) - h(C_{i-1})) \\
 &=& h(\hat 1),
\end{eqnarray*}
where the last inequality holds because $(w_j)_{j=1}^m$ is a
fractional edge cover of $H_{\bC}$ and $h$ is monotone.
\ep
\brmk
Note that, if $\bC$ is only good for a subset $\bR'$ of $\bR$, then we can apply 
the bound to $\bR'$ with its corresponding chain hypergraph.
\ermk

\begin{example} \label{ex:cb:q1}
  Consider our running example from Figure~\ref{fig:q1}; assume $|R|=|S|=|T|=N$.
  Consider the chain $\hat 0 \prec y \prec yz \prec xyzu = \hat 1$,
  whose hypergraph has three vertices and hyperedges $e_R = \{y, xyzu\},\ e_S = \{y,
   yz\},\ e_T = \{yz, xyzu\}$
   (isomorphic to the co-atomic hypergraph in Fig.\ref{fig:hypergraph1}). 
   Thus, the chain bound on this chain is $N^{3/2}$, which is tight.
  (Consider the input $R=S=T=[\sqrt N] \times [\sqrt N]$, where the fd
   $xz\to u$ is defined by the UDF $f(x,z)=x$ and $yu \to x$ by $g(y,u)=u$.)
\end{example}

If there were no FD's, then the chain bound is exactly Shearer's lemma
\cite{MR859293} (or, equivalently, AGM bound).  
\bcor[AGM bound and
Shearer's lemma] Consider a join query on $n$ variables with no fd's.  
The chain bound on the chain ${\bC}$: $C_0 = \emptyset \prec
C_1 = [1] \prec C_2 = [2] \prec \cdots \prec C_n = [n]$ is exactly
Shearer's lemma.  
\ecor

\begin{algorithm}[th]
  \caption{The Chain Algorithm.}
  \label{algo:chain}
  \begin{algorithmic}[1]
     \Require{A query $(\bL,\bR)$, over $\bR=\{R_1, \ldots, R_m\}$}
    \Require{A good chain $\bC : \zerohat = C_0 \prec C_1 \prec \cdots \prec C_k = \onehat$}
    \Require{Chain hypergraph $H_{\bC} = ([k],\{e_1,\dots,e_m\})$}
    \State Expand $Q$ to $Q^+$
    \Comment{Thus $\text{vars}(R_j)=R_j$, $j\in [m]$}
    \State $Q_0 = \set{()}$
    \For {$i=1$ {\bf to} $k$}
       \State $Q_i = \emptyset$
       \For {each $t \in Q_{i-1}$}
          \State $T = \bigcap_{j: i \in e_j} (t \Join \Pi_{R_j \wedge
          C_i}(R_j))^+$ \Comment{Takes time $\displaystyle{\tilde O\left(\min_{j: i \in e_j}
                |t \Join \Pi_{R_j \wedge C_i}(R_j)|\right)}$. See text.}
          \label{ln:computeT}
          \State $Q_i = Q_i \cup T$
       \EndFor
    \EndFor
    \State \Return{$Q_k$}
  \end{algorithmic}
\end{algorithm}

{\bf The Chain Algorithm.} 
In the proof above of the chain bound, the main idea is to take mixtures
of conditional entropies, climbing up the chain. 
This strategy corresponds combinatorially to conditional search. From this proof, 
we derive Algorithm~\ref{algo:chain} that
computes a query $Q$ in time bounded by {\em any} fractional
edge cover of the chain hypergraph.  The algorithm assumes a fixed, good chain
${\bC}$, where every node $i$ is covered\footnote{If $i$ is not covered,
   in other words $H_{\bC}$ has an isolated vertex, then $\rho^*(H_{\bC}) = \infty$, 
and the algorithm will not work.}. It preprocesses input relations by
indexing them in an attribute order consistent with the chain.  Then,
it starts by expanding the query, as explained at the
end of Sec.~\ref{sec:background}.
Thus far it takes $\tilde O(N)$-time. Next,
the algorithm computes inductively
\[ Q_i \defeq \left( \Join_{j : R_j \wedge C_i \neq \hat 0} \Pi_{R_j\wedge
   C_i}(R_j) \right)^+,
\]
%\mak{Don't we need to include $Q_{i-1}$ in the above join in order to define $Q_i$ correctly?}
%\hqn{No: this is the definition of what we are going to compute, not the
%algorithm, which requires $Q_{i-1}$}
for $i=0,1,\ldots,k$, where $Q_i$ is an intermediate relation with
attributes $C_i$. 
Initially $Q_0$ consists of just the empty tuple.  
Evidently, when $i=k$, $Q_k$ is the output $Q$.

To compute $Q_i$, let $R_j$ be some relation such that $i \in e_j$, in other 
words $R_j$
has some new variable that occurs in $C_i$ but not in $C_{i-1}$. Note
that there exists at least one such $R_j$, because $i$ is covered.
Define 
$$T_{ij} = Q_{i-1} \Join \Pi_{C_i \wedge R_j}(R_j),$$ 
whose attributes are $X_{ij} \defeq C_{i-1} \cup (C_i \wedge R_j)$; and, by
Eq.\eqref{eq:good}, its closure is $X_{ij}^+ = C_{i-1} \vee (C_i
\wedge R_j) = C_i$.  Consider $T_{ij}$'s expansion $T_{ij}^+$ (see
Sec.~\ref{sec:background}): it has the same size as $T_{ij}$ and has
attributes $X_{ij}^+ = C_i$.  Therefore, 
%our desired $Q_i$ is exactly $\bigcap_{j: i \in e_j} T_{ij}^+$. 
\[ Q_i = \bigcap_{j: i \in e_j} T_{ij}^+ =
   \bigcap_{j: i \in e_j} \left( Q_{i-1} \Join \Pi_{C_i \wedge R_j}(R_j)
   \right)^+.
\]
However, we do not want to compute all the 
$T_{ij}^+$ and then compute the intersection to obtain $Q_i$, because this 
na\"ive strategy will push the runtime over the budget.
In order to stay within the time
budget, the algorithm computes this intersection differently: it
iterates over all tuples $t \in Q_{i-1}$, and for each tuple computes
the intersection $T$ in line~\ref{ln:computeT} in time $\tilde
O(\min_{j: i \in e_j}(|t \Join \Pi_{R_j \wedge C_i}(R_j)|)$.  This can
be accomplished by first computing $j_*= \argmin_{j: i \in e_j}(|t
\Join \Pi_{R_j\wedge C_i}(R_j)|)$ and tentatively setting $T = (t \Join
\Pi_{R_{j_*}\wedge C_i}(R_{j_*}))^+$.  Then, the algorithm removes from $T$ any
tuple $t'$ that is not in the intersection defined in
line~\ref{ln:computeT}.  A tuple $t' \in T$ is not removed from $T$
only if the following holds: for {\bf every} $j \neq j_*$ s.t. $R_j$
covers $i$, we have $\Pi_{C_i \wedge R_j}(t') \in \Pi_{C_i \wedge
  R_j}(R_j)$ {\bf and} $(t \Join \Pi_{C_i \wedge R_j}(t'))^+ = t'$.  
Note the crucial fact that the relation $R_{j_*}$ that is used to iterate over
may depend on the tuple $t \in Q_{i-1}$.  
Due to the pre-processing step where every input relation is indexed with
an attribute order consistent with the chain, 
$j_*$ can easily be computed in logarithmic time in data complexity.

\begin{thm} \label{th:chain} Assume the chain ${\bC}$ is good for
   $\bR$, and every node $i$ is covered (i.e. no isolated vertices).
  Then, for any fractional edge cover of the chain hypergraph,
  $(w_j)_{j=1}^m$, the time taken by the Chain Algorithm to compute $Q$
  is $\tilde O(N+\prod_{j=1}^mN_j^{w_j})$, where 
  $\tilde O$ hides a logarithmic factor needed for index lookup or 
  binary search, and a small polynomial factor in query complexity. 
\end{thm}
\begin{proof}
  We define some notation.  For any tuple $t$ and each relation $R_j$
  such that $R_j \wedge C_i \neq \hat 0$, denote $n_{ijt} = |t \Join
  \Pi_{R_j \wedge C_i}(R_j)|$.  Note that if $t=()$ is the empty
  tuple, then $n_{ijt} = n_{ij()} = |\Pi_{R_j \wedge C_i}(R_j)|$.

Fix an iteration $i$.
Denote by $\text{Time}_i$ the time taken by
iteration $i$ (to compute $Q_i$). We will show that, for any $i \in [k]$,
\begin{equation}
   \text{Time}_i = \tilde O\left(\prod_{j : R_j \wedge C_i \neq \hat 0}
   n_{ij()}^{w_j}\right). 
                   \label{eqn:timei}
\end{equation}
Since $n_{ij()} \leq |R_j|$, the sum over all $\text{Time}_i$ will be at most
$\tilde O(\prod_{j=1}^m |R_j|^{w_j})$, as desired.
To bound $\text{Time}_i$, note that the number of steps taken to compute $Q_i$ is 
\[
   \sum_{t \in Q_{i-1}} \min_{j : i \in e_j} n_{ijt}
   \leq
   \sum_{t \in Q_{i-1}} \prod_{j : i \in e_j} n_{ijt}^{w_j}
   =
   \sum_{t \in Q_{i-1}} \prod_{j : R_j\wedge C_i \neq \hat 0} n_{ijt}^{w_j}.
\]
The equality follows from the fact that, if $t \in Q_{i-1}$ and 
$R_j \wedge C_i \neq \hat 0$ but $i \notin e_j$, then $n_{ijt} = 1$.
Hence, to show~\eqref{eqn:timei}, it is sufficient to show the following:
\begin{equation}
   \sum_{t \in Q_{i-1}} \prod_{j : R_j\wedge C_i \neq \hat 0} n_{ijt}^{w_j}
   \leq \prod_{j : R_j \wedge C_i \neq \hat 0} |n_{ij()}|^{w_j}.
  \label{eqn:timei2}
\end{equation}
For $0 \leq \ell \leq i-1$, define 
%\[ T_\ell \defeq \Pi_{C_\ell}(Q_{i-1}), \text{ and }
%g(\ell) \defeq
% \sum_{t \in T_{\ell}} \prod_{j : R_j\wedge C_i \neq \hat 0}
% n_{ijt}^{w_j}.
% \]
\begin{eqnarray*}
   T_\ell &\defeq& \Pi_{C_\ell}(Q_{i-1}), \\
   g(\ell) &\defeq&
 \sum_{t \in T_{\ell}} \prod_{j : R_j\wedge C_i \neq \hat 0}
 n_{ijt}^{w_j}.
 \end{eqnarray*}
By convention, $T_0$ has a single empty tuple $t=()$.  
Then, \eqref{eqn:timei2} is equivalent to $g(i-1) \leq g(0)$. 
Thus, to show~\eqref{eqn:timei2} it is sufficient to show that $g(\ell)$ is 
non-increasing in $\ell$:
\begin{eqnarray*}
g(\ell) &=&\sum_{t \in T_\ell} \prod_{j : R_j\wedge C_i \neq \hat 0} n_{ijt}^{w_j}\\
&=&
\sum_{u \in T_{\ell-1}}
\sum_{\substack{v :\\(u,v) \in T_{\ell}}}
\prod_{j : R_j\wedge C_i \neq \hat 0} n_{ij(u,v)}^{w_j}\\
&=&
\sum_{u \in T_{\ell-1}}
\prod_{\substack{j : R_j\wedge C_i \neq \hat 0\\ i-1 \notin e_j}} n_{iju}^{w_j}
\sum_{\substack{v :\\(u,v) \in T_\ell}}
\prod_{\substack{j : R_j\wedge C_i \neq \hat 0\\ i-1\in e_j}} n_{ij(u,v)}^{w_j}\\
&\leq&
\sum_{u \in T_{\ell-1}}
\prod_{\substack{j : R_j\wedge C_i \neq \hat 0\\ i-1 \notin e_j}} n_{iju}^{w_j}
\prod_{\substack{j : R_j\wedge C_i \neq \hat 0\\ i-1\in e_j}} 
\left(\sum_{\substack{v : \\(u,v) \in T_{\ell}}}
n_{ij(u,v)}\right)^{w_j}\\
&\leq &
\sum_{u \in T_{\ell-1}}
\prod_{\substack{j : R_j\wedge C_i \neq \hat 0\\ i-1 \notin e_j}} n_{iju}^{w_j}
\prod_{\substack{j : R_j\wedge C_i \neq \hat 0\\ i-1\in e_j}} 
n_{iju}^{w_j}\\
&=&
\sum_{u \in T_{\ell-1}}
\prod_{j : R_j\wedge C_i \neq \hat 0} n_{iju}^{w_j}\\
&=& g(\ell-1).
\end{eqnarray*}
The first inequality is generalized H\"older inequality, which applies because
$(w_j)_{j=1}^m$ fractionally cover vertex $i-1$ of the chain hypergraph $H_{\bC}$.
The second inequality holds because, for every $u \in T_{\ell-1}$,
$u \Join \Pi_{R_j\cap C_i}(R_j) \supseteq T_\ell$.
\end{proof}
We present some examples.

\begin{example} \label{ex:ca:q1}
  Continue with Example~\ref{ex:cb:q1}.
  The variable order corresponding to
  the chain $\zerohat \prec y \prec yz \prec \hat 1$
  is $y, z, (xu)$, where $x, u$ can be arranged in any order.  
  The Chain Algorithm computes three intermediate relations:
\begin{eqnarray*}
   Q_1(y)  &=& \Pi_y(R(xy)) \cap \Pi_y(S(yz)) \\
   Q_2(yz) &=& Q_1(y) \Join S(yz) \\
   Q_3(yzxu) &=& (Q_2(yz) \Join R(xy))^+ \cap (Q_2(yz) \Join T(zu))^+
\end{eqnarray*}
The first two steps are straightforward.  In the third, the join
$Q_2(yz) \Join R(xy)$ results in a relation with attributes $xyz$,
which needs to be expanded with $u$ (e.g. by computing the
user-defined function $u=f(x,z)$), and similarly for the second join.
However, the algorithm does not compute the joins first then intersect,
instead it iterates over tuples $t \in Q_2(yz)$ and computes an
intersection on a per-tuple basis using the less expensive option. 
In particular, for each $t = (y,z) \in Q_2$, it compares 
$|t \Join R(x,y)|$ and $|t \Join T(z,u)|$ (which can be done in logarithmic
time given that $R$ was indexed with attribute order $(y,x)$
and $T$ with order $(z,u)$).
Suppose $|t \Join R(x,y)| \leq |t \Join T(z,u)|$, then for 
each $(y,z,x) \in t \Join R(x,y)$ the algorithm uses the $xz \to u$ FD to
obtain the tuple $(y,z,x,u)$. The next task is to verify that this tuple
is indeed in the intersection as defined in line~\ref{ln:computeT}. This is done with
{\em two} verifications: we make sure that $(z,u) \in T$, {\bf and} that $yu \to
x$ is indeed satisfied.\footnote{This is a subtle step in the algorithm that is 
easy to miss at the first read.}

On this chain the algorithm runs in
optimal time $O(N^{3/2})$.  We note that all previously proposed known
worst-case optimal join algorithms for queries without
FD's~\cite{skew,DBLP:conf/pods/NgoPRR12,LFTJ} require $\Omega(N^2)$ to
compute the previous query on this instance:
$R = S = T = \{(1, i) \suchthat i \in [N/2]\} \cup \{(i, 1)
\suchthat i \in [N/2]\}.$ For example, LFTJ with variable order $y,
z, x, u$ computes queries $Q_1(y),Q_2(yz),Q_3(xyz),Q_4(xyzu)$, where
$|Q_3| = N^2$.
Note, however, that not every maximal chain gives an optimal bound: for 
example the chain
$\hat 0 \prec x \prec xu \prec xyu \prec xyzu=\hat 1$ has hyperedges:
$e_R = \{x, xyu\}$, $e_S = \{xyu, xyzu\}$, $e_T = \{xu, xyzu\}$
(isomorphic to the atomic hypergraph in Fig.\ref{fig:hypergraph1}),
and the optimal fractional edge covering number $\rho^*=2$, hence the
chain bound is $|Q^D| \leq N^2$, which is sub-optimal.
\end{example}

\begin{figure}[t]
\centering
\begin{tikzpicture}[domain=0:20, scale=0.6]
\node[] at (7.5,0) (0) {$\hat 0$};
\node[] at (0,2) (12) {$a$};
\node[] at (3,2) (13) {$b$};
\node[] at (6,2) (14) {$c$};
\node[] at (9,2) (23) {$d$};
\node[] at (12,2) (24) {$e$};
\node[] at (15,2) (34) {$f$};

\node[draw,rectangle] at (1.5,4) (1) {$abc$};
\node[thick, blue, above=0 of 1] {$R$};
\node[draw,rectangle] at (5.5,4) (2) {$ade$};
\node[thick, blue, above=0 of 2] {$S$};
\node[draw,rectangle] at (9.5,4) (3) {$bdf$};
\node[thick, blue, above=0 of 3] {$T$};
\node[draw,rectangle] at (13.5,4) (4) {$cef$};
\node[thick, blue, above=0 of 4] {$U$};

\node[] at (7.5,6) (top) {$\hat 1$};

\path[] (0) edge (12);
\path[] (0) edge (13);
\path[] (0) edge (14);
\path[] (0) edge (23);
\path[] (0) edge (24);
\path[] (0) edge (34);
\path[] (12) edge (1);
\path[] (12) edge (2);
\path[] (13) edge (1);
\path[] (13) edge (3);
\path[] (14) edge (1);
\path[] (14) edge (4);
\path[] (23) edge (2);
\path[] (23) edge (3);
\path[] (24) edge (2);
\path[] (24) edge (4);
\path[] (34) edge (3);
\path[] (34) edge (4);

\path[] (1) edge (top);
\path[] (2) edge (top);
\path[] (3) edge (top);
\path[] (4) edge (top);
\end{tikzpicture}
\caption{A query where the chain bound is not optimal}
\label{fig:bad:for:chain}
\end{figure}
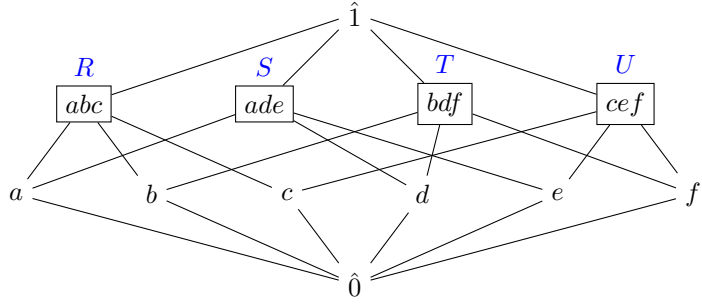

\paragraph*{Choosing a good chain.}

When there are FDs, it is not clear if there even exists a chain with a finite
chain bound.
We show in this section how to select a chain so that the chain hypergraph does
not have an isolated vertex. This means the fractional edge cover number of the
chain hypergraph is finite. The first method to select the chain corresponds to
a generalization of Shearer's lemma.
The second method interestingly corresponds to a {\em dual} version of 
Shearer's lemma.

\bcor[Shearer's Lemma for FDs]
\label{cor:generalizing Shearer}
Consider a query $(\bL, \bR)$.
Let $\bJ$ be the set of join-irreducibles that are below
the order ideal generated by elements in $\bR$.
There exists a set $\{X_1,\dots,X_n\}$ of members of $\bJ$ satisfying the
following.
Let ${\bC}$ be the chain $C_0 = \hat 0 \prec C_1 \prec \cdots
\prec C_n = \hat 1$, where $C_i = \bigvee_{j=1}^i X_j$.
Then the chain ${\bC}$ is good and the hypergraph $H_{\bC}$ has no isolated vertex.
\ecor
\bp
Note that $\bigvee_{Z\in \bJ} Z = \onehat$ because $\bigvee_{j=1}^m R_j = \onehat$.
We construct the sequence $X_1,\dots,X_n$ inductively.
For $i \in [n]$, suppose we have constructed the sequence
$X_1,\dots,X_{i-1}$. Let $Y = \bigvee_{j=1}^{i-1}X_j$.
(If $i=1$ then $Y = \zerohat$.)
Let $X_i$ be an element in $\bJ - \{X_1,\dots,X_{i-1}\}$ such that
the $Y \prec Y \vee X_i$ and $Y\vee X_i$
is {\em minimal} among all such $X_i$.
We stop at $X_n$ when the join is $\onehat$.

For $i \in [n]$, let $C_i = \bigvee_{j=1}^i X_j$. 
We first show that the chain $\bC = (C_i)_{i=1}^n$ is a good chain for $\bR$.
Note that all elements $C_i$ are distinct. 
Consider any $R_j \in \bR$, and suppose $R_j \wedge C_{i-1} \neq R_j \wedge C_i$
but $(R_j \wedge C_i) \vee C_{i-1} \prec C_i$.
Then, $S = (R_j \wedge C_i) \vee C_{i-1}$ is strictly in between $C_{i-1}$ and
$C_i$. Because, if $S = C_{i-1}$ then $R_j \wedge C_{i-1} = R_j \wedge C_i$.
Let $Y = \bigvee_{j=1}^{i-1} X_j$.
We show that there is an element $X \in \bJ - \{X_1,\dots,X_{i-1}\}$ such
that $C_{i-1} \prec Y \vee X \prec C_i$. 
This will violate the choice of $X_i$ and the proof of the
claim would be complete.
Note that $S$ is the join of all join-irreducibles below $C_{i-1}$ and 
below $R_j \wedge C_i$. The join-irreducibles below $R_j \wedge C_i$ are below
$R_j $,
and thus they are in the set $\bJ$. 
Pick $X$ to be a join-irreducible below $R_j \wedge C_i$ but not below $C_{i-1}$ 
and we are done.

To see that every vertex $i\in [n]$ is covered in the chain-cover hypergraph, 
note that if $X_i \preceq R_j $ then $R_j \wedge C_i \neq R_j \wedge C_{i-1}$.
\ep

\begin{figure}[th]
\begin{center}
\begin{tikzpicture}[scale=.8]
\node[] at (0,0) (0) {$\hat 0$};
\node[draw, rectangle] at (-3,1.5) (x) {$x$};
\node[] at (0,1.5) (z) {$z$};
\node[draw, rectangle] at (3,1.5) (y) {$y$};
\node[] at (-2,3) (xz) {$xz$};
\node[] at (2,3) (yz) {$yz$};
\node[] at (0,4.5) (xyz) {$xyz=\onehat$};

\path[] (0) edge (x);
\path[] (0) edge (y);
\path[] (0) edge (z);
\path[] (x) edge (xz);
\path[] (y) edge (yz);
\path[] (z) edge (xz);
\path[] (z) edge (yz);
\path[] (xz) edge (xyz);
\path[] (yz) edge (xyz);
\end{tikzpicture}
\end{center}
\caption{Lattice for $Q \cd R(x),S(y),xy\to z$}
\label{fig:maximal:no:good}
\end{figure}
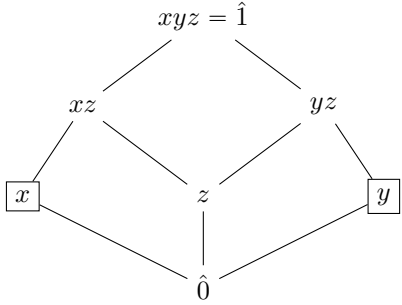

\begin{example}
Recall an example query with UDFs: $Q \cd R(x), S(y), z = f(x,y)$ where
$f$ is a UDF. The lattice is shown in Fig.~\ref{fig:maximal:no:good}.
If we selected any maximal (and thus good) chain, such
as $\hat 0 \prec z \prec xz \prec xyz$, 
or $\hat 0 \prec x \prec xz \prec xyz$,
then $z$ or $xz$ would be an isolated vertex in the chain hypergraph.
Corollary~\ref{cor:generalizing Shearer} tells us to construct a chain by
joining the join irreducibles below $R$ and $S$, which are $x,y$.
Hence, we would select a chain such as $\hat 0 \prec x \prec xyz$ which has no
isolated vertices. 
The algorithm runs in time $O(N^2)$, which is worst-case optimal.
This chain is not maximal.
\end{example}

We can flip the above proof, working from the meet-irreducibles instead of the
join-irreducibles, to obtain the dual version of Shearer's lemma.

\bcor[Dual Shearer's Lemma for FDs]
\label{cor:dual Shearer}
Let $(\bL, \bR)$ represent an input query.
There exists a sequence $X_1,\dots,X_n$ of meet-irreducibles of $\bL$
satisfying the following.
For $i = 0, 1, \dots n$, define $C_i = \bigwedge_{j=1}^{n-i}X_j$. The chain
$\bC = (C_i)_{i=0}^n$ is good for $\bR$ and the chain hypergraph $H_{\bC}$ has no
isolated vertex.
\ecor

\paragraph*{A condition for the chain bound to be tight.}
Now that we know how to select a chain so that the bound is finite, the next 
question is whether the chain bound is tight for some class of queries.

\begin{example} The chain bound is optimal on some non-normal
  lattices.  Consider the query $$R(x),S(y),T(z),xy\rightarrow
  z,xz\rightarrow y, yz\rightarrow x,$$ whose lattice is $M_3$ in
  Fig.~\ref{fig:non:normal} (a non-normal lattice), with
  $|R|=|S|=|T|=N$.  The chain $\hat 0 \prec x \prec xyz= \hat 1$ gives
  a the tight upper bound $N^2$, because its chain hypergraph is
  $e_R = \{x\}, e_S = \{\hat 1\}, e_T =\{\hat 1\}$, has optimal edge
  covers $(w_x,w_y,w_z) = (1,1,0)$ and $=(1,0,1)$.
\end{example}

Before presenting a sufficient condition for when the chain bound is tight,
we need a simple lemma.

\blmm
Let $(\bL, \bR)$ represent a query.
Let $\bC: \zerohat = C_0 \prec C_1 \prec \cdots \prec C_k = \onehat$ be a chain 
that is good for every $X \in L$. 
For every $S \in L$, define 
$e(S) = \{ i \in [k] \suchthat S \wedge C_i \neq S \wedge C_{i-1}\}$.
Then, $X \preceq Y$ implies $e(X) \subseteq e(Y)$.
\label{lmm:monotone}
\elmm
\bp
Consider any $i \notin e(Y)$; then,
$X \wedge C_i \preceq Y \wedge C_i = Y \wedge C_{i-1} \preceq C_{i-1}.$
Thus, $X \wedge C_i = X \wedge C_{i-1}$ which means $i \notin e(X)$.
\ep

\bthm
Let $\zerohat = C_0 \prec C_1 \cdots \prec C_k = \onehat$ be a chain that is
good for every $X\in L$. Also, suppose for every $X, Y \in L$, the following
holds
\begin{equation}
e(X \vee Y) \subseteq e(X) \cup e(Y).
\label{eqn:tightchain}
\end{equation}
Then the chain bound is tight on this lattice.
\ethm
\bp
Let $h^*$ be an optimal (polymatroid) solution to LLP on $L$. We define a new
$\bL$-function $u$ as follows.
\begin{eqnarray*}
   u(\zerohat) &=& 0\\
   u(X) &=& \sum_{i \in e(X)} (h^*(C_i) - h^*(C_{i-1})), \ X \in L.
\end{eqnarray*}

{\bf Claim 1.} $u$ is also an optimal solution to LLP on $\bL$.

We first show that $u$ is a polymatroid. Non-negativity of $u$ follows from
monotonicity of $h^*$. Monotonicity of $u$ follows from
Lemma~\ref{lmm:monotone}. Submodularity follows from the assumption
that $e(X \vee Y) \subseteq e(X) \cup e(Y)$. In fact, from
Lemma~\ref{lmm:monotone} $e(X) \subseteq e(X\vee Y)$ and $e(Y) \subseteq e(X\vee
Y)$, which means $e(X \vee Y) = e(X) \cup e(Y).$ And thus $u$ is modular.

Next, we show that $u(X) \leq h^*(X)$ for all $X \in L$. In particular, $u$ is a
feasible solution to LLP on $L$. This is proved by induction on
$|e(X)|$. The base case when $|e(X)| = 0$ is trivial. For the inductive step,
let $j$ be the maximum number in $e(X)$. Then,
\begin{eqnarray*}
   u(X) &=& \sum_{i \in e(X)} (h^*(C_i) - h^*(C_{i-1}))\\
    &=& \sum_{i \in e(X\wedge C_{j-1})} (h^*(C_i) - h^*(C_{i-1}))
   + h^*(C_j) - h^*(C_{j-1})\\
   &=& u(X \wedge C_{j-1}) + h^*(C_j) - h^*(C_{j-1})\\
   (\text{induction hypothesis})  
   &\leq& h^*(X \wedge C_{j-1}) + h^*(C_j) - h^*(C_{j-1})\\
   (\text{submodularity of $h^*$})
   &\leq& h^*(X)
\end{eqnarray*}

Since $u(\onehat) = h^*(\onehat)$, $u$ is an optimal solution to $\lelp$ on
$L$ as claimed.

{\bf Claim 2.} $u$ is materializable.

Let $g$ be a polymatroid on the Boolean algebra $B_k$, where
\begin{eqnarray*}
   g(i) &=& h^*(C_i) - h^*(C_{i-1})\\
   g(X) &=& \sum_{i \in X} g(i). \\
\end{eqnarray*}
Then $g$ is a modular polymatroid and it can be materialized with a product
instance: $D = \prod_{i=1}^k [2^{g(i)}]$.
We exhibit an embedding from $u$ to $f$. 
The map $f: L \to B_k$ defined by $f(X) = e(X)$ is an embedding because
$X\vee Y \to e(X\vee Y) = e(X) \cup e(Y)$.
One can verify that $u = g \circ f$.
Hence, $u$ is materializable from
Proposition~\ref{prop:entropy-preserving-embedding}.
\ep

\bcor\label{cor:cb tight on distributive lattices}
The chain bound is tight on distributive lattices.
\ecor
\bp
Consider any maximal chain on the distributive lattice $L$.
We only have to verify that
\[ e(X\vee Y) \subseteq e(X) \cup e(Y). \]
We prove this by showing that if $i \notin e(X) \cup e(Y)$ then $i \notin
e(X\vee Y)$. Suppose
\begin{eqnarray*}
   X\wedge C_i &=& X\wedge C_{i-1}\\
   Y\wedge C_i &=& Y\wedge C_{i-1}.
\end{eqnarray*}
Then,
\begin{eqnarray*}
   (X\vee Y) \wedge C_i &=& (X\wedge C_i) \vee (Y \wedge C_i)\\
                        &=& (X\wedge C_{i-1}) \vee (Y \wedge C_{i-1})\\
                        &=& (X\vee Y) \wedge C_{i-1}.
\end{eqnarray*}
\ep

\begin{example}[Tightness on non-distributive lattice]
   The chain bound is tight on the lattice (and the chain) shown in
   Figure~\ref{fig:ef-ex2}. The red sets are the sets $e(X)$, $X\in L$.
   In particular, the characterization condition~\eqref{eqn:tightchain} goes
   beyond distributive lattices.
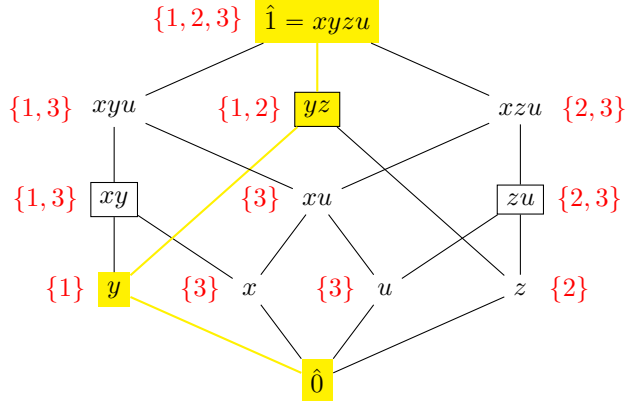
\begin{figure}[t]
\begin{center}
\begin{tikzpicture}[domain=0:20, scale=0.6]
  \node[fill=yellow] at (4.5,0) (0) {$\hat 0$};
  \node[fill=yellow] at (0,2) (y) {$y$};
  \node[] at (3,2) (x) {$x$};
  \node[] at (6,2) (u) {$u$};
  \node[] at (9,2) (z) {$z$};
  \node[draw,rectangle] at (0,4) (xy) {$xy$};
  \node[] at (4.5,4) (xu) {$xu$};
  \node[draw,rectangle] at (9,4) (zu) {$zu$};
  \node[] at (0,6) (xyu) {$xyu$};
  \node[fill=yellow,draw,rectangle] at (4.5,6) (yz) {$yz$};
  \node[] at (9,6) (xzu) {$xzu$};
  \node[fill=yellow] at (4.5,8) (1) {$\hat 1 = xyzu$};

  \node[thick, red, left=0.04 of y] {$\{1\}$};
  \node[thick, red, left=0.04 of x] {$\{3\}$};
  \node[thick, red, left=0.04 of u] {$\{3\}$};
  \node[thick, red, right=0.04 of z] {$\{2\}$};
  \node[thick, red, left=0.04 of xy] {$\{1,3\}$};
  \node[thick, red, left=0.04 of xu] {$\{3\}$};
  \node[thick, red, right=0.04 of zu] {$\{2,3\}$};
  \node[thick, red, left=0.00 of xyu] {$\{1,3\}$};
  \node[thick, red, left=0.04 of yz] {$\{1,2\}$};
  \node[thick, red, right=0.00 of xzu] {$\{2,3\}$};
  \node[thick, red, left=0.04 of 1] {$\{1,2,3\}$};

  \path[] (0) edge (x);
  \path[thick,yellow] (0) edge (y);
  \path[] (0) edge (z);
  \path[] (0) edge (u);
  \path[] (x) edge (xy);
  \path[] (x) edge (xu);
  \path[] (u) edge (xu);
  \path[] (u) edge (zu);
  \path[] (z) edge (yz);
  \path[] (z) edge (zu);
  \path[thick,yellow] (y) edge (yz);
  \path[] (y) edge (xy);
  \path[] (xy) edge (xyu);
  \path[] (xu) edge (xyu);
  \path[] (xu) edge (xzu);
  \path[] (zu) edge (xzu);
  \path[] (xyu) edge (1);
  \path[thick,yellow] (yz) edge (1);
  \path[] (xzu) edge (1);
  \end{tikzpicture}
\end{center}
   \caption{Condition \eqref{eqn:tightchain} holds for this lattice}
   \label{fig:ef-ex2}
\end{figure}
\end{example}

From Proposition~\ref{prop:simple:FD:implies:distributive:lattice} we obtain
the following which subsumes results from \cite{GLVV} in the simple FD case.
\bcor
If all FDs are simple, then the chain bound is tight and
the chain algorithm is worst-case optimal.
\ecor

\begin{example}[Chain bound is not always tight] \label{ex:bad:for:chain} 
For some queries, even with
  normal lattices, no chains give a tight upper bound.  This is
  illustrated by $Q$ in Fig.~\ref{fig:bad:for:chain}.  Consider the
  chain $\hat 0 \prec a \prec abc \prec \hat 1$.  Denoting
  $\set{1,2,3}$ the vertices of the chain hypergraph, its edges are
  $e_R = \{1,2\}$, $e_S = \{1,3\}$, $e_T = e_U = \{2,3\}$.  The
  optimal fractional edge cover is $(w_R, w_S, w_T, w_U) =
  (1/2,1/2,1/2,0)$, $\rho^*=3/2$, therefore the chain bound is $|Q^D|
  \leq N^{3/2}$ and the Chain Algorithm will run in this time.  All
  other maximal chains give the same bound (non-maximal chains are not
  good).  However, we show in the next section that $|Q^D| \leq
  N^{4/3}$, which means that no chain bound is tight.
\end{example}

\paragraph*{Discussion} 
We have shown that the chain bound implies Shearer's lemma and AGM bound from
the Boolean algebra; more generally it is tight on all distributive 
lattices. The lattice corresponding to simple FDs is distributive, thus the chain 
algorithm is worst-case optimal for simple FDs. We also presented examples and 
proved results on how to choose a good
chain so that the chain hypergraph does not have an isolated vertex.

For brevity, we have described the Chain Algorithm using
breadth-first search (or bottom-up). It can also be
adapted to a depth-first implementation (or top-down), as used for
example in LogicBlox' LFTJ~\cite{LFTJ}, which does not
materialize intermediate relations.

%!TEX root = main.tex

\subsection{Sub-Modularity Bound and Sub-Modularity Algorithm (SMA)}
\label{sec:sma}

For some lattices, like Example~\ref{ex:bad:for:chain}, the Chain
Algorithm is sub-optimal no matter what chain we choose.  We describe
here a second proof technique for Shearer's lemma and its adaptation
to lattices, and derive a new algorithm.  We assume w.l.o.g.  that all
coefficients $w_j$ in (\ref{eq:ineq1}) are rational numbers, $w_j =
q_j/d$, and rewrite the set $\set{R_1, \ldots, R_m}$ of lattice 
elements as a multiset $\calB = \set{B_1, B_2, \ldots, B_{\sum_j q_j}}
\subseteq L$ where each lattice element $R_j$ is copied $q_j$ times.  Then, 
inequality (\ref{eq:ineq1}) becomes:
\begin{equation}
  \sum_i h(B_i) \geq d \cdot h(\hat 1) \label{eq:ineq2}
\end{equation}

\paragraph*{The Sub-modularity Proof Sequence.} Balister and
Bollob{\'{a}}s~\cite{DBLP:journals/combinatorica/BalisterB12} give a
simple proof of (\ref{eq:ineq2}) that uses only the sub-modularity
inequality, which we adapt here for arbitrary lattices. We refer to this proof
strategy a {\em sub-modularity proof} or SM-proof.
An SM-proof starts with a mutiset 
$\calB = \set{B_1, B_2, \ldots}$, and applies repeatedly
{\em sub-modularity steps} (SM-steps)\footnote{Called
  ``elementary compression''
  in~\cite{DBLP:journals/combinatorica/BalisterB12}.}. 
An SM-step, $(X,Y)
\rightarrow (X \wedge Y, X \vee Y)$, consists of removing two
incomparable elements $X,Y$ from $\calB$, and replacing them with $X\wedge
Y, X \vee Y$.  A {\em sub-modularity proof sequence}, or simply an
SM-proof, repeatedly applies SM-steps to a multiset $\calB$ until all elements 
in $\calB$ are comparable; at that point $\calB$ is
a chain $\hat 1 \succeq C_1 \succeq C_2 \succeq\ldots$, and we denote $d,
d_1, d_2, \ldots$ the multiplicities of the elements in $\calB$.  The {\em
sub-modularity bound} corresponding to this proof sequence is:
\[
\sum_i h(B_i) \geq d \cdot h(\hat 1) + d_1\cdot h(C_1) + \ldots \geq d \cdot h(\hat 1)
\]

\begin{example}%TODO: move to appendix in PODS version
The following is a valid instance of inequality~\eqref{eq:ineq1}
on the Boolean algebra lattice $2^{\{a,b,c,d\}}$:
   \begin{equation}
      h(abc)+\frac 1 3 h(abd)+\frac 1 3 h(acd)+\frac 1 3 h(bcd) \geq h(\hat 1). 
      \label{eqn:redundant:SM:ex}
   \end{equation}
This corresponds to $d=3$, $\calB=\{ abc,abc,abc,abd, acd, bcd\}$, and the rewritten
inequality
\[ h(abc)+h(abc)+h(abc)+ h(abd)+h(acd)+ h(bcd) \geq 3 \cdot h(\hat 1). \]
We can prove this inequality by applying the following SM-steps:
   \begin{eqnarray*}
      h(abc)+h(abd) &\geq& h(\hat 1)+h(ab)\\
      h(abc)+h(acd) &\geq& h(\hat 1)+h(ac)\\
      h(abc)+h(bcd) &\geq& h(\hat 1)+h(bc)\\
      h(ab)+h(ac) &\geq& h(abc)+h(a)\\
      h(a)+h(bc)&\geq& h(abc) +h(\hat 0).
   \end{eqnarray*}
   Altogether the SM-steps proved that 
   \[ h(abc)+h(abc)+h(abc)+ h(abd)+h(acd)+ h(bcd) \geq 3 \cdot h(\hat 1)+
   2h(abc)+h(\hat 0), \]
   which implies~\eqref{eqn:redundant:SM:ex}. This example also shows that there
   is extra information we could not make use of: the $2h(abc)$ term that is left
   ``dangling''. As we shall see later, in some cases this is a manifestation 
   of the limitation of SM-proofs.
\end{example}

\begin{example} Continuing Example~\ref{ex:bad:for:chain} where the Chain bound
   is not tight, the SM-proof is very simple:
  \begin{eqnarray*}
    h(abc)+h(ade) &\geq& h(\hat 1) + h(a)\\
    h(bdf)+h(cef) &\geq& h(\hat 1) + h(f)\\
    h(a)+h(f) & \geq & h(\hat 1) + h(\hat 0),
  \end{eqnarray*}
  resulting in the SM-bound:
  \[ h(\hat 1) \leq \frac 1 3 h(abc) + \frac 1 3 h(ade) + \frac 1 3
  h(bdf)+\frac 1 3 h(cef).
  \]
  In particular, when all input relations have size $N$, the output size bound 
  is $N^{4/3}$.  This coincides with the co-atomic hypergraph cover,
  hence it is tight.  
  \label{ex:ex:bad:for:chain:good:for:smp}
\end{example}

Obviously, any SM-bound is a correct inequality of the form
(\ref{eq:ineq2}), but the converse does not always hold, as we shall
see.  The converse holds, however, in distributive lattices.
We next show that given
any fractional edge cover $(w_j)_{j=1}^m$ of the
co-atomic hypergraph of a distributive lattice, 
inequality (\ref{eq:ineq1}) is provable through
an SM-proof sequence and, moreover, one can apply the SM-steps in any
order. It follows that every distributive lattice is normal!

For each element
$X \in L$, define $e_X = \setof{Z}{Z \in V_\co, X \not\preceq Z}$;
notice that, if $X$ is an input, $X = R_j$, then $e_{R_j}$ is a
hyperedge of $H_\co$.  A $d$-cover of the co-atoms is a multiset $\calB$
s.t. for each vertex $Z \in V$ there are at least $d$ elements $X \in
B$ s.t. $Z \in e_X$.  Then:

\begin{lmm}
  Suppose $L$ is a distributive lattice, and $B$ is a $d$-cover of the
  co-atoms.  If $B'$ is obtained from $B$ by applying one SM-step
  $(X,Y)\rightarrow (X\wedge Y, X \vee Y)$, then $B'$ is also a
  $d$-cover of the co-atoms.  
\end{lmm}
\begin{proof} 
 We prove this using a similar argument as
 in~\cite{DBLP:journals/combinatorica/BalisterB12}.  
 By Lemma~\ref{lemma:compaction}, for any co-atom $Z$,
  $c(Z,\set{X,Y})= c(Z,\set{X\wedge Y, X \vee Y})$, hence the number
  of elements in $B$ that cover $Z$ is the same as the number of
  elements in $B'$ that cover $Z$.  
\end{proof}

\bcor
Given any fractional edge cover $(w_j)_{j=1}^m$ of the
co-atomic hypergraph, inequality (\ref{eq:ineq1}) is provable through
an SM-proof sequence and, moreover, one can apply the SM-steps in any
order. 
\ecor
\bp
Progress is ensured by the fact that $\sum_i |e_{B_i}|^2$
strictly increases after each sub-modularity step:
\[ |e_X|^2 +
|e_Y|^2 < |e_X \cap e_Y|^2 + |e_X \cup e_Y|^2 = |e_{X\vee Y}|^2 +
|e_{X \wedge Y}|^2.
\]
Finally, when the process ends, the multiset $\calB$
is a chain and $h(\hat 1)$ must occur $d$ times, because any co-atom
$Z$ that is $\succeq C_1$ (the next largest element in the chain) is
covered only by $\hat 1$.  
\ep

\begin{cor}
  Every distributive lattice is normal.
  \label{cor:every:dist:lattice:is:normal}
\end{cor}

\paragraph*{The Submodularity Algorithm.} The SM-Algorithm 
(Algorithm \ref{algo:sma})
starts by using the cardinalities $(N_j)_{j=1}^m$ of the input
relations to obtain an optimal solution $h^*$ of the LLP
(Eq.(\ref{eq:llp})), and an optimal dual solution $s^*, w^*$: the
coefficients $(w^*_j)_{j=1}^m$ form a valid inequality
(\ref{eq:ineq1}) (see the discussion at the beginning of
Sec~\ref{sec:proof:algorithms}), which is tight for $h^*$; as before,
write $s^*_{X,Y}, w_j^*$ as rational numbers, $w_j^* = q_j/d,
s^*_{X,Y}=p_{X,Y}/d$.\footnote{Extreme points of the dual polytope are
data-independent!}
The SM-algorithm requires as an input an SM-proof sequence of this
inequality, then computes the query $Q^D$ as follows.
%in time $\tilde O(N + \prod_j
%N_j^{w_j^*})$, as follows.

The algorithm performs the SM-steps in the proof sequence, maintaining
the multiset $\calB \subseteq L$.\footnote{I.e. every member of the multiset is
a member of the set $L$}
It also maintains a cache of temporary
relations, in one-to-one correspondence with $\calB$: for each $B \in \calB$,
there is one temporary relation $T(B)$, with set of attributes $B$; if
$B$ occurs multiple times in $\calB$, then there are multiple temporary
relations $T(B)$.  Initially the temporary tables are the input
relations, s.t.  each relation $R_j$ is copied $q_j$ times.  Next, the algorithm
applies the SM-steps in the proof sequence, and for each step
$(X,Y)\rightarrow (X\wedge Y, X\vee Y)$ performs a {\em
  sub-modularity join}:

\begin{quote}
{\bf Sub-modularity join.} Let $Z = X\wedge Y$ be the set of common
   variables in $T(X), T(Y)$.  The {\em degree} of some value $v \in
   \prod_{z\in Z}\Dom(Z)$ is the number of tuples in $T(Y)$ with $Z = v$:
  \begin{equation}
     \Deg_{T(Y)}(v) \defeq |\sigma_{Z=v}(T(Y))|
  \label{eqn:degree}
  \end{equation}
   The SM-join partitions $\Pi_Z(T(Y))$ into {\em light hitters} and
  {\em heavy hitters}, consisting of values with degree $\leq
   2^{h^*(Y)-h^*(Z)}$ or $> 2^{h^*(Y)-h^*(Z)}$ respectively; denote
  them $\Light$ and $\Heavy$.  Define $T(X\vee Y)$ the subset of the
   join $T(X) \Join T(Y)$ restricted to light hitters, and define $T(X \wedge Y)$ to be the set of heavy hitters.  The SM-join removes
   $T(X),T(Y)$ from the cache, and adds $T(X \wedge Y), T(X\vee Y)$ to
  the cache of temporary tables.
\end{quote}

After processing the entire proof sequence, the algorithm returns the
union of all $d$ temporary tables $T(\hat 1)$, then semi-join reduces
them with all input relations. 

\begin{algorithm}[th]
  \caption{The Sub-modularity Algorithm}
  \label{algo:sma}
  \begin{algorithmic}[1]
    \Require{A query $Q$, over relations $R_1, \ldots, R_m$}
    \Require{A SM-proof sequence of $\sum_{B \in \calB} h^*(B) \geq d \cdot h^*(\hat 1)$}
    \State Initialize $|\calB|$ temporary tables $T(B)$, $B\in \calB$:
    \State \ \ \ where each $T(B)$ is initially  some $R_j$ \Comment{See text}
    \For {each SM-step $(X,Y) \rightarrow (X \wedge Y,X\vee Y)$}
       \State Let $Z = X \wedge Y$
       \State $\Light(Z) \defeq \setof{v}{\log \Deg_{T(Y)}(v) \leq h^*(Y) - h^*(Z)}$
       \State $\Heavy(Z) \defeq \setof{v}{\log \Deg_{T(Y)}(v) > h^*(Y) - h^*(Z)}$
       %% \State Replace the tables $T(X),T(Y)$ with
       %% we don't replace, we keep everything!
       \State Add the following tables to the cache
       \State $T(X \wedge Y) = \Pi_{Z}(T(X)) \cap \Pi_{Z}(T(Y)) \cap \Heavy(Z)$
       \State $T(X\vee Y) =  (T(X) \Join (T(Y) \ltimes \Light(Z)))^+$
    \EndFor
    \State \Return{$\bigcup_{B \in \calB, B=\hat 1}T(B)$ semi-join reduced with all
    inputs}
  \end{algorithmic}
\end{algorithm}

The following invariant is maintained by the algorithm:
\begin{lmm}
  At each step in the algorithm, for any relation $T(B)$ in cache,
  $\log |T(B)| \leq h^*(B)$.
  \label{lmm:cache-bound}
\end{lmm}
\begin{proof}
   We induct on the number of SM-steps. If $T(B)$ was a copy of an input
   relation $R_j$ that participates in the proof, then 
   $\log |T(B)| = n_j = h^*(B)$, because due to complementary slackness
   $w^*_j>0$ implies the primal constraint has to be tight.
   For every SM-step $(X,Y) \to (X \wedge Y, X\vee Y)$, it must hold that
   $h^*(Y) - h^*(X \wedge Y) = h^*(X\vee Y) - h^*(X)$ because inequality
   \eqref{eq:ineq1} is an {\em equality} for $h^*$: we started from
   an LLP-dual optimal solution.
   Hence,
  \begin{eqnarray*}
     |T(X\vee Y)| & \leq & |T(X)| \cdot 2^{h^*(Y) - h^*(Z)} \\
                   & \leq & 2^{h^*(X)} \cdot 2^{h^*(X\vee Y) - h^*(X)} \\
                   & \leq & 2^{h^*(X\vee Y)} \\
     |T(X \wedge Y)| & \leq & |T(Y)| / 2^{h^*(Y) - h^*(Z)} \leq 2^{h^*(Z)}
  \end{eqnarray*}
\end{proof}

\begin{example} 
  Consider the SM proof sequence from 
  Example~\ref{ex:ex:bad:for:chain:good:for:smp}, we explain how the SM 
  algorithm works for this proof sequence.
  The optimal solution LLP solution is
  \begin{eqnarray*}
     h^*(\hat 1) &=& (4/3)\cdot \log N,\\
     h^*(X) &=& \log N, \ \text{ for } X \in \{abc,ade,bdf,cef\}, \\
     h^*(X) &=& (2/3)\cdot \log N, \ \text{ for } X \in \{a,b,c,d,e,f\}.
  \end{eqnarray*}

  The SM algorithm works as follows.
  \bi \item It first SM-joins $R(abc)$ with $S(ade)$, producing
  relations $Q_1(\hat 1)$ and $\Heavy_1(a)$, where $\Heavy_1(a)$ is the set of 
  values $a$ whose $S$-degree is at least $N^{1-2/3} = N^{1/3}$.
  Thus, $|\Heavy_1(a)| \leq N^{2/3}$.
  Since the light part has degree at  most $N^{1/3}$, 
  $|Q_1(\hat 1)| \leq N^{4/3}$.

  \item Similarly it SM-joins $T(bdf)$ with
  $U(cef)$ producing relations $Q_2(\hat 1)$ and $\Heavy_2(f)$, and
  finally it computes the cross product $\Heavy_1(a) \times
  \Heavy_2(f)$, then expands the result (since $\{a,f\}^+=\hat 1$) to
  obtain a relation $Q_3(\hat 1)$.  
  Since both
  $|\Heavy_1(a)| \leq N^{2/3}$ and
  $|\Heavy_2(f)| \leq N^{2/3}$, their cross-product has size at most
  $N^{4/3}$ and can be computed within that time budget.
  
  \item Finally, it returns $Q_1 \cup Q_2 \cup Q_3$ semi-joined with input relations.
  \ei
  \label{ex:bad:for:chain:good:sma}
\end{example}

%\hqn{TODO: re-work the meaning of ``good''}

There are two reasons why the SM-algorithm may fail: some branches of
heavy or light elements may never join into a $T(\hat 1)$, and the
algorithm may attempt to join some light with heavy values from the
same relation. 
We give a sufficient condition for the SMA algorithm to be
correct.  The criterion consists of iterating over the SM-proof
sequence, and assigning a set $\querylabels(B)$ of labels to each copy $B \in \calB$. (Copies
of the same lattice element receive their own label sets.) Initially,
every $B \in \calB$ receives a single label $1$, namely $\querylabels(B) =
\{1\}$, the same for all $B \in \calB$.  Consider an SM-step $(X,Y) \to
(X \wedge Y, X\vee Y)$.  Let $\calA(X,Y) = \querylabels(X) \cap
\querylabels(Y)$.  Assign $\querylabels(X\vee Y) = \calA(X,Y)$ and,
if $X \wedge Y \neq \zerohat$, assign
to $\querylabels(X \wedge Y)$ a {\em fresh} set of $a$ labels:
$\querylabels(X \wedge Y) = \{f(j) \suchthat j \in \calA(X,Y)\}$ for some
label assignment $f$.
For each $Z \notin \{X,Y\}$, set
$\querylabels(Z) = \querylabels(Z) \cup \setof{f(j)}{\suchthat j \in
  \querylabels(Z) \cap \calA(X,Y)}$.  Note the crucial fact that in this
description we did not remove $X,Y$ from $\calB$. We keep accumulating
elements to $\calB$, unlike in the proof sequence where each step replaces
an old pair with a new pair.  
Note the important fact that that each copy of lattice element in the multiset 
$\calB$ gets its own label set.

\bdefn[Good SM-proof sequence]
Call the SM-proof sequence {\em good} if
$\calA(X,Y) \neq \emptyset$ for all SM-steps $(X,Y)\to(X\vee Y, X\wedge Y)$,
{\em and} if in the end every
label is present in $\bigcup_{\hat 1 \in B} \querylabels(\hat 1)$.
\edefn

\begin{example}
Continue with Example~\ref{ex:bad:for:chain:good:sma},
we check that the SM-proof is good.  Initially all elements in
  $\{abc,ade,bdf,cef\}$ have $\querylabels=\{1\}$. After the first SM-step:
  $\querylabels(\hat 1) = \set{1}$ and $\querylabels(a) = \{2\}$, where $2$ is a
  fresh label; $2$ is also added to $\querylabels$ of
  $abc,ade,bdf,cef$, so they are all equal to $\{1,2\}$.  
  After the second SM-step there are two copies of
  $\onehat$, where $\bigcup_{\onehat \in \calB} \querylabels(\hat 1) =
  \set{1,2}$ and $\querylabels(f) = \{3,4\}$, two fresh labels, which are
  added to $\querylabels(a)$ so that $\querylabels(a)=\{2,3,4\}$.  
  After the third SM-step,
  $\bigcup_{\onehat \in \calB} \querylabels(\hat 1) = \set{1,2,3,4}$; 
  it follows that the proof sequence is good.
\end{example}

We next show that if there is a good proof sequence, then SMA runs in time that 
matches the bound.

\begin{thm}
  If the SM-proof sequence for $\sum_j w_j^* h(R_j) \geq h(\hat 1)$ is
  good, then the SM algorithm correctly computes $Q^D$, and runs in
  time $\tilde O(N+\prod_j N_j^{w^*_j})$, where $\tilde O$ hides a
  $\log N$ factor, a polynomial in query size and SM-proof length.
  \label{thm:sma-main}
\end{thm}
\begin{proof}
   At any point in time, let 
   $L_i = \{ B \in \calB \suchthat i \in \querylabels(B)\}$.
   Let $Q_i \defeq \Join_{B \in L_i} T(B)$.
   We show by induction on the number of SM-steps
   that the algorithm maintains the following invariant:
   $Q \subseteq \bigcup_i Q_i$.
   The base case holds because initially there is only one label, and $Q=Q_1$.
   Consider an SM-step, $(X,Y)\rightarrow (X \wedge Y, X\vee Y)$. 
   Only the subqueries $Q_i$ for which $i \in \calA(X,Y)$ contain both $X$ and $Y$.
   For each such query, consider two new queries: 
   \begin{eqnarray*}
      Q_i^{\Heavy} &=& T(X \wedge Y) \Join \Join_{B \in L_i} T(B) \\
      Q_i^{\Light} &=& T(X\vee Y) \Join \Join_{B \in L_i} T(B).
   \end{eqnarray*}
   Then, obviously $Q_i \subseteq Q_i^{\Heavy} \cup Q_i^{\Light}$ and thus
   the invariant is maintained after the SM-step.
   In the end, each sub-query reaches $\hat 1$ and the semi-join reduction
   filters the result.
(We need this step because large input
relations $R_j$ might have had coefficients $w^*_j$ set to $0$ in the
optimal solution of dual-LLP.)
\end{proof}

\begin{figure}[th]
\begin{center}
\begin{tikzpicture}[scale=.8]
\node[] at (0,0) (0) {$\hat 0$};
\node[] at (-4.5,1.5) (c) {$C$};
\node[] at (0,1.5) (b) {$B$};
\node[draw, rectangle] at (-6,3) (z) {$Z$};
\node[draw, rectangle] at (-3,3) (x) {$X$};
\node[draw, rectangle] at (0,3) (y) {$Y$};
\node[draw, rectangle] at (3,3) (u) {$U$};
\node[] at (-3,4.5) (a) {$A$};
\node[] at (1.5,4.5) (d) {$D$};
\node[] at (-3,6) (1) {$\hat 1$};

\path[] (0) edge (c);
\path[] (0) edge (b);
\path[] (0) edge (u);
\path[] (c) edge (z);
\path[] (c) edge (a);
\path[] (b) edge (x);
\path[] (b) edge (y);
\path[] (b) edge (d);
\path[] (x) edge (a);
\path[] (y) edge (a);
\path[] (u) edge (d);
\path[] (z) edge (1);
\path[] (a) edge (1);
\path[] (d) edge (1);
\end{tikzpicture}
\end{center}
  \caption{A lattice with an SM-proof that is not good}
  \label{fig:non-tree}
\end{figure}
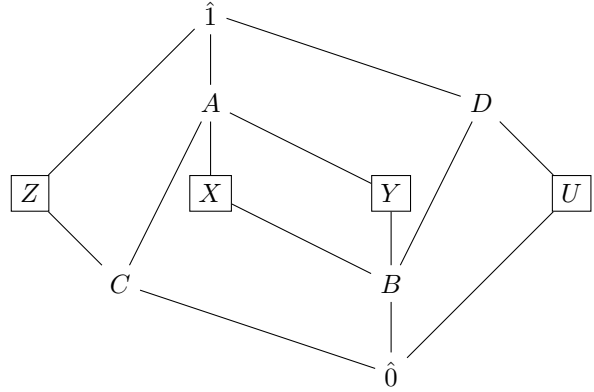

\begin{example}
  Not every SM-proof sequence is good.  Consider the following
  SM-proof sequence of $h(X)+h(Y)+h(Z)+h(U)\geq 2h(\hat 1)+2h(\hat 0)$
  on the lattice in Fig.~\ref{fig:non-tree}:
  \begin{eqnarray*}
     h(X)+h(Y) &\geq & h(A) + h(B) \\
     h(A)+h(Z) &\geq & h(\hat 1) + h(C)\\
     h(B)+h(U) &\geq & h(D) + h(\hat 0)\\
     h(C)+h(D) &\geq & h(\hat 1)+h(\hat 0)
  \end{eqnarray*}
  From the first three SM-steps, 
   $\querylabels(B) = \{2\}$ and $\querylabels(C) = \{3\}$,
  and $\querylabels(D) = \{2\}$. Thus, at the last step 
   $\calA(C,D) = \emptyset$ which is
  not good.
  The inequality admits a different proof sequence, which is good: 
  $(X,Z)\rightarrow
  (C,\hat 1)$, $(Y,U)\rightarrow (\hat 0,D)$, $(C,D)\rightarrow (\hat
  0, \hat 1)$.  It is unknown whether every SM-proof sequence can be
  transformed into a good one.
\end{example}

\begin{example}
   For the query in Fig.~\ref{fig:bad:sm:proof}, an SM-proof is
   \begin{eqnarray*}
      h(X)+h(Y)&\geq&h(C)+h(A)\\
      h(Z)+h(W)&\geq&h(D)+h(B)\\
      h(A)+h(D)&\geq&h(\hat 1)+h(\hat 0)\\
      h(B)+h(C)&\geq&h(\hat 1)+h(\hat 0).
   \end{eqnarray*}
   %c: 1,3
   %a: 2,3
   %xyzw: 1,2
   %d: 1,2
   %b: 3
   After the second SM-step, $\querylabels(C)=\{1,3\}$, 
   $\querylabels(D)=\{1,2\}$,
   $\querylabels(A)=\{2,3\}$, $\querylabels(B)=\{3\}$.
   Hence, labels $2$ and $3$ are pushed to copies of $\querylabels(\hat 1)$,
   but label $1$ is not present in any set $\querylabels(\hat 1)$. So this proof 
   sequence is no good, for a different reason from the previous example.
\end{example}

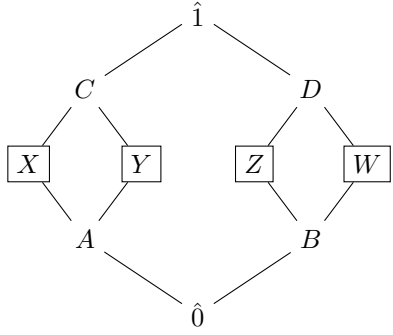
\begin{figure}[t]
   \centering
\begin{tikzpicture}[domain=0:20, scale=0.5]
\node[] at (1.5,2) (a) {$A$};
\node[] at (1.5,6) (c) {$C$};
\node[draw,rectangle] at (0,4) (x) {$X$};
\node[draw,rectangle] at (3,4) (y) {$Y$};
\node[] at (4.5,0) (0) {$\hat 0$};
\node[draw,rectangle] at (6,4) (z) {$Z$};
\node[] at (7.5,2) (b) {$B$};
\node[] at (7.5,6) (d) {$D$};
\node[draw,rectangle] at (9,4) (w) {$W$};
\node[] at (4.5,8) (1) {$\hat 1$};

\path[] (0) edge (a);
\path[] (0) edge (b);
\path[] (a) edge (x);
\path[] (a) edge (y);
\path[] (b) edge (z);
\path[] (b) edge (w);
\path[] (x) edge (c);
\path[] (y) edge (c);
\path[] (z) edge (d);
\path[] (w) edge (d);
\path[] (c) edge (1);
\path[] (d) edge (1);
\end{tikzpicture}
\caption{Another example with a bad SM-proof}
\label{fig:bad:sm:proof}
\end{figure}

Here, we show that some lattices don't have any SM-proofs:
%%%%%%%%%%%%%%%%%%%%%%%%%%%%%%%%%%%%%%%%%%%%%%%%%%%%%%%
%%%%%%%%%%%%%%%%%%%%%%%%%%%%%%%%%%%%%%%%%%%%%%%%%%%%%%%
%
% The following example is used later as the running example in CSMA section.
% Changes in the naming of e.g. lattice nodes in here, will have to be reflected in there.
%
%%%%%%%%%%%%%%%%%%%%%%%%%%%%%%%%%%%%%%%%%%%%%%%%%%%%%%%
%%%%%%%%%%%%%%%%%%%%%%%%%%%%%%%%%%%%%%%%%%%%%%%%%%%%%%%

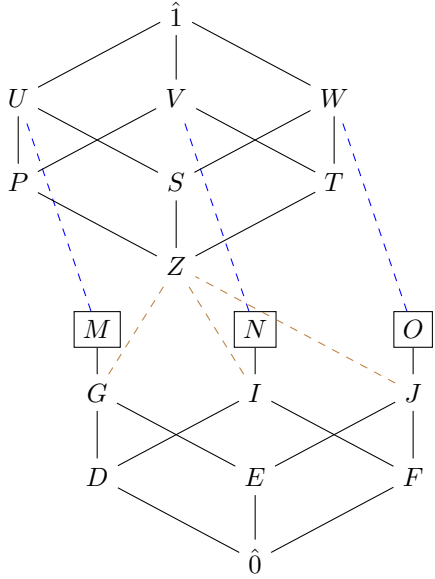
\begin{figure}[th]
\begin{center}
\begin{tikzpicture}[scale=.7, yscale=.8]
\node[] at (0,0) (0) {$\hat 0$};
\node[] at (-3,2) (m) {$\Lm$};
\node[] at (0,2) (n) {$\Ln$};
\node[] at (3,2) (p) {$\Lp$};
\node[] at (-3,4) (mn) {$\Lmn$};
\node[] at (0,4) (mp) {$\Lmp$};
\node[] at (3,4) (np) {$\Lnp$};
\node[draw, rectangle] at (-3,5.5) (x) {$\Lx$};
\node[draw, rectangle]  at (0,5.5) (y) {$\Ly$};
\node[draw, rectangle]  at (3,5.5) (z) {$\Lz$};
\begin{scope}[shift={(-1.5, 7)}]
\node[] at (0,0) (mnp) {$\Lmnp$};
\node[] at (-3,2) (a) {$\La$};
\node[] at (0,2) (b) {$\Lb$};
\node[] at (3,2) (c) {$\Lc$};
\node[] at (-3,4) (ab) {$\Lab$};
\node[] at (0,4) (ac) {$\Lac$};
\node[] at (3,4) (bc) {$\Lbc$};
\node[] at (0, 6) (1) {$\onehat$};
\end{scope}

\path[] (0) edge (m);
\path[] (0) edge (n);
\path[] (0) edge (p);
\path[] (m) edge (mn);
\path[] (m) edge (mp);
\path[] (n) edge (mn);
\path[] (n) edge (np);
\path[] (p) edge (mp);
\path[] (p) edge (np);
\path[] (mn) edge (x);
\path[] (mp) edge (y);
\path[] (np) edge (z);

\path[dashed, brown] (mn) edge (mnp);
\path[dashed, brown] (mp) edge (mnp);
\path[dashed, brown] (np) edge (mnp);

\path[dashed, blue] (x) edge (ab);
\path[dashed, blue] (y) edge (ac);
\path[dashed, blue] (z) edge (bc);

\path[] (mnp) edge (a);
\path[] (mnp) edge (b);
\path[] (mnp) edge (c);
\path[] (a) edge (ab);
\path[] (a) edge (ac);
\path[] (b) edge (ab);
\path[] (b) edge (bc);
\path[] (c) edge (ac);
\path[] (c) edge (bc);
\path[] (ab) edge (1);
\path[] (ac) edge (1);
\path[] (bc) edge (1);
\end{tikzpicture}
\end{center}
  \caption{A lattice with no SM-proof sequence}
  \label{fig:no-smp}
\end{figure}
\begin{example}
\label{ex:no-smp}
  Not every lattice admits an SM-proof.  Consider the lattice in
  Fig.~\ref{fig:no-smp}.  It satisfies the inequality $h(\Lx)+h(\Ly)+h(\Lz)
  \geq 2h(\hat 1)+h(\hat 0)$, but there is no SM-proof sequence that
  derives it.  This is somewhat surprising, because Lemma~\ref{lmm:generalized:inequality}
  showed that every inequality can be proven by adding up sub-modularity 
  inequalities, and indeed our inequality can be obtained by summing up:
  \begin{eqnarray}
     h(\Lx) + h(\Lmnp) &\geq & h(\Lab) + h(\Lmn) \\
     h(\Ly) + h(\Lmnp) &\geq & h(\Lac) + h(\Lmp) \\
     h(\Lz) + h(\Lmnp) &\geq & h(\Lbc) + h(\Lnp) \\
     h(\Lab)+h(\Lac)&\geq& h(\onehat) + h(\La)\label{eq:no-smp-tri-abc-1}\\
     h(\Lbc)+h(\La) &\geq & h(\onehat) + h(\Lmnp) \label{eq:no-smp-tri-abc-2}\\
     h(\Lmn)+h(\Lmp)&\geq&h(\Lmnp)+h(\Lm)\label{eq:no-smp-tri-mnp-1}\\
     h(\Lnp)+h(\Lm) &\geq & h(\Lmnp) + h(\hat 0) \label{eq:no-smp-tri-mnp-2}
  \end{eqnarray}
  (The last two inequalities each consists of 2 SM-steps, identical to
  the proof of $h(xy)+h(xz)+h(yz)\geq 2h(xyz)+h(\hat 0)$ for the
  triangle query, see Example~\ref{ex:running:2}.)  But this is not an
  SM-proof, because in the first step we need $h(\Lmnp)$, which
  is only produced later.  More surprisingly, the lattice is normal.
  %, because its co-atomic hypergraph is $e_x = \set{ac,bc},
  %e_y=\set{ab,bc}, e_Z=\set{ab,ac}$, hence there are essentially only
  %four inequalities to check, corresponding to
  %$(w_x,w_y,w_z)=(1/2,1/2,1/2)$, or $(1,1,0)$, or $(1,0,1)$, or
  %$(0,1,1)$: we have verified the first, and the other three are
  %straightforward.
\end{example}

%!TEX root = main.tex

\subsection{Meeting LLP-Bound with Conditional SM Algorithm -- CSMA}
\label{sec:csma}

Recall that in the chain proof we decompose an input $h(R_j)$ into
a telescoping sum $h(R_j) = \sum_{i\in e_j}(h(R_j\wedge C_i)-h(R_j\wedge
C_{i-1}))$. If $h$ is entropic, then $h(Y)-h(X) =h(Y|X)$ is the conditional
entropy. Hence, the chain proof decomposes an input
$h(R_i)$ into a sum of conditional terms.
Taking cues from this conditional decomposition, and the SM-steps of the
SM-proof, we devise a new type of proof sequence that resolves
the dilemma of the query of Example~\ref{ex:no-smp}, depicted in Fig.~\ref{fig:no-smp}.
In what follows, when $X\preceq Y$ we write $h(Y|X)$ as a short hand for
$h(Y)-h(X)$. We assume $h$ is a polymatroid.
%\mak{Check the following example, and make sure it is consistent with the previous section.
%The reviewer who caught a bug in it before will look again into it.}
\begin{example*}[\continued{ex:no-smp}]
The key issue is that, starting from $h(\Lx)+h(\Ly)+h(\Lz)$ there is no SM-step
that will help prove the desired inequality $h(\Lx)+h(\Ly)+h(\Lz) \geq
2h(\onehat)+h(\zerohat)$. For example, applying $h(\Lx)+h(\Ly)\geq
h(\onehat)+h(\Lm)$ and we are left to show that $h(\Lm)+h(\Lz) \geq
h(\onehat)+h(\zerohat)$, but $\Lm \vee \Lz = \Lbc \prec \onehat$ and so this is
simply impossible.
The trick is to apply an SM-step between $\Lmn$ and $\Lmp$, which are {\em
below} the input relations, as suggested by~\eqref{eq:no-smp-tri-mnp-1}.
To obtain $\Lmn$ and $\Lmp$, we {\em decompose} $h(\Lx)$ into 
$h(\Lx|\Lmn)+h(\Lmn)$ and $h(\Ly)$ into $h(\Ly|\Lmp)+h(\Lmp)$; these
decomposition steps get the proof off the ground.
\end{example*}
Formally, our proof strategy involves the following 
{\em three} basic rules, instead of just the SM-rule as before:
\bi
 \item CD-rule: for $X \prec Y$, $h(Y) \to h(Y|X)+h(X)$.
 \item CC-rule: for $X \prec Y$, $h(Y|X)+h(X) \to h(Y)$.
 \item SM-rule: for $A \incomp B$, $h(A)+h(B | A\wedge B) \to h(A\vee B).$
\ei
(CD stands for {\em conditional decomposition}, CC {\em conditional
composition}, and SM {\em sub-modularity}.)
Together, they are called the {\em CSM rules} (for {\em conditional
submodularity}).
The vision behind CSMA is to show that 
\bi
 \item[{\bf (A)}] every output
inequality~\eqref{eq:ineq1} can be proved using a series of CD, CC, and SM
rules, and 
 \item[{\bf (B)}] each rule can be interpreted combinatorially to become an
    algorithmic step, and together they constitute an algorithm (called CSMA)
    that runs in time $\tilde O(2^{h^*(\onehat)})$
\ei
Interestingly, we fell short of objective (A) yet were able to achieve a {\em
stronger} result than objective (B)!
In particular, with respect to (A)
we conjecture that the three rules above form a complete proof system.
We only managed to prove a weaker version of the conjecture 
(Theorem~\ref{thm:csm-seq}) which is sufficient
for CSMA to work, at the price of a poly-log factor.
On the other hand, with respect to (B) we will work with output inequalities 
that are more general than~\eqref{eq:ineq1}, and with constraints more general
than the fd-constraints.

%\paragraph*{Conditional $\lelp$ ($\cllp$).}
\subsubsection{Conditional $\lelp$ ($\cllp$)}

A key technical tool that helps realize the vision is a conditional version of
$\lelp$, called the {\em conditional $\lelp$} ($\cllp$).
A nice by-product of $\cllp$ is that we will be able to compute join
queries with given degree bounds, which is even more general than join queries
with functional dependencies! 

Let $\calP$ be some set of pairs $(X,Y) \in L^2$ such that $X \prec Y$.
Assume for every pair $(X,Y)\in \calP$, there is a given non-negative number
$n_{Y|X}$ called a {\em $\log$-degree bound}. The $\cllp$ is defined as:
\begin{eqnarray*}
\max \ & h(\hat 1)\\
\text{s.t.} \ & h(Y)-h(X) \leq n_{Y|X} & \forall (X,Y) \in \calP\\
\ & h(A\vee B)+h(A\wedge B)-h(A)-h(B) \leq 0 & \forall A \incomp B\in L\\
\ & h(X) - h(Y) \leq 0 & \forall X \prec Y \in L.
\end{eqnarray*}
(By default, $h(\hat 0)=0$, and $h(X)\geq 0$.)
In other words, $\cllp$ requires $h$ to be a polymatroid, subject to the 
$\log$-degree constraints. 
The cardinality constraints are a special case of the $\log$-degree constraints
$h(Y)-h(\zerohat)\leq n_{Y|\zerohat} = n_Y$.
(Intuitively, cardinality constraints are
degree bounds of the empty tuple; fd-constraints $X \to Y$ imposes the
$Y$-degree bound of $0$ for every $X$-tuple;
and so the degree bounds $h(Y | X) \leq n_{Y|X}$ strictly generalizes both
cardinality constraints and fd constraints.)
The following is obvious:

\bprop
$\lelp$ is exactly the same as $\cllp$ for the special case when
$\calP = \setof{(\hat 0, R_j)}{R_j\in \bR}$.
\eprop

Furthermore, we can easily use the $\log$-degree constraints to encode input relations 
with known maximum degree bounds.

We will also need the dual $\cllp$. Let $c_{Y|X}, s_{A,B}$ and $m_{X,Y}$ denote the
(non-negative) dual variables corresponding to the $\log$-degree,
sub-modularity, and monotonicity constraints, respectively. 
For each $Z \in L-\{\hat 0\}$,
define
\begin{multline*}
   \flow(Z) \defeq 
   \sum_{\substack{X: X\prec Z\\(X,Z)\in \calP}}c_{Z|X}-
   \sum_{\substack{Y: Z\prec Y\\(Z,Y)\in \calP}}c_{Y|Z}+
   \sum_{\substack{A\incomp B\\A\wedge B = Z}}s_{A,B}+
   \sum_{\substack{A\incomp B\\A\vee B = Z}}s_{A,B}-
   \sum_{A: A\incomp Z}s_{A,Z}-
   \sum_{X: X\prec Z}m_{X,Z}+
   \sum_{Y: Z\prec Y}m_{Z,Y}.
\end{multline*}
Then, the dual $\cllp$ is
\begin{eqnarray}
   \min & \sum_{(X,Y)\in \calP}n_{Y|X}c_{Y|X}&\nonumber\\
   \text{s.t.} & \flow(\hat 1) \geq 1 &\label{eqn:dual:cllp}\\
   & \flow(Z) \geq 0 & \forall Z \in L-\{\hat 1, \hat 0\}.\nonumber
\end{eqnarray}

\begin{example*}[\continued{ex:no-smp}]
Consider the lattice in Figure~\ref{fig:no-smp}.
Suppose the input relations are $T({\Lx}), T({\Ly}), T({\Lz})$
(along with FDs imposing the lattice structure),\footnote{We abuse
notation here, when the same notation symbol is used for different input
relations. This is to avoid notation cluttering later in
the description of the algorithm.}
and we are not given any other 
bounds on degrees/cardinalities, then the $\cllp$ has 
$\calP=\{(\zerohat,\Lx), (\zerohat,\Ly), (\zerohat,\Lz)\}$ where 
$n_{\Lx}=n_{\Lx|\zerohat}=\log_2|T(\Lx)|$,
$n_{\Ly}=n_{\Ly|\zerohat}=\log_2|T(\Ly)|$, and
$n_{\Lz}=n_{\Lz|\zerohat}=\log_2|T(\Lz)|$. 
In this case $\cllp$ is just $\lelp$.
If in addition we also knew, for example, an upper bound $d$ on the degree of 
$\Lmn$ in table $T({\Lx})$, then we can extend the $\cllp$ by adding 
$(\Lmn,\Lx)$ to $\calP$ where $n_{\Lx|\Lmn}=\log_2 d$.

As mentioned earlier, this lattice satisfies the inequality 
$2h(\onehat)\leq h(\Lx)+h(\Ly)+h(\Lz)$, which does not admit an SM-proof. 
In the $\cllp$, this inequality corresponds to the constraint 
$h(\onehat) \leq \frac{n_{\Lx}+n_{\Ly}+n_{\Lz}}{2}$, which results from the 
following dual solution.
\[c_{\Lx}=c_{\Ly}=c_{\Lz}=1/2,\]
\begin{equation}s_{\Lx,\Lmnp}=s_{\Ly,\Lmnp}=s_{\Lz,\Lmnp}=1/2,\label{eq:no-smp-sol}\end{equation}
\[s_{\Lab,\Lac}=s_{\La,\Lbc}=s_{\Lmn,\Lmp}=s_{\Lm,\Lnp}=1/2,\]
(where all dual variables that are not specified above are zeros). In this dual 
solution, $\flow(\onehat)=1$ while all other lattice nodes have a 
$\flow$ value of $0$. Hence, this solution is feasible.
\end{example*}

%\paragraph*{CSM proof sequence.}
\subsubsection{CSM proof sequence}

The analog of output inequality~\eqref{eq:ineq1} in the conditional world is
\begin{equation}
   \sum_{(X,Y)\in \calP} c_{Y|X}h(Y|X) \geq h(\hat 1)
   \label{eqn:coi} % conditional output inequality
\end{equation}
Identical to Lemma~\ref{lmm:generalized:inequality}, we can show
that~\eqref{eqn:coi} holds for all polymatroids if there are vectors 
$s$ and $m$ such that $(c,s,m)$ is feasible to the dual-$\cllp$.

To answer question {\bf (A)}, we
prove a ``reachability'' lemma that helps us construct a CSM proof sequence.
Let $(c,s,m)$ be any feasible solution to the dual-$\cllp$ \eqref{eqn:dual:cllp}.  
Let $K \subseteq L$ be a set of lattice elements that contains $\hat 0$.
The {\em conditional closure} of $K$ (with respect to $(c,s,m)$) 
is computed from $K$ by repeatedly applying the following two steps:
(1) CD-step (``Conditional Decomposition''): if $Y \in K$ and $X \prec Y$ then add $X$ to $K$,
(2) CC-step (``Conditional Composition''): if $X \in K$ and $c_{Y|X}>0$ then add $Y$ to $K$.

\blmm
For any dual-feasible solution $(c,s,m)$,
let $K \subseteq L$ be a set that contains $\hat 0$, and $\bar K$ be its
conditional closure. If
$\hat 1 \notin \bar K$, then
there are two lattice elements $A,B \in \bar K$ such that
$A\vee B \notin \bar K$ and $s_{A,B}>0$.
\label{lmm:reachability}
\elmm
\bp
If $\hat 1 \notin \bar K$ then $S \defeq \sum_{Z \notin \bar K} \flow(Z) \geq 1$.
If there is no such pair $(A,B)$, then every dual variable
$c_{Y|X}, m_{X,Y}, s_{A,B}$ contributes a non-positive amount to the sum $S$,
which is a contradiction.
\ep

\begin{example*}[\continued{ex:no-smp}]
Consider the dual solution in \eqref{eq:no-smp-sol}. Let's take $K=\{\zerohat\}$
as an example.
Because $c_{\Lx}, c_{\Ly}, c_{\Lz} >0$ (where $c_{\Lx}$ is just an alias 
for $c_{\Lx|\zerohat}$ and so on), the conditional closure $\bar K$ contains 
$\Lx, \Ly, \Lz$ due to the CC-steps. After applying CD-steps, $\bar K$ includes 
all lattice elements below $\Lx, \Ly, \Lz$ as well:
\[\bar K=\{\zerohat, \Lm, \Ln, \Lp, \Lmn, \Lmp, \Lnp, \Lx, \Ly, \Lz\}.\]
(At this point $\bar K$ is closed.)
Applying Lemma~\ref{lmm:reachability}, we can find $\Lmn, \Lmp\in\bar K$ with $s_{\Lmn,\Lmp}>0$.
If we add $\Lmnp=\Lmn \vee \Lmp$ to $K$ (and $\bar K$), we can apply Lemma~\ref{lmm:reachability} again and find $\Lx, \Lmnp\in\bar K$ with $s_{\Lx,\Lmnp}>0$. After adding $\Lab=\Lx\vee \Lmnp$ to $K$, we can find $\Ly, \Lmnp\in \bar K$ with $s_{\Ly,\Lmnp>0}$. After adding $\Lac= \Ly\vee \Lmnp$, we will find $\Lab, \Lac\in \bar K$ and add $\onehat=\Lab \vee \Lac$.
\end{example*}

The above lemma allows us to constructively
prove a weaker inequality than \eqref{eqn:coi}, which is
our answer to question {\bf (A)} above. 

\bthm
Let $(c,s,m)$ be an arbitrary feasible solution to the dual-$\cllp$, where
$c_{Y|X} = q_{Y|X}/d$ and $s_{A,B} = q'_{A,B}/d$ are rational numbers.
Let $\calB$ be a multiset of variables $h(Y|X)$, $(X,Y) \in \calP$, where each
variable $h(Y|X)$ occurs $4^{|L|} \cdot q_{Y|X}$ times.
Then, there is a sequence of CD-, CC-, and SM-rules that transforms $\calB$ into 
another multiset $\bar \calB$ which contains the variable $h(\hat 1)$.
Moreover, in this sequence all occurrences of identical rules are consecutive 
in the sequence (hence, they can be combined into a single ``rule with a 
multiplicity'').
\label{thm:csm-seq}
\ethm
\bp
We prove the theorem with an algorithm.
We start from the set $K = \{ \zerohat \}$, and keep adding
elements to it using conditional closure (CC- and CD-) steps until $K$ is 
conditionally closed; then we add a new element $A\vee B$ with an SM-step 
using the pair $(A,B)$ found by Lemma~\ref{lmm:reachability}. 
This process is repeated until $\hat 1\in K$. 
With regard to $\calB$, initially we will pretend that for each variable 
$h(Y|X) \in \calB$ we have only $q_{Y|X}$ copies of it instead of $4^{|L|}q_{Y|X}$ 
copies. While adding elements to $K$, we will also add variables to $\calB$
so eventually $\calB$ contains $\leq 4^{|L|}q_{Y|X}$ of each variable.

We maintain the following invariants throughout the execution of the algorithm:
\bi
\item For every $X\in K$, there is at least one copy of $h(X)$ in $\calB$. 
   To maintain this invariant, whenever we add a new element $X$ to $K$ while 
   $\calB$ does not contain any copy of $h(X)$, we apply some conditional
   closure rules on $\calB$ to produce $h(X)$.
\item For every copy of $h(Y|X)$ currently in $\calB$, there will always remain at 
   least one copy of $h(Y|X)$ after each step of the algorithm
   (i.e.\ in all subsequent multisets $\calB$). To maintain this invariant, 
   before applying any rule (where we will be losing one copy of each term on 
   the left-hand side of that rule in order to gain one copy of each term on the 
   right-hand side), we duplicate the multiplicities of all terms in the current 
   $\calB$, all previous multisets $\calB$, and all rules that have been applied 
   previously. (In effect, we re-run the entire history of rule application
   once.)
\ei
The above invariants are initially satisfied. Now we take the conditional 
closure of $K$. Whenever $X$ is added to $K$ due to some $Y\in K$ that 
satisfies $X\prec Y$ (signalling a CD-step), we first check whether $\calB$ 
contains a copy of $h(X)$. If it does, then no further action is needed. If 
not, we duplicate multiplicities in the current and all previous $\calB$ and 
all previous rules, and then we apply a CD-rule $h(Y)\rightarrow h(X)+h(Y|X)$.
Whenever $Y$ is added to $K$ due to $c_{Y|X}>0$ for some $X\in K$, we first 
check whether $h(Y)$ is in $\calB$ already. If not, we duplicate all 
multiplicities as before, and then apply a CC-rule $h(X)+h(Y|X)\rightarrow h(Y)$.

Now, suppose $K$ is already conditionally closed, we add $A\vee B$ using an 
SM-step guaranteed by Lemma~\ref{lmm:reachability}. We check whether 
$h(A\vee B)$ is in $\calB$. If it is not, we check whether $h(B|A\wedge B)$ is 
in $\calB$. If it is not, we duplicate all multiplicities and apply a CD-rule 
$h(B)\rightarrow h(A\wedge B)+h(B|A\wedge B)$. Now, duplicate all multiplicities 
again, and apply an SM-rule $h(A)+h(B|A\wedge B)\rightarrow h(A\vee B)$.

For each element that is added to $K$, we have to duplicate its multiplicity at 
most twice, and there are at most $|L|$ such elements.

Finally, we show that the same rule cannot be applied multiple times. In the 
above, before we applied any CC-rule that produced $h(Y)$, we checked 
whether $h(Y)$ was already in $\calB$. Only if it was not, we applied the rule 
adding $h(Y)$ to $\calB$, and letting the second invariant preserve it in 
$\calB$. The same holds for SM-rules. In the above arguments, we applied 
CD-rules of the form $h(Y)\rightarrow h(X)+h(Y|X)$ in two different places: 
In the first, we checked that $h(X)$ was not in $\calB$ before we added both 
$h(X)$ and $h(Y|X)$ to $\calB$. In the second, we checked that 
$h(B|A\wedge B)$ was not in $\calB$ before we added both $h(B|A\wedge B)$ and 
$h(A\wedge B)$ to $\calB$.
\ep

The series of CC-, CD-, SM-rules with multiplicities is called
a {\em CSM proof sequence}.
Interpreted integrally, we think of the proof sequence as
having $D \leq 4^{|L|}d$ copies of $h(\hat 1)$ that it tries to reach, but at
$\bar \calB$ it reaches at least one copy and we stop.

\begin{example*}[\continued{ex:no-smp}]
Now, we simulate the proof of Theorem~\ref{thm:csm-seq} on the dual solution in \eqref{eq:no-smp-sol} (which corresponds to the inequality $h(\onehat) \leq \frac{h(\Lx)+h(\Ly)+h(\Lz)}{2}$).
We will start off with four copies of the right-hand side (i.e.\ $2h(\Lx)+2h(\Ly)+2h(\Lz)$), and generate a CSM proof sequence that will produce one copy of $h(\onehat)$ (out of 4 copies).
\begin{eqnarray}
2h(\Lx)\rightarrow 2h(\Lx|\Lmn) + 2h(\Lmn) & & \text{(adding $\Lmn$ to $\bar K$ based on $\Lx \in \bar K$)} \label{eq:csm-step1}\\
2h(\Ly)\rightarrow 2h(\Ly|\Lmp) + 2h(\Lmp) & & \text{(adding $\Lmp$ to $\bar K$ based on $\Ly \in \bar K$)} \label{eq:csm-step2}\\
2h(\Lmp)\rightarrow 2h(\Lmp|\Lm) + 2h(\Lm) & & \text{(extracting $h(\Lmp|\Lm)$ for the next SM-rule)} \label{eq:csm-step3}\\
2h(\Lmn) + 2h(\Lmp|\Lm) \rightarrow 2h(\Lmnp) & & \text{(SM-rule based on $s_{\Lmn, \Lmp}>0$)} \label{eq:csm-step4}\\
h(\Lmnp)+h(\Lx|\Lmn)\rightarrow h(\Lab)& & \text{(SM-rule based on $s_{\Lx, \Lmnp}>0$)} \label{eq:csm-step5}\\
h(\Lmnp)+h(\Ly|\Lmp)\rightarrow h(\Lac)& & \text{(SM-rule based on $s_{\Ly, \Lmnp}>0$)} \label{eq:csm-step6}\\
h(\Lac) \rightarrow h(\Lac|\La)+h(\La) & & \text{(extracting $h(\Lac|\La)$ for the next SM-rule)} \label{eq:csm-step7}\\
h(\Lab)+h(\Lac|\La) \rightarrow h(\onehat) & &\text{(SM-rule based on $s_{\Lab, \Lac}>0$)} \label{eq:csm-step8}
\end{eqnarray}
Note that rules \eqref{eq:csm-step1}\ldots\eqref{eq:csm-step4} above had multiplicities of $2$, because we needed to produce two copies of $h(\Lmnp)$: one for \eqref{eq:csm-step5} and another for \eqref{eq:csm-step6}.
\end{example*}

%\paragraph*{The CSM algorithm (CSMA).}
\subsubsection{The CSM algorithm (CSMA)}

CSMA is our answer to question {\bf (B)}. The algorithm takes as input a join
query with functional dependencies {\em and} maximum degree bounds (if any) from
input relations. This input is represented by the set $\calP^{(0)}$, corresponding
$\log$-degree bounds, and the linear program $\cllp^{(0)}$.
For example, if we use CSMA for the original join query with functional 
dependencies (with no other max-degree bounds), then we would be
starting with $\calP^{(0)} = \{(\hat 0, R_j) \suchthat R_j \in \bR\}$; in
this case $\cllp^{(0)}$ is equivalent to $\lelp$. 

Let $h^{(0)}$ and $(c^{(0)},s^{(0)},m^{(0)})$ be a pair of primal and dual optimal solutions to
$\cllp^{(0)}$, and $\opt$ be its optimal objective value.
CSMA takes the CSM-proof sequence for $\cllp^{(0)}$ as symbolic {\em
instructions}.
For each instruction, CSMA does some computation, spawning a number
of sub-problems, creating new intermediate tables for the sub-problems if 
needed. The final output is contained in the union of outputs of the sub-problems.

For each sub-problem, CSMA constructs a new pair-set $\calP'$, and a new 
linear program $\cllp'$ with dual-feasible solution $(c',s',m')$.
Note that to construct $\cllp'$, there has to be a corresponding $\log$-degree
bound $n'_{Y|X}$ for each pair $(X,Y) \in \calP'$.
The following two invariants are maintained:
\bi
 \item[(Inv1)] For any $(X,Y) \in \calP'$, there is a table $T(Y)$ (an input relation
    of the sub-problem) that ``guards'' the
    constraint $h(Y|X) \leq n'_{Y|X}$ of $\cllp'$ in the sense that 
    \[ \max_{v \in \Pi_X(T)}\log_2\deg_T(v) \leq n'_{Y|X}. \]
    (Note that if $X =\hat 0$ then the above says $\log_2|T| \leq
    n'_{Y|\zerohat} = n'_Y$, which is a cardinality constraint.)
 \item[(Inv2)] $(c',s',m')$ is feasible to dual-$\cllp'$ 
    with objective value satisfying
    $\obj' \defeq \sum_{(X,Y)\in \calP} n'_{Y|X}c'_{Y|X} \leq \opt$.
\ei
It can be verified that the two invariants are satisfied at $\cllp^{(0)}$.
Next we describe how CSMA deals with each instruction from the CSM sequence.

\bi
%\noindent
\item[(1)] {\bf CD-rule} $h(Y) \to h(Y|X)+h(X)$ with multiplicity $t$.

   Recall that $h(Y)$ is just a short hand for $h(Y|\zerohat)$.
   By (Inv1) there is a table $T(Y)$ with $\log_2 |T| \leq n_Y$.
   Define 
\begin{eqnarray*}
   n'_{Y|X} &\defeq& \max_{v\in\Pi_X(T)}\log_2 \deg_T(v)\\
   n'_{X}   &\defeq& \log_2 |\Pi_X(T)|
\end{eqnarray*}
Lemma~\ref{lmm:partition} shows that $T$ can be partitioned into at most 
$\ell = 2\log N$ sub-tables $T^{(1)},\dots,T^{(\ell)}$
such that $n^{(j)}_{Y|X}+n^{(j)}_X \leq n_Y$, for all $j \in [\ell]$, where
%\[
%   n^{(j)}_{Y|X} \defeq \max_{v\in\Pi_X(T^{(j)})}\log_2 \deg_{T^{(j)}}(v), \ 
%   n^{(j)}_{X} \defeq \log_2 |\Pi_X(T^{(j)})|.
%\]
\begin{eqnarray*}
   n^{(j)}_{Y|X} &\defeq& \max_{v\in\Pi_X(T^{(j)})}\log_2 \deg_{T^{(j)}}(v), \\ 
   n^{(j)}_{X} &\defeq& \log_2 |\Pi_X(T^{(j)})|.
\end{eqnarray*}

For each of these sub-tables $T^{(j)}$ of $T$, we create a sub-problem with $T$ 
replaced by $T^{(j)}$.
For the $j$th sub-problem, we add $(\hat 0, X)$ and $(X,Y)$ to $\calP$
with corresponding $\log$-degree constraints 
$h(X) \leq n^{(j)}_X$ and $h(Y|X) \leq n^{(j)}_{Y|X}$, respectively.
We compute the projection of $T^{(j)}$ onto $X$ so we have guards for the two new
constraints.
Set 
%$c'_{Y|X} = c_{Y|X}+t/D$, $c'_X = c_X+t/D$, $c'_Y = c_Y-t/D \geq 0$. 
\begin{eqnarray*}
   c'_{Y|X} &=& c_{Y|X}+t/D, \\
   c'_X &=& c_X+t/D, \\
   c'_Y &=& c_Y-t/D \geq 0. 
\end{eqnarray*}
If $c'_Y = 0$, then we remove $(\zerohat,Y)$ from $\calP$.
By examining $\flow(Z)$ at each node, we can verify that $(c',s,m)$ is a feasible
solution to the new $\cllp'$ with a {\em reduction} in objective value of 
$(n_Y-n'_{Y|X}-n'_X)t/D$.

%\noindent
\item[(2)] {\bf CC-rule} $h(Y|X)+h(X)\to h(Y)$ with multiplicity $t$.

Let $R$ be the guard for $h(Y|X)\leq n_{Y|X}$ and $S$ for $h(X) \leq n_X$.
Let $\theta$ be a threshold to be determined.
If $n_{Y|X}+n_X \leq \opt + \theta$, then we can compute the table
$T(Y) \defeq S(X) \Join R(Y)$ by going over all tuples in $S$ and expanding them
using matching tuples in $R$. The runtime is $\tilde O(2^{n_X+n_{Y|X}})
= \tilde O(2^{\opt + \theta})$.
The dual solution is modified by setting 
%$c'_{Y|X} = c_{Y|X}-t/D$,
%$c'_X=c_X-t/D$, $c'_Y=c_Y+t/D$, $n'_{Y} = n_{Y|X}+n_X$.
\begin{eqnarray*}
   c'_{Y|X} &=& c_{Y|X}-t/D,\\
   c'_X&=&c_X-t/D, \\
   c'_Y&=&c_Y+t/D, \\
   n'_{Y} &=& n_{Y|X}+n_X.
\end{eqnarray*}

If $n_{Y|X}+n_X > \opt + \theta$, we will start afresh from an optimal solution
to the {\em current} $\cllp$. 
Lemma~\ref{lmm:reduced-obj} below shows that the current 
$\cllp$ has an optimal objective value at most $\opt - \theta/(D-1)$.

%\noindent
\item[(3)] {\bf SM-rule} $h(A)+h(B|A\wedge B) \to h(A\vee B)$ with 
multiplicity $t$.

This is similar to case (2), and becomes identical to case (2) when $A=A\wedge B
= X$ and $B = Y$.
Let $R$ be the guard for $h(A)\leq n_{A}$ and $S$ for $h(B|A\wedge B) \leq
n_{B|A\wedge B}$.
If $n_A+n_{B|A\wedge B} \leq \opt+\theta$, then we can compute 
$T(A\vee B) = R \Join S$ in time $\tilde O(2^{\opt+\epsilon})$.
If $n_A+n_{A|A\wedge B} > \opt+\theta$, then we start afresh from a new optimal
solution to the current $\cllp$.
The minor difference to case (2) is that we have to modify the variable
$s'_{A,B} = s_{A,B}-t/D$.
By selecting the correct threshold $\theta$, we can prove that
CSMA runs in time $O((\log N)^e2^{\opt})$, where $e$ is a
data-independent constant (Theorem~\ref{thm:main-csma}).
\ei

\begin{example*}[\continued{ex:no-smp}]
Given a conjunctive query $Q$ whose functional dependencies correspond to the
lattice in Figure~\ref{fig:no-smp} and whose input relations are $T({\Lx}),
T({\Ly}), T({\Lz})$, where $|T({\Lx})|=|T({\Ly})|=|T({\Lz})|= N$.
In the $\cllp$, we have $\calP=\{(\zerohat, \Lx), (\zerohat, \Ly), (\zerohat, \Lz)\}$ where $n_{\Lx}=n_{\Ly}=n_{\Lz}=\log_2 N=:n$.
The optimal objective value is $\opt=\frac{3n}{2}$ (which implies that $|Q|\leq N^{\frac{3}{2}}$), 
and the feasible dual solution given by \eqref{eq:no-smp-sol} is optimal.
Consider the CSM sequence \eqref{eq:csm-step1}\ldots\eqref{eq:csm-step8} that was constructed earlier for \eqref{eq:no-smp-sol}.
We will explain how to run CSMA on this sequence in order to answer $Q$ in time within a polylogarithmic factor of $2^{\opt}=N^{\frac{3}{2}}$.

\bi
 \item The first rule~\eqref{eq:csm-step1} in the sequence is a CD-rule:
    $h(\Lx)\rightarrow h(\Lx|\Lmn) + h(\Lmn)$. The corresponding algorithmic
    step would be to project $T({\Lx})$ on $\Lmn$ while making sure that the
    projection size times the maximum degree of the projection in $T({\Lx})$
    does not exceed $|T({\Lx})|$
(i.e.\ while making sure that the sum of $n'_{\Lmn}   \defeq \log_2
|\Pi_{\Lmn}(T({\Lx}))|$
and 
$n'_{\Lx|\Lmn} \defeq \max_{v\in\Pi_{\Lmn}(T({\Lx}))}\log_2 \deg_{T({\Lx})}(v)$ 
does not exceed $n_{\Lx}$). If all tuples $v$ in the projection 
$\Pi_{\Lmn} (T({\Lx}))$ have the same degree $\deg_{T({\Lx})}(v)$ (i.e., if 
$T({\Lx})$ is ``uniform'' with respect to $\Lmn$), then the required condition 
is met.  Otherwise, let's assume for simplicity that all degrees 
$\deg_{T({\Lx})}(v)$ are powers of $2$ (and $|T({\Lx})|$ is also a power of 2). 
If this is the case, then based on $\deg_{T({\Lx})}(v)$ we can partition
$T({\Lx})$ into a logarithmic number (namely $n_{\Lx}+1$) of parts 
$T^{(0)}({\Lx}), T^{(1)}({\Lx}), \ldots$ such that the required condition is 
met in each one of them.  (In particular, let $T^{(i)}({\Lx})$ satisfy 
$n_{\Lx|\Lmn}^{(i)}\leq i$ and $n_{\Lmn}^{(i)}\leq n_{\Lx}-i$.)
Now, the execution of CSMA will split into a logarithmic number of branches,
each of which will continue on a different part of $T({\Lx})$. 
For some arbitrarily-fixed $i$, let's track the execution of the $i$-th 
branch (i.e.\ the one on $T^{(i)}({\Lx})$).

\item The second rule~\eqref{eq:csm-step2} is another CD-rule: 
   $h(\Ly)\rightarrow h(\Ly|\Lmp) + h(\Lmp)$.  Similar to above, it will 
   result in the partitioning of $T({\Ly})$ into $O(\log N)$ parts $T^{(0)}({\Ly}), 
   T^{(1)}({\Ly}), \ldots$ (such that $T^{(j)}({\Ly})$ satisfies
$n_{\Ly|\Lmp}^{(j)}\leq j$ and $n_{\Lmp}^{(j)}\leq n_{\Ly}-j$.)
The current $i$-th branch on $T^{(i)}({\Lx})$ will now branch further into 
$O(\log N)$ branches corresponding to $T^{(0)}({\Ly}), T^{(1)}({\Ly}), \ldots$
Let's keep track of the $j$-th branch (on $T^{(j)}({\Ly})$) for some arbitrary 
$j$.

\item The third rule~\eqref{eq:csm-step3} is yet another CD-rule:
   $h(\Lmp)\rightarrow h(\Lmp|\Lm) + h(\Lm)$. It will result in the partitioning
   of $T({\Lmp})\defeq \Pi_{\Lmp} (T^{(j)}({\Ly}))$ (where $\Pi_{\Lmp} 
   (T^{(j)}({\Ly}))$ resulted from the previous CD-step) into $O(\log N)$ parts: 
   $T^{(0)}(\Lmp), T^{(1)}(\Lmp), \ldots$ (such that $T^{(k)}(\Lmp)$ satisfies 
   $n_{\Lmp|\Lm}^{(k)}\leq k$  and $n_{\Lm}^{(k)}\leq n_{\Lmp}-k$.)
   Let's track the $k$-th branch (on $T^{(k)}(\Lmp)$).

\item The fourth rule~\eqref{eq:csm-step4} is an SM-rule: $h(\Lmn) + h(\Lmp|\Lm)
   \rightarrow h(\Lmnp)$. The corresponding algorithmic step is to join the
   table $T(\Lmn)$ (or more precisely $\Pi_{\Lmn}(T^{(i)}(\Lx))$ from Step 1) 
   with the table $T(\Lmp)$ (or more precisely $T(\Lmp)$ from Step 3) in order 
   to get a relation $T(\Lmnp)$.
The time required to compute this join is bounded by the size of $T(\Lmn)$
times the maximum degree of $\Lm$ in $T(\Lmp)$. Depending on the current
execution branch $(i, j, k)$, this time might or might not exceed our budget of
$2^{\opt}=N^{3/2}$. For example, in the branch where both $i$ and $j$ are
maximal, we will have $|T(\Lmn)|\leq 1$ and $|T(\Lmp)|\leq 1$, hence the join
takes $O(1)$ time. On the other hand, in the branch where $i=j=0$ and $k$ is
maximal, each one of  $|T(\Lmn)|$, $|T(\Lmp)|$, and the maximum degree of $\Lm$
in $T(\Lmp)$ could be as large as $N$, hence the join could take time $O(N^2)$.
Luckily, Lemma~\ref{lmm:reduced-obj} below implies that ``when one door closes,
another one opens'': In those particular branches where the join cannot be
computed within our time budget, there are extra constraints that if considered
in the $\cllp$ they would reduce the optimal objective value $\opt$. For
example, in the branch where $i=j=0$ and $k$ is maximal, the value of $\Lm$ is
already fixed. This is because the degree of $\Lm$ in table $T(\Lmp)$ is
maximal, hence it is equal to $|T(\Lmp)|$.

\item The remaining rules \eqref{eq:csm-step5}\ldots \eqref{eq:csm-step8} are similar to the previous ones.
\ei
\end{example*}

%\paragraph*{Analysis of CSMA.}
\subsubsection{Analysis of CSMA}

\blmm
Let $X\prec Y$ be in the lattice $L$.
Let $T(Y)$ be a table with $\log_2 |T| \leq n_Y$.
Then, $T$ can be partitioned into at most 
$\ell = 2\log N$ sub-tables $T^{(1)},\dots,T^{(\ell)}$
such that $n^{(j)}_{Y|X}+n^{(j)}_X \leq n_Y$, for all $j \in [\ell]$, where
\begin{eqnarray*}
   n^{(j)}_{Y|X} &\defeq& \max_{v\in\Pi_X(T^{(j)})}\log_2 \deg_{T^{(j)}}(v)\\
   n^{(j)}_{X} &\defeq& \log_2 |\Pi_X(T^{(j)})|.
\end{eqnarray*}
\label{lmm:partition}
\elmm
\bp
To obtain the copies $T^{(j)}$, observe that the number of elements $v \in
\Pi_X(T)$ with $\log$-degree in the interval
$[j,j+1)$ is at most $|T|/2^j \leq 2^{n_Y-j}$.
Hence, if we partition $T$ based on which of the buckets $[j,j+1)$ the
$\log$-degree falls into, we would almost have  the required inequality:
$n^{(j)}_{Y|X}+n^{(j)}_X \leq (j+1)+(n_Y-j) = n_Y+1$.
To resolve the situation, we partition each $T^{(j)}$ into two equal-sized
tables. Overall, we need $\ell = 2\log N$.
\ep

\blmm
Given a $\cllp$ whose optimal objective value is $\opt$, and a feasible dual
solution $(c, s, m)$ whose objective value is  $\obj=\sum_{(\bar X, \bar Y)\in
\calP}c_{\bar Y|\bar X} n_{\bar Y| \bar X}$,
let $c_Y>0$ for some $(\zerohat, Y)\in \calP$.
\bi
\item If $n_Y>\obj$, then $\opt < \obj$.
\item Given $\theta\geq 0$, $0<\epsilon<1$, $\epsilon\leq c_Y$, if $n_Y>\obj+\theta$, 
then $\opt < \obj - \frac{\epsilon\theta}{1-\epsilon}$.
\ei
\label{lmm:reduced-obj}
\elmm
\bp
Let $(c', s, m')$ be a dual solution obtained by setting $c'_Y=c_Y-\epsilon$, $m'_{Y,\onehat}=m_{Y,\onehat}+\epsilon$ and keeping other $c'$ and $m'$ values identical to their $c$ and $m$ counterparts. $(c',s, m')$ is not necessarily feasible. However, it satisfies $\flow(\onehat)\geq 1-\epsilon$ and $\flow(Z)\geq 0$ for all $Z\in L-\{\onehat,\zerohat\}$. Let $(c'', s'', m'')$ be a dual solution obtained by multiplying $(c',s,m')$ by $\frac{1}{1-\epsilon}$. Now $(c'', s'', m'')$ is indeed feasible. Let $\obj''$ be its objective value. Because it is dual feasible, $\obj''\geq \opt$.

\begin{eqnarray*}
   \obj'' &=& \sum_{(\bar X, \bar Y)\in \calP}c''_{\bar Y|\bar X} n_{\bar Y| \bar X}\\
   &=& \sum_{(\bar X,\bar Y)\in \calP} \frac{c'_{\bar Y|\bar X}}{1-\epsilon}n_{\bar Y|\bar X}\\
        &=& \frac{1}{1-\epsilon}\sum_{(\bar X,\bar Y)\in \calP} c_{\bar Y|\bar X}n_{\bar Y|\bar X}-\frac{\epsilon}{1-\epsilon} n_Y\\
        &< & \frac{1}{1-\epsilon}\obj-\frac{\epsilon}{1-\epsilon}(\obj+\theta)\\
        &= & \obj - \frac{\epsilon\theta}{1-\epsilon}.
\end{eqnarray*}
\ep

\bthm
CSMA runs in time $O(N+(\log N)^e2^{\opt})$, where $\opt$ is the optimal 
objective value for $\cllp^{(0)}$, $N$ is the input size, and $e$ is a
data-independent constant.
\label{thm:main-csma}
\ethm
\bp
Since the number of CSM rules is at most $|L|^2$ (these are rules with
multiplicities), so no rule is repeated.
The worst case is obtained when the algorithm branches as far as possible only
to have to reduce the objective by $\epsilon = \frac{\theta}{D-1}$ at the very
end and all leaf nodes of the execution tree have to be started afresh with the
new optimal value. Let $x = 2^\opt$ and $T(x)$ denote the runtime of the 
algorithm on the $\cllp^{(0)}$ with objective value $\opt$. Then, the 
recurrence for the runtime is
\begin{eqnarray*}
T(x) &=& (\ell)^{|L|^2}T\left(2^{\opt - \frac{\theta}{D-1}}\right)
   + \ell^{|L|^2} 2^{\theta}x\\
   &=& (\ell)^{|L|^2}T\left(\frac{x}{2^{\theta/(D-1)}}\right)
      + \ell^{|L|^2} 2^{\theta}x
\end{eqnarray*}
To get the exponential decay effect we set $\theta$ so that
$2^{\theta/(D-1)} = 2\ell^{|L|^2}$, which means
$\theta = (D-1)(|L|^2\log_2\ell+1).$
\ep

%!TEX root = main.tex

\section{Conclusions}
\label{sec:conclusions}

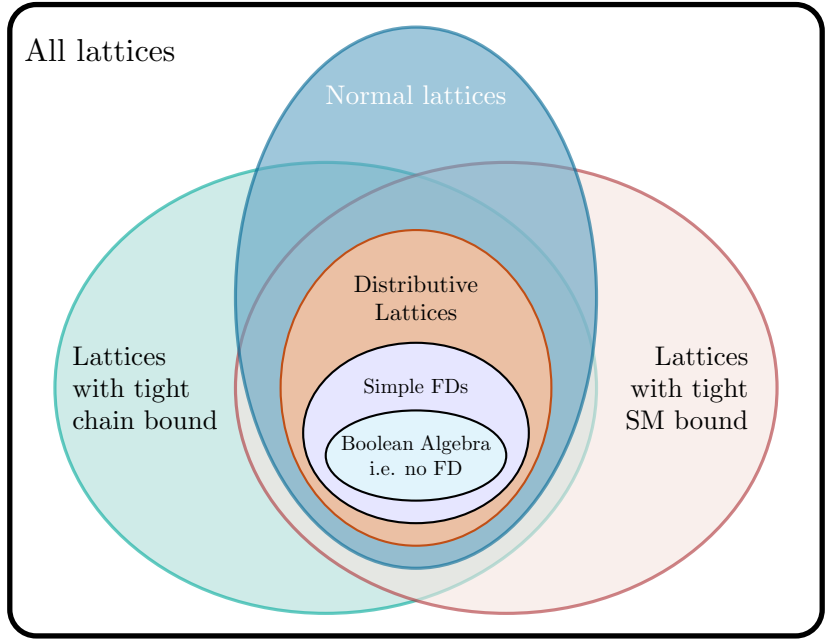
\begin{figure}
\centering
\begin{tikzpicture}[scale=0.6, every node/.style={scale=0.6}]
  \draw [rounded corners=10pt,color=black,line width=2pt,fill=white] (-9,-9) rectangle (9,5);
  \draw [color=JungleGreen,very thick,fill=JungleGreen!30,opacity=0.6] (-2,-3.5) ellipse (6 and 5);
  \draw [color=Maroon,very thick,fill=Maroon!10,opacity=0.6] (2,-3.5) ellipse (6 and 5);
  \draw [color=MidnightBlue,very thick,fill=MidnightBlue!50,opacity=0.6] (0,-1.5) ellipse (4 and 6);
  \draw [color=Bittersweet,thick,fill=Bittersweet!30] (0,-3.5) ellipse (3 and 3.5);
  \draw [color=black,thick,fill=blue!10] (0,-4.5) ellipse (2.5 and 2);
  \draw [color=black,thick,fill=cyan!10] (0,-5) ellipse (2 and 1);
  \node [align=left, scale=2] at (-7,4) {All lattices}; 
  \node [align=center,scale=1.3] at (0,-5) {Boolean Algebra\\i.e. no FD}; 
  \node [align=center,scale=1.3] at (0,-3.5) {Simple FDs};
  \node [align=center,scale=1.5] at (0,-1.5) {Distributive\\Lattices};
  \node [align=left, scale=1.7] at (-6,-3.5) {Lattices\\with tight\\chain bound}; 
  \node [align=right, scale=1.7] at (6,-3.5) {Lattices\\with tight\\SM bound}; 
  \node [align=center,scale=1.7,color=white] at (0,3) {Normal lattices};
\end{tikzpicture}
\caption{A summary of the lattices discussed in this paper}
\label{fig:summary}
\end{figure}

We studied ways to prove worst-case output size bounds, and algorithms meeting
the bounds for join queries with functional dependencies.
A main aim was to design an algorithm running within time bounded by the
entropy-based linear program proposed by Gottlob et al.~\cite{GLVV}.
For this purpose we developed new, lattice theoretic techniques, of
independent interest.  We described several classes of lattices, and
several ways to prove upper or lower bounds on the worst-case query
output, summarized in Fig.~\ref{fig:summary}.
On the algorithmic side, we devised the novel idea of turning a proof of an
inequality into an algorithm. Three proof techniques lead to three different
algorithms with increasing complexity. 
Our main algorithm does meet GLVV bound, up to a poly-$\log$ factor. The
algorithm manages to solve a stronger problem, where input relations have 
prescribed maximum degree bounds, of which functional dependencies and
cardinality bounds are special cases.

\bibliographystyle{acm}
\bibliography{main}

\begin{thebibliography}{10}

\bibitem{DBLP:journals/corr/AbergerNOR15}
{\sc Aberger, C.~R., N{\"{o}}tzli, A., Olukotun, K., and R{\'{e}}, C.}
\newblock Emptyheaded: Boolean algebra based graph processing.
\newblock {\em CoRR abs/1503.02368\/} (2015).

\bibitem{logicblox}
{\sc Aref, M., ten Cate, B., Green, T.~J., Kimelfeld, B., Olteanu, D., Pasalic,
  E., Veldhuizen, T.~L., and Washburn, G.}
\newblock Design and implementation of the logicblox system.
\newblock In {\em Proceedings of the 2015 {ACM} {SIGMOD} International
  Conference on Management of Data, Melbourne, Victoria, Australia, May 31 -
  June 4, 2015\/} (2015), T.~K. Sellis, S.~B. Davidson, and Z.~G. Ives, Eds.,
  {ACM}, pp.~1371--1382.

\bibitem{AGM}
{\sc Atserias, A., Grohe, M., and Marx, D.}
\newblock Size bounds and query plans for relational joins.
\newblock In {\em 49th Annual {IEEE} Symposium on Foundations of Computer
  Science, {FOCS} 2008, October 25-28, 2008, Philadelphia, PA, {USA}\/} (2008),
  pp.~739--748.

\bibitem{DBLP:journals/combinatorica/BalisterB12}
{\sc Balister, P., and Bollob{\'{a}}s, B.}
\newblock Projections, entropy and sumsets.
\newblock {\em Combinatorica 32}, 2 (2012), 125--141.

\bibitem{DBLP:journals/pvldb/BenediktLT15}
{\sc Benedikt, M., Leblay, J., and Tsamoura, E.}
\newblock Querying with access patterns and integrity constraints.
\newblock {\em {PVLDB} 8}, 6 (2015), 690--701.

\bibitem{boyd:vandenberghe:2004}
{\sc Boyd, S., and Vandenberghe, L.}
\newblock {\em Convex Optimization}.
\newblock Cambridge University Press, 2004.

\bibitem{DBLP:conf/sigmod/ChuBS15}
{\sc Chu, S., Balazinska, M., and Suciu, D.}
\newblock From theory to practice: Efficient join query evaluation in a
  parallel database system.
\newblock In {\em Proceedings of the 2015 {ACM} {SIGMOD} International
  Conference on Management of Data, Melbourne, Victoria, Australia, May 31 -
  June 4, 2015\/} (2015), pp.~63--78.

\bibitem{MR859293}
{\sc Chung, F. R.~K., Graham, R.~L., Frankl, P., and Shearer, J.~B.}
\newblock Some intersection theorems for ordered sets and graphs.
\newblock {\em J. Combin. Theory Ser. A 43}, 1 (1986), 23--37.

\bibitem{DBLP:journals/dam/DemetrovicsLM92}
{\sc Demetrovics, J., Libkin, L., and Muchnik, I.~B.}
\newblock Functional dependencies in relational databases: {A} lattice point of
  view.
\newblock {\em Discrete Applied Mathematics 40}, 2 (1992), 155--185.

\bibitem{DBLP:journals/corr/GogaczT15}
{\sc Gogacz, T., and Toru{\'{n}}czyk, S.}
\newblock Entropy bounds for conjunctive queries with functional dependencies.
\newblock {\em CoRR abs/1512.01808\/} (2015).

\bibitem{GLVV}
{\sc Gottlob, G., Lee, S.~T., Valiant, G., and Valiant, P.}
\newblock Size and treewidth bounds for conjunctive queries.
\newblock {\em J. {ACM} 59}, 3 (2012), 16.

\bibitem{grohe2013bounds}
{\sc Grohe, M.}
\newblock Bounds and algorithms for joins via fractional edge covers.
\newblock In {\em In Search of Elegance in the Theory and Practice of
  Computation}. Springer Berlin Heidelberg, 2013, pp.~321--338.

\bibitem{GM06}
{\sc Grohe, M., and Marx, D.}
\newblock Constraint solving via fractional edge covers.
\newblock In {\em SODA\/} (2006), ACM Press, pp.~289--298.

\bibitem{Harremoes2011g}
{\sc Harremo{\"e}s, P.}
\newblock Functional dependences and {B}ayesian networks.
\newblock In {\em Proceedings WITMSE 2011\/} (Helsinki, 2011).

\bibitem{DBLP:journals/corr/JoglekarR15}
{\sc Joglekar, M., and R{\'{e}}, C.}
\newblock It's all a matter of degree: Using degree information to optimize
  multiway joins.
\newblock {\em CoRR abs/1508.01239\/} (2015).

\bibitem{DBLP:journals/actaC/Levene95}
{\sc Levene, M.}
\newblock A lattice view of functional dependencies in incomplete relations.
\newblock {\em Acta Cybern. 12}, 2 (1995), 181--207.

\bibitem{MR3144912}
{\sc Marx, D.}
\newblock Tractable hypergraph properties for constraint satisfaction and
  conjunctive queries.
\newblock {\em J. ACM 60}, 6 (2013), Art. 42, 51.

\bibitem{DBLP:conf/pods/NgoPRR12}
{\sc Ngo, H.~Q., Porat, E., R{\'e}, C., and Rudra, A.}
\newblock Worst-case optimal join algorithms: [extended abstract].
\newblock In {\em PODS\/} (2012), pp.~37--48.

\bibitem{skew}
{\sc Ngo, H.~Q., R{\'{e}}, C., and Rudra, A.}
\newblock Skew strikes back: new developments in the theory of join algorithms.
\newblock {\em {SIGMOD} Record 42}, 4 (2013), 5--16.

\bibitem{radhakrishnan}
{\sc Radhakrishnan, J.}
\newblock 6. entropy and counting.
\newblock {\em Computational Mathematics, Modelling and Algorithms\/} (2003),
  146.

\bibitem{MR1956925}
{\sc Schrijver, A.}
\newblock {\em Combinatorial optimization. {P}olyhedra and efficiency. {V}ol.
  {B}}, vol.~24 of {\em Algorithms and Combinatorics}.
\newblock Springer-Verlag, Berlin, 2003.
\newblock Matroids, trees, stable sets, Chapters 39--69.

\bibitem{MR2868112}
{\sc Stanley, R.~P.}
\newblock {\em Enumerative combinatorics. {V}olume 1}, second~ed., vol.~49 of
  {\em Cambridge Studies in Advanced Mathematics}.
\newblock Cambridge University Press, Cambridge, 2012.

\bibitem{LFTJ}
{\sc Veldhuizen, T.~L.}
\newblock Triejoin: {A} simple, worst-case optimal join algorithm.
\newblock In {\em Proc. 17th International Conference on Database Theory
  (ICDT), Athens, Greece, March 24-28, 2014.\/} (2014), N.~Schweikardt,
  V.~Christophides, and V.~Leroy, Eds., OpenProceedings.org, pp.~96--106.

\bibitem{Yeung:2008:ITN:1457455}
{\sc Yeung, R.~W.}
\newblock {\em Information Theory and Network Coding}, 1~ed.
\newblock Springer Publishing Company, Incorporated, 2008.

\bibitem{DBLP:journals/tit/ZhangY98}
{\sc Zhang, Z., and Yeung, R.~W.}
\newblock On characterization of entropy function via information inequalities.
\newblock {\em {IEEE} Transactions on Information Theory 44}, 4 (1998),
  1440--1452.

\end{thebibliography}
\appendix
%!TEX root = main.tex

\appendix

\section{Additional Material for Sec~\ref{sec:intro}}
\label{sec:appendix:intro}

We show that the query shown in \eqref{eqn:degree-bound} has output size
bounded by $O(N^{3/2},Nd_1,Nd_2)$.
W.L.O.G we may
allow an outgoing edge to be colored with more than one color,
i.e. the FD $xy\rightarrow c_1$ is not required.  This is because
the largest output is obtained when the number of pairs $(x,y)$ in
$R$ is maximized, while $|R|\leq N$, and this happens when each edge
has only one color.
%
%\ds{remember to show, see comment in Latex}
% , because
The bound follows easily because
 \[ h(zx) + h(c_1) \geq h((zxc_1)^+) = h(xyzc_1c_2). \]

\section{Additional Material for Sec~\ref{sec:lattice}}
\label{sec:appendix:lattice}

The following proposition can be found in~\cite{MR1956925}, pp. 774; we
reproduce it here for completeness.
\bprop[Lovasz's monotonization]
Let $\bL=(L,\preceq)$ be a lattice, and $h$ be a non-negative submodular 
function on the lattice. Define $\bar h : L \to \R$ by
\[
\bar h(X) = \begin{cases}
0 & X = \zerohat\\
\min_{Y: X \preceq Y}h(Y) & X \neq \zerohat
\end{cases}
\]
Then, $\bar h$ is a polymatroid where
$\bar h(\hat 1) = h(\hat 1)$ and $\forall X$, $\bar h(X) \leq h(X)$.
\eprop
\bp
We verify the only non-trivial property, that $\bar h$ is sub-modular.
Fix $X \incomp Y$, and let $\bar X = \argmin_{Z : X \preceq Z} h(Z)$ and 
$\bar Y = \argmin_{Z : Y \preceq Z} h(Z)$. Noting that 
$X \vee Y \preceq \bar X \vee \bar Y$ and
$X \wedge Y \preceq \bar X \wedge \bar Y$, we have
\begin{eqnarray*}
\bar h(X)+\bar h(Y) &=& h(\bar X)+h(\bar Y)\\
&\geq& h(\bar X \wedge \bar Y) + h(\bar X \vee \bar Y)\\
&\geq& \bar h(X \wedge Y) + \bar h(X \vee Y).
\end{eqnarray*}
\ep

%\section{Additional Material for Sec~\ref{sec:proof:algorithms}}
%\label{sec:appendix:ca}
%
%\subsection{For Sec.~\ref{sec:csma}}
%\label{sec:appendix:csma}
%%%

\end{document}